\documentclass[a4paper]{article}

\usepackage[left=3cm,right=3cm,top=3cm,bottom=3cm]{geometry}
\usepackage{authblk}

\usepackage{amsmath}
\usepackage{amssymb}
\usepackage{amsfonts}
\usepackage{amsthm}
\usepackage{adjustbox}
\usepackage{caption}
\usepackage{graphicx}
\usepackage{hyperref}

\usepackage{tikz} 
\usetikzlibrary{arrows}
\usetikzlibrary{arrows,automata}
\usetikzlibrary{automata,positioning,arrows}
\usetikzlibrary{shapes,backgrounds}
\usetikzlibrary{matrix}
\usetikzlibrary{chains,fit,shapes}
\usetikzlibrary{calc, decorations.pathreplacing, patterns, graphs, external}
\usetikzlibrary{snakes}

\DeclareMathOperator{\subseq}{\preceq}
\DeclareMathOperator{\nsubseq}{\npreceq}

\DeclareMathOperator{\DFA}{DFA}

\DeclareMathOperator{\states}{\mathsf{states}}

\DeclareMathOperator{\bigO}{O}
\DeclareMathOperator{\smallO}{o}

\DeclareMathOperator{\npclass}{NP}
\DeclareMathOperator{\pclass}{P}
\DeclareMathOperator{\fptclass}{FPT}
\DeclareMathOperator{\wclass}{W}

\DeclareMathOperator{\emptyword}{\varepsilon}

\DeclareMathOperator{\REG}{REG}
\DeclareMathOperator{\poly}{poly}
\DeclareMathOperator{\nz}{\mathsf{nz}}
\DeclareMathOperator{\code}{\mathsf{code}}
\DeclareMathOperator{\codeSketch}{\mathsf{C}}
\DeclareMathOperator{\checkP}{\mathsf{check}}

\DeclareMathOperator{\size}{\mathsf{size}}
\DeclareMathOperator{\levelAncestor}{levelAncestor}
\DeclareMathOperator{\depth}{depth}
\DeclareMathOperator{\LA}{LA}

\DeclareMathOperator{\PatClass}{\mathcal{P}}

\newcommand{\simpleLSConstAlph}[1]{\PatClass^{k}_{\mathsf{s}, \mathsf{LC}}}

\DeclareMathOperator{\OV}{\textsf{OV}}

\DeclareMathOperator{\ETH}{\textsf{ETH}}
\DeclareMathOperator{\SETH}{\textsf{SETH}}
\DeclareMathOperator{\OVH}{\textsf{OVH}}

\newcommand{\lang}[1]{\mathcal{L}(#1)}
\newcommand{\powerset}[1]{\mathcal{P}(#1)}

\newcommand{\alphabet}[1]{\textsf{alph}(#1)}

\newcommand{\gaptuple}{gc}

\newcommand{\gap}[3]{\mathsf{gap}_{#2}(#1, #3)}

\newcommand{\lowerBoundShort}[1]{L^{-}(#1)}
\newcommand{\upperBoundShort}[1]{L^{+}(#1)}

\newcommand{\subseqSet}[2]{SubSeq(#1, #2)}
\newcommand{\parikhK}[2]{\Psi_{#2}(#1)}
\newcommand{\len}[1]{\lvert #1 \rvert}
\newcommand{\myquad}[1][1]{\hspace*{#1em}\ignorespaces}

\newcommand{\matchProb}{\textsc{Match}}
\newcommand{\nuniProb}{\textsc{NUni}}
\newcommand{\ncontProb}{\textsc{NCon}}
\newcommand{\nequiProb}{\textsc{NEqu}}
\newcommand{\uniProb}{\textsc{Uni}}
\newcommand{\contProb}{\textsc{Con}}
\newcommand{\equiProb}{\textsc{Equ}}
\newcommand{\metaNuniProb}{\textsc{MetaNUni}}

\newcommand{\distinctgcsubseq}[3]{\len{#1}_{#2,#3}}

\DeclareMathOperator{\ta}{\mathtt{a}}
\DeclareMathOperator{\tb}{\mathtt{b}}
\DeclareMathOperator{\tc}{\mathtt{c}}
\DeclareMathOperator{\td}{\mathtt{d}}
\DeclareMathOperator{\te}{\mathtt{e}}

\newcommand{\subsequence}[2]{\mathsf{subseq}_{#2}(#1)}

\DeclareMathOperator{\SatProb}{\textsc{CNF-Sat}}

\DeclareMathOperator{\kISProb}{k-\textsc{Independent Set}}

\newcommand{\gapName}{\mathsf{gap}}

\newtheorem{lemma}{Lemma}[section]
\newtheorem{theorem}[lemma]{Theorem}
\newtheorem{proposition}[lemma]{Proposition}
\newtheorem{corollary}[lemma]{Corollary}

\theoremstyle{definition}

\newtheorem{remark}[lemma]{Remark}

\begin{document}

\title{Subsequences With Gap Constraints: Complexity Bounds for Matching and Analysis Problems}

\author[1]{Joel D. Day}
\author[2]{Maria Kosche}
\author[2]{Florin Manea}
\author[3]{Markus L.\ Schmid}

\affil[1]{Loughborough University, UK, \texttt{J.Day@lboro.ac.uk}}
\affil[2]{Computer Science Department, Universität Göttingen, Germany, \texttt{maria.kosche@cs.uni-goettingen.de}, \texttt{florin.manea@cs.informatik.uni-goettingen.de}}
\affil[3]{Humboldt-Universit\"at zu Berlin, Germany, \texttt{MLSchmid@MLSchmid.de}}

\date{}

\maketitle

\begin{abstract}
We consider subsequences with gap constraints, i.\,e., length-$k$ subsequences $p$ that can be embedded into a string $w$ such that the induced gaps (i.\,e., the factors of $w$ between the positions to which $p$ is mapped to) satisfy given gap constraints $\gaptuple = (C_1, C_2, \ldots, C_{k-1})$; we call $p$ a $\gaptuple$-subsequence of~$w$. In the case where the gap constraints $\gaptuple$ are defined by lower and upper length bounds $C_i = (L^-_i, L^+_i) \in \mathbb{N}^2$ and/or regular languages $C_i \in \REG$, we prove tight (conditional on the orthogonal vectors (OV) hypothesis) complexity bounds for checking whether a given $p$ is a $\gaptuple$-subsequence of a string $w$. We also consider the whole set of all $\gaptuple$-subsequences of a string, and investigate the complexity of the universality, equivalence and containment problems for these sets of $\gaptuple$-subsequences. 
\end{abstract}

\section{Introduction}

For a string $v = v_1 v_2 \ldots v_n$, where each $v_i$ is a single symbol from some alphabet $\Sigma$, any string $u = v_{i_1} v_{i_2} \ldots v_{i_k}$ with $k \leq n$ and $1 \leq i_1 \leq i_2 \leq \ldots \leq i_{k} \leq n$ is called a \emph{subsequence} (or \emph{scattered factor} or \emph{subword}) of $v$ (denoted by $u \subseq v$). This is formalised by the \emph{embedding} from the positions of $u$ to the positions of $v$, i.\,e., the increasing mapping $e : \{1, 2, \ldots, k\} \to \{1, 2, \ldots, n\}$ with $j \mapsto i_j$ (we use the notation $u \subseq_e v$ to denote that $u$ is a subsequence of $v$ via embedding~$e$). For example, the string $\ta \tb \ta \tc \tb \tb \ta$ has among its subsequences $\ta \ta \ta$, $\ta \tb \tc \ta$, $\tc \tb \ta$, and $\ta \tb \ta \tb \tb \ta$. With respect to $\ta \ta \ta$, there exists just one embedding, namely $1 \mapsto 1$, $2 \mapsto 3$, and $3 \mapsto 7$, but there are two embeddings for $\tc \tb \ta$. \par
In this paper, we are interested in \emph{subsequences with gap constraints} that can be embedded in such a way that the \emph{gaps} of the embedding, i.\,e., the factors $v_{e(i) + 1} v_{e(i) + 2} \ldots v_{e(i + 1) - 1}$ between the images of the mapping, satisfy certain properties. We begin by discussing why the concept of classical subsequences (i.\,e., without gap constraints) is a central one in computer science, and then we will motivate and describe in detail our approach.\par
The concept of subsequences is employed in many different areas of computer science: in formal languages and logics (e.\,g., piecewise testable languages~\cite{simonPhD,Simon72,KarandikarKS15,CSLKarandikarS,journals/lmcs/KarandikarS19}, or subword order and downward closures~\cite{HalfonSZ17,KuskeZ19,Kuske20,Zetzsche16}), in combinatorics on words~\cite{RigoS15,FreydenbergerGK15,LeroyRS17a,Rigo19,Seki12,Mat04,Salomaa05}, for modelling concurrency~\cite{Riddle1979a, Shaw1978, BussSoltys2014}, in database theory (especially \emph{event stream processing}~\cite{ArtikisEtAl2017,GiatrakosEtAl2020,ZhangEtAl2014}). Moreover, many classical algorithmic problems are based on subsequences, e.\,g., {longest common subsequence} \cite{DBLP:journals/tcs/Baeza-Yates91} or {shortest common supersequence} \cite{Maier:1978}. Note that the longest common subsequence problem, in particular, has recently regained substantial interest in the context of fine-grained complexity (see~\cite{DBLP:conf/fsttcs/BringmannC18,BringmannK18,AbboudEtAl2015,AbboudEtAl2014}).\par
There are two main types of algorithmic problems for subsequences investigated in the literature. Firstly, \emph{matching}: the problem to decide whether a string $u$ is a subsequence of a string $v$, i.\,e., whether $u \subseq v$ (the term matching is motivated by the point of view that $u$ is a pattern that is to be matched with the string $v$). Secondly, the \emph{analysis problems} are concerned with the sets $\subseqSet{k}{v}$ of all length-$k$ subsequences of a given string $v$. More precisely, for given string $v \in \Sigma^+$ and integer $k \in \mathbb{N}$, we want to decide whether $\subseqSet{\gaptuple}{v} = \Sigma^k$ (\emph{universality}), or, for an additional string $v'$, whether $\subseqSet{\gaptuple}{v} \subseteq \subseqSet{\gaptuple}{v'}$ (\emph{containment}) or $\subseqSet{\gaptuple}{v} = \subseqSet{\gaptuple}{v'}$ (\emph{equivalence}). For classical subsequences (as defined above), the matching problem is trivial, while the analysis problems are well-investigated and relatively well-understood. For instance, the equivalence problem was introduced by Imre Simon in his PhD thesis \cite{simonPhD}, and was intensely studied in the combinatorial pattern matching community (see \cite{TCS::Hebrard1991,garelCPM,SimonWords,DBLP:conf/wia/Tronicek02,CrochemoreMT03,KufMFCS} and the references therein), before being optimally solved in 2021 \cite{GawrychowskiEtAl2021}. In this work, we consider these problems with respect to an extended setting of subsequences, which we shall explain and motivate next.\looseness=-1\par
\textbf{Motivation for Our Setting.} In the theoretical literature, problems on subsequences are usually considered in the setting where the embeddings (as witnesses for subsequences) can be arbitrary. This means that any subsequence $u$ of string $v$ is witnessed by a \emph{canonical} embedding $e$ that greedily maps each position $i$ of $u$ to the leftmost occurrence of symbol $u_i$ in the suffix $v_{e(i-1)} v_{e(i-1) + 1} \ldots v_{n}$. For example, $u = \ta \tb \ta$ can be embedded into $v = \ta \tb \ta \tc \tb \tb \ta$ in six different ways, but the canonical embedding maps $u$ to the prefix $v[1..3]$. This makes it often rather simple to deal with subsequences algorithmically: matching can be decided greedily in linear time; the set of all subsequences of a string $v$ can be represented by a \emph{deterministic} automaton of size $\bigO(|v||\Sigma|)$ (which means that the analysis problems can be solved in polynomial time, although much more efficient methods exist in certain cases~\cite{GawrychowskiEtAl2021}).\looseness=-1 \par
For practical scenarios, on the other hand, it seems reasonable to also postulate some properties with respect to the \emph{gaps} that are induced by the embedding.
For example, if we model the scheduling of several threads on a single processor by shuffling several sequences into one string, then a-priori knowledge about the scheduling strategy may tell us that the subsequences describing the single threads will not have huge gaps (any kind of fairness property of the scheduling strategy implies this).
Another example is finding alignments of bio-sequences by computing longest common subsequences. While any common subsequence of two strings can be interpreted as an alignment, it is questionable if this interpretation is still useful if roughly half of the positions of the common subsequence are mapped to the beginning of the strings, while the other half is mapped to the end of the strings, with a huge gap (say thousands of symbols) in between. This situation should rather be seen as two individual alignments. In fact, in this scenario the optimisation goal of finding a \emph{longest} common subsequence, without further constraints, even seems counterproductive, since it may favour alignments that are to a large extent disconnected and are therefore less likely to describe relevant properties. In the context of complex event processing, it might be desirable to describe the situation that between the events of a job $A$ only events associated to a job $B$ appear (e.\,g., due to unknown side-effects this leads to a failure of job $A$). In this case, we are interested in embedding a string as a subsequence such that the gaps only contain symbols from a certain subset of the alphabet (i.\,e., the events associated to job $B$). So, in practice, it makes sense to reason both about the length and the actual content of gaps induced by embeddings.\looseness=-1
\par
The large algorithmic tool box for problems based on subsequences is not always capable of handling the practically relevant scenarios, where we are interested in subsequences that can be embedded not just in \emph{any} way, but in {\em some specific way} that is reasonable for the application scenario. 
We therefore investigate basic problems on subsequences in the setting where the gaps of the subsequences (or rather of the embeddings) have certain constraints. \par
\textbf{Related Work.} Subsequences with various types of gap constraints are considered in different contexts. Not unexpectedly, one of the main areas in which such subsequences were investigated is combinatorial pattern matching with biological motivations, see~\cite{BilleEtAl2012} and the references therein. In \cite{LiW08,LiYWL12}, mining such subsequences is presented as a typical data-mining problem with applications in classification and clustering algorithms. In~\cite{KleestMeissnerEtAl2021}, a query class for event streams is introduced, which is based on subsequences with upper and lower length bounds as gap constraints. 
The longest common subsequence problem has also been extended to the case where the gaps have length constraints (see, e.\,g.,~\cite{IliopoulosEtAl2007} and the references therein). \looseness=-1\par
Coming back to \cite{BilleEtAl2012}, a rather well-researched problem that is related to our setting is that of matching \emph{variable length gap patterns}. In this setting, a pattern, defined as the string 

\begin{equation*}
u_1 (p_1, q_1) u_2 (p_2, q_2) \ldots u_{m-1} (p_{m-1}, q_{m-1}) u_m 
\end{equation*}

with $u_i \in \Sigma^+$, $(p_i, q_i) \in \mathbb{N}^2$ and $0 \leq p_i \leq q_i$, matches a string $w$ if $w = w_0 u_1 w_1 \ldots u_{m-1} w_{m-1} u_m w_m$ with $p_i \leq |w_i| \leq q_i$. For this special pattern matching problem many algorithmic results exist (see~\cite{BilleEtAl2012} and the references therein); moreover, it has also been investigated in more practical papers that provide experimental evaluations of algorithms solving it, see, e.\,g.,~\cite{BaderEtAl2016,CaceresEtAl2020}. The above can be seen as special variants of the matching problem for subsequences with gap-length constraints. \par 
While the works above address mostly patterns with length constraints, the area of string constraint solving (with applications in formal verification, and a strong algorithm engineering component, see \cite{amadini2021survey}) addresses the problem of aligning two strings containing constants (or contiguous sequences of one or more letters) and variables (or gaps). In general (see the aforementioned survey \cite{amadini2021survey} and the references therein), the variables/gaps are subject to conjunctions of pairwise string-equality, length, or regular constraints (see also Appendices \ref{NFAvsDFA} and \ref{ap:RegAndLen}). Moreover, the problem of checking whether factors of words are part of a given regular language were addressed in the context of sliding window algorithms \cite{GanardiHKLM18,GanardiHL16,GanardiHL18,GanardiHL18b,GanardiHLS19} or in the streaming model \cite{BathieS21,DudekGGS22}. \looseness=-1\par
On the other hand, we are not aware of any works that are concerned with (non-trivial) gap-constrained variants of the analysis problems (i.\,e., universality, containment, and equivalence). Let us now formally define the setting considered in this paper.\par
\textbf{Subsequences With Gap Constraints.} Since the gaps induced by an embedding are essentially strings (or words), it seems natural to formalise gap constraints for length-$k$ subsequences by $(k-1)$-tuples of sets of strings (i.\,e., languages) $\gaptuple = (C_1, \ldots, C_{k-1})$, where $C_i \subseteq \Sigma^*$ for every $i \in \{1, \ldots, k-1\}$; we denote $|\gaptuple|=k-1$. A length-$k$ subsequence $u = u_1 u_2 \ldots u_k$ of $v = v_1 v_2 \ldots v_n$ satisfies $\gaptuple$ (i.\,e., it is a $\gaptuple$-subsequence) if $u \subseq_e v$ for an embedding $e$ that satisfies $\gaptuple$ in the sense that, for every $i \in [k - 1]$, $v_{e(i)+1} \ldots v_{e(i+1)-1} \in C_i$. By $\subseqSet{\gaptuple}{v}$ we denote the set of all $\gaptuple$-subsequences of $v$. 
In this setting, we consider:
\begin{itemize}
\item the \emph{matching problem} $\matchProb$: decide, for given strings $p$, $w$, and gap constraints $\gaptuple$ with $|\gaptuple|=|p|-1$, whether $p$ is a $\gaptuple$-subsequence of $w$ (i.\,e., whether $u \in \subseqSet{\gaptuple}{v}$);
\item the \emph{universality problem} $\uniProb$: decide, for given string $w$ and gap constraints $\gaptuple$ with $|\gaptuple|=k-1$ , whether $\subseqSet{\gaptuple}{w} = \Sigma^{k}$;
\item the \emph{equivalence problem} $\equiProb$ (respectively,  the \emph{containment problem} $\contProb$): decide, for given strings $w, w'$, and gap constraints $\gaptuple$, whether $\subseqSet{\gaptuple}{w} = \subseqSet{\gaptuple}{w'}$ (respectively, $\subseqSet{\gaptuple}{w} \subseteq \subseqSet{\gaptuple}{w'}$).
\end{itemize}
Our formalisation of gap constraints is as general as possible. In order to obtain meaningful results we focus on \emph{regular constraints}, where each $C_i$ is a regular language, and on \emph{length constraints}, where each $C_i$ has the form $\{v \in \Sigma^* \mid  \lowerBoundShort{i} \leq |v| \leq \upperBoundShort{i}\}$ with $\lowerBoundShort{i}, \upperBoundShort{i} \in \mathbb{N}\cup\{0,+\infty\}$ and is represented as the pair $(\lowerBoundShort{i}, \upperBoundShort{i})$. We also consider {\em conjunctions} $(C_i, (\lowerBoundShort{i}, \upperBoundShort{i}))$ {\em of regular and length constraints}, i.\,e., the gap must be from $C_i$ \emph{and} of length between $\lowerBoundShort{i}$ and $\upperBoundShort{i}$ (note that simply ``pushing'' the length constraint into the regular language $C_i$ would increase $C_i$'s representation by a factor $\upperBoundShort{i}$, which is exponential in $\upperBoundShort{i}$'s binary representation). These constraints cover the existing cases in the literature.\par

\textbf{Our Contribution.} 
We provide a comprehensive picture of the computational complexity of both the matching and the analysis problems, proving tight upper and lower bounds for them, with a focus on the latter. \looseness=-1\par
With respect to matching, we show that we can check whether $u$ is a $\gaptuple$-subsequence of $v$ in \emph{rectangular} time $O(|v| |\gaptuple|)$, where, if each $C_i$ is the conjunction of a regular constraint and a length constraint, $|\gaptuple|$ is the number of states of the $\DFA$s that represent the regular constraints. In the absence of regular constraints (so, for length constraints only), such rectangular upper bounds are already reported in the literature (see~\cite{IliopoulosEtAl2007}). Moreover, the case when length constraints are absent (so, we have regular constraints only) is rather straightforward. Our algorithm dealing with the case of conjunctions of regular and length constraints requires, however, a non-trivial extension of the existing approaches. Nevertheless, our main contribution in this area is that we can also prove a conditional lower bound that essentially states that these running times of those algorithms cannot be improved unless the \emph{orthogonal vectors hypothesis} fails. More precisely, adding length or regular constraints to subsequences changes the matching problem from a trivial problem to a problem with provably rectangular complexity. Additionally, this proves also a conditional lower bound for matching variable length gap patterns (mentioned above), for which many upper bounds, but no matching lower bound were known before. It is also worth noting that the lower bound holds for the case of a constant alphabet and constant length constraints.\par
With respect to the problems of universality, equivalence, and containment, we show strong intractability results for both the cases of length constraints and of regular constraints. More precisely, these problems are $\npclass$-complete even for a fixed binary alphabet and for small, constant length (or regular) constraints (note that the problems are trivial for a unary alphabet). Moreover, for any fixed constant alphabet, the problems can be solved by brute-force algorithms in exponential time $2^{\bigO(k)} |\gaptuple| \ell$ (recall that $k$ is the length of subsequences; $\ell$ is the maximum length of the input strings), and we can show that for alphabets of size at least~$3$, the exponent can neither be lowered to any $\smallO(k)$ (unless the \emph{exponential time hypothesis} fails), nor to $k(1-\epsilon)$ for any $\epsilon \geq 1$ (unless the \emph{strong exponential time hypothesis} fails), and these lower bounds even hold for small constant length constraints. 
If we parameterise by both $|\Sigma|$ and $k$, then the brute-force algorithm is a trivial fpt-algorithm. However, we can exclude fpt-running times for the cases where we parameterise by only $|\Sigma|$, or by only $k$ (based on the assumptions $\pclass \neq \npclass$ and $\fptclass \neq \wclass[1]$, respectively). Note that for classical subsequences all these problems can be easily solved in polynomial time, so our results emphasise the fundamentally different nature of constrained subsequences.\par
Additionally (and only in Appendix \ref{sec:Special}, due to space constraints), we investigate some natural extensions of both the matching problem, involving gap-lengths equality, and the analysis problems, which involve counting the number of occurrences of subsequences. \looseness=-1

\section{Preliminaries}

Let $\mathbb{N} = \{1, 2, \ldots\}$ and $[n] = \{1, \ldots, n\}$ for $n \in \mathbb{N}$. By $\powerset{S}$, we denote the power set of a set~$S$.\looseness=-1 \par
For a finite alphabet $\Sigma$, $\Sigma^+$ denotes the set of non-empty words over $\Sigma$ and $\Sigma^* = \Sigma^+ \cup \{\emptyword\}$ (where $\emptyword$ is the empty word). For a word $w \in \Sigma^*$, $|w|$ denotes its length (in particular, $|\emptyword| = 0$); for every $b \in \Sigma$, $|w|_{b}$ denotes the number of occurrences of $b$ in $w$; we set $w^1 = w$ and $w^k = w w^{k-1}$ for every $k \geq 2$. For a string $w = w_1 w_2 \ldots w_n$ with $w_i \in \Sigma$ for every $i \in [n]$, and for every $i, j \in [|w|]$ with $i \leq j$, we define $w[i..j] = w_i w_{i+1} \ldots w_j$; moreover, we use $w[i]$ as shorthand for $w[i..i]$. For any string $w \in \Sigma^*$, we define $\alphabet{w} = \{b \in \Sigma \mid |w|_b \geq 1\}$. A \emph{factor} of a string $w \in \Sigma^*$ is a string $v \in \Sigma^*$ such that $w = u v u'$ for $u, u' \in \Sigma^*$; if $u = \emptyword$, then $v$ is called a \emph{prefix} of $w$, and if $u' = \emptyword$, then $v$ is called a \emph{suffix} of $w$.\par
By $\REG$, we denote the class of regular languages (see~\cite{HopcroftUllman} for more details). For the considered algorithmic problems we use as computational model the standard unit-cost RAM with logarithmic word size, with inputs over integer alphabets (see Appendix \ref{sec:compModel}).

\textbf{Hypotheses.} We now recall some basic computational problems and respective algorithmic hypotheses. We shall use these hypotheses to obtain our conditional lower bounds. \par
The problem $\SatProb$ gets as input a Boolean formula $F$ in conjunctive normal form as a set of clauses $F = \{c_1, c_2, \ldots, c_m\}$ over a set of variables $V = \{v_1, v_2, \ldots, v_n\}$, i.\,e., for every $i \in [m]$, we have $c_i \subseteq \{v_1, \neg v_1, \ldots, v_n, \neg v_n\}$. The question is whether $F$ is satisfiable.
By $k$-$\SatProb$, we denote the variant where $|c_i| \leq k$ for every $i \in [m]$.\par
The \emph{Orthogonal Vectors problem} ($\OV$ for short) is defined as follows: Given sets $A, B$ each containing $n$ Boolean-vectors of dimension~$d$, check whether there are vectors $\vec{a} \in A$ and $\vec{b} \in B$ that are orthogonal, i.\,e., $\vec{a}[i] \cdot \vec{b}[i] = 0$ for every $i \in [d]$.\par
We shall use the following algorithmic hypotheses based on $\SatProb$ and $\OV$ that are common for obtaining conditional lower bounds in fine-grained complexity (see the literature mentioned below for further details). In the following, $\poly$ is any fixed polynomial function.
\begin{itemize}
\item \emph{Exponential Time Hypothesis}  ($\ETH$)~\cite{ImpagliazzoEtAl2001, LokshtanovEtAl2011}: $3$-$\SatProb$ cannot be solved in time $2^{\smallO(n)} \poly(n + m)$.
\item \emph{Strong Exponential Time Hypothesis} ($\SETH$)~\cite{ImpagliazzoPaturi2001, Williams2015}: For every $\epsilon > 0$ there exists a $k$ such that $k$-$\SatProb$ cannot be decided in $O(2^{n(1-\epsilon)} \poly(n))$. 
\item \noindent \emph{Orthogonal Vectors Hypothesis} ($\OVH$)~\cite{Bringmann2014, Bringmann2019, Williams2015}: For every $\epsilon > 0$ there is no algorithm solving OV in time $\bigO(n^{2 - \epsilon} \poly(d))$. 
\end{itemize}
\label{sec:GapConstraints}
\hspace{\parindent}{\bf Subsequences With Gap Constraints.} 
We now define subsequences with gap constraints (see also the introduction). In the following, let $\Sigma$ be a finite \emph{alphabet}. Recall that for a string $w$, an embedding is a function $e : [k] \to [|w|]$ such that $i < j$ implies $e(i) < e(j)$ for all $i, j \in [k]$, and it \emph{induces the subsequence $\subsequence{w}{e} = w[e(1)]w[e(2)]\ldots w[e(k)]$ of $w$}. 
For every $j \in [k-1]$, the \emph{$j^{\text{th}}$ gap of $w$ induced by $e$} is the string $\gap{w}{e}{j} = w[e(j)+1..e(j+1)-1]$. We say that $e$ is the embedding of $\subsequence{w}{e}$ in $w$. \par
An \emph{$\ell$-tuple of gap constraints} is a tuple $\gaptuple = (C_1, C_2, \ldots, C_{\ell})$ with $C_i \subseteq \Sigma^*$ for every $i \in [\ell]$. For convenience, we set $\gaptuple[i] = C_i$ for every $i \in [\ell]$. We say that an embedding $e$ \emph{satisfies a $(k-1)$-tuple of gap constraints $\gaptuple$ with respect to a string $w$} if it has the form $e: [k] \to [|w|]$, and, for every $i \in [k - 1]$, $\gap{w}{e}{i} \in C_i$. Moreover, for a $(k-1)$-tuple $\gaptuple$ of gap constraints, the set $\subseqSet{\gaptuple}{w}$ contains all subsequences of $w$ induced by embeddings that satisfy $\gaptuple$, i.\,e., $\subseqSet{\gaptuple}{w} = \{\subsequence{w}{e} \mid e \text{ is an embedding that satisfies } \gaptuple \text{ w.\,r.\,t. $w$}\}$. The elements of $\subseqSet{\gaptuple}{w}$ are also called the \emph{$\gaptuple$-subsequences of $w$}. Note that tuples of gap constraints do not have constraints for the prefix $w[1..e(1)]$ or suffix $w[e(k)..|w|]$. However, our formalism can model this case too (for details, see Appendix \ref{prefixSuffix}). For a $(|u|-1)$-tuple $\gaptuple$ of gap constraints, we write $u \subseq_{\gaptuple} v$ to denote that $u \subseq_e v$ for some embedding $e : [|u|] \to [|v|]$ that satisfies $\gaptuple$ with respect to $v$, i.\,e., $u \subseq_{\gaptuple} v$ means that $u$ is a $\gaptuple$-subsequence of $v$. We note that for tuples of gap constraints $\gaptuple = (C_1, C_2, \ldots, C_{k-1})$ with $C_i = \Sigma^*$ for every $i \in [k-1]$, the set $\subseqSet{\gaptuple}{w}$ is just the set of all length-$k$ subsequences of $w$. 

\textbf{Special Types of Gap Constraints.} We now define the types of gap constraints that are relevant for our work.
We say that the gap constraints $\gaptuple = (C_1,\ldots, C_{k-1})$ are 
\begin{itemize}
\item \emph{regular constraints} if $C_i \in \REG$ for every $i \in [k-1]$. For every $i \in [k-1]$, we represent the regular constraint $C_i$ by a deterministic finite automaton (for short, $\DFA$) $A_i$ accepting it. See Appendix \ref{NFAvsDFA} for a discussion on the choice of DFAs to represent regular constraints. 
\item \emph{length constraints} if, for every $i \in [k-1]$, there are $\lowerBoundShort{i}, \upperBoundShort{i} \in \mathbb{N}\cup\{0,+\infty\}$ with $\lowerBoundShort{i} \leq \upperBoundShort{i}$, such that $C_i = \{v \in \Sigma^* \mid  \lowerBoundShort{i} \leq |v| \leq \upperBoundShort{i}\}$. We represent length constraints succinctly by pairs of numbers $(\lowerBoundShort{i}, \upperBoundShort{i})$, $i \in [k-1]$, in binary encoding. 
\item \emph{reg-len constraints} if, for every $i \in [k-1]$, $C_i$ is the conjunction of a regular constraint $C'_i$ and a length constraint $(\lowerBoundShort{i}, \upperBoundShort{i})$, i.\,e., $C_i = C'_i \cap \{v \in \Sigma^* \mid  \lowerBoundShort{i} \leq |v| \leq \upperBoundShort{i}\}$. We represent such constraints by $((\lowerBoundShort{i}, \upperBoundShort{i}),A'_i)$, where $A'_i$ is a $\DFA$ accepting $C'_i$.
\end{itemize}
\hspace{\parindent}A gap constraint $C_i$ is a {\em zero-gap} if and only if $C_i=\{\emptyword\}$. Let $\nz(\gaptuple)$ be the number of non-zero-gaps of $\gaptuple$ (that is, the number of positions $i$ such that $C_i \neq \{\emptyword\}$). 
For a tuple of regular or reg-len gap constraints $\gaptuple$, let $\size(\gaptuple)$ be the size of the overall representation of the respective constraints (total size of the automata defining the constraints) and let $\states(\gaptuple)$ be the total number of states of the $\DFA$s $A_i$, for $i \in [k-1]$, corresponding to the non-zero gaps of $\gaptuple$.\looseness=-1

Clearly, length constraints are the simplest type of gap constraints considered above. In particular, length constraints, and therefore reg-len constraints, can also be seen as a particular case of regular constraints. However, transforming length or reg-len constraints into a single automaton may cause an exponential size increase.

\textbf{Problems for Subsequences With Gap Constraints.} 
In this paper, we investigate the matching problem $\matchProb$ and the analysis problems $\uniProb$, $\contProb$, and $\equiProb$ (see definitions in the introduction). For simplicity, the pairs $(p, \gaptuple)$, which play the role of the patterns in $\matchProb$, will be called \emph{gap-constrained sequences}, or simply \emph{gapped sequences} for short. By $\matchProb_{\Sigma}$, we denote the problem variant where all instances are over the fixed alphabet $\Sigma$; for some class $\mathcal{C}$ of gap constraints, we use ``$\matchProb$ with $\mathcal{C}$-constraints'' to refer to the variant where the constraints are from $\mathcal{C}$. We use analogous notations for the analysis problems.\par

If $\gaptuple = (\Sigma^*, \Sigma^*, \ldots, \Sigma^*)$, then $\matchProb$ boils down to the simple task of checking whether a given string is a subsequence of another string. The equivalence problem for such trivial gap constraints, on the other hand, boils down to the well-known problem of deciding the Simon congruence for two strings (see the discussion in the introduction).
Our setting naturally models many other classical problems; some are discussed in Appendix~\ref{relatedProbs}. 
Finally, even though our framework allows arbitrary gap constraints, we will stick to the specific natural and relevant types of constraints defined above (i.\,e., length, regular, reg-len constraints).

\section{Matching Gapped Subsequences}\label{sec:matching}

This section contains two main results. Firstly, we show that $\matchProb$ with reg-len constraints can be solved in $O(|w|\states(\gaptuple)+\size(\gaptuple))$ time, which implies also rectangular upper bounds for $\matchProb$ with either length or regular constraints. Secondly, we show that, assuming $\OVH$ holds, there are no algorithms solving any of these problems polynomially faster. \looseness=-1

Note that, when dealing with length constraints, a constraint $(\lowerBoundShort{i}, \upperBoundShort{i})$ is equivalent to the regular language $C_i=\{x\in \Sigma^* \mid \lowerBoundShort{i} \leq |x| \leq \upperBoundShort{i} \}$, which is accepted by a DFA with $\Theta(\upperBoundShort{i})$ states. So, we could also interpret a tuple $\gaptuple$ of reg-len constraints as a tuple of regular constraints only, by considering in each component of $\gaptuple$ the intersection of the regular constraint with the regular language defined by the length constraints. However, this would lead to a growth in the number of states needed to model $\gaptuple$, and, as we will see in the following, to a less efficient algorithm for $\matchProb$. In this setting, we state our first main result. The full proof is given in Appendix \ref{sec:MatchingAlgo}. To emphasise the merits of our approach, we overview in Appendix \ref{sec:MatchingAlgo} several simpler approaches and their complexity (and shortcomings).

\begin{theorem}\label{constantPatternsRegularLength}
$\matchProb$ with reg-len constraints can be solved in $O(|w|\states(\gaptuple) + \size(\gaptuple))$~time.\looseness=-1
\end{theorem}

\begin{proof}[Proof Sketch]
Assume $|w|=n$, $|p|=m$, and $\gaptuple =((\lowerBoundShort{1},$ $\upperBoundShort{1}),A_{1}), \ldots, (\lowerBoundShort{m-1},$ $\upperBoundShort{m-1}),A_{m-1}))$, where $A_i=(Q_i,q_{0,i},F_i,\delta_i)$ are DFAs defining the regular constraints and $(\lowerBoundShort{i}, \upperBoundShort{i})$ are pairs of numbers defining the length constraints. 
Let $i_1, \ldots, i_{k-1}\in [m-1]$ be such that $C_i\neq \{\emptyword\}$ (i.\,e., $C_i$ is a non-zero constraint of $\gaptuple$), for all $i\in \{i_1, \ldots, i_{k-1}\}$, and $C_i=\{\emptyword\}$, for all $i\notin \{i_1, \ldots, i_{k-1}\}$.
Clearly, $\states(\gaptuple)=\sum_{j=1}^{k-1}|Q_{i_j}|\geq \nz(\gaptuple)$. With $i_0=0$ and $i_k=m$, we compute the words $p_j=p[(i_{j-1}+1)..i_j]$, for $j\in [k]$, and we construct in linear time longest common extension data structures (see \cite{dinklage_et_al:LIPIcs:2020:12905} and the references therein) for the word $x=wp$, allowing us to check in constant time whether $w[i+1..i+|p_j|]=p_j$, for $i\leq n$. 

After this preprocessing part, the main part of our algorithm consists in a dynamic programming approach. We compute a two-dimensional $n\times k$ array $D[\cdot][\cdot]$, where $D[i][\ell]=1$ if and only if $p[1..|p_1\cdots p_\ell|]$ can be embedded in $w[1..i]$ and this embedding satisfies the first $\ell-1$ non-zero constraints of $\gaptuple$ and maps $p_\ell$ to the suffix of length $|p_\ell|$ of $w[1..i]$. Otherwise, $D[i][\ell]=0$. To start the computation of $D$, we initialize all the elements of $D$ with $0$. We then set $D[i][1]=1$ if and only if $w[1..i]$ ends with $p_1$, i.\,e., $w[i-|p_1|+1..i]=p_1$. 

Further, assume that, for some $t\in [k-1]$, we have computed $D[\cdot][\ell]$, for all $\ell\leq t$, and we want to compute $D[\cdot][t+1]$. This is the most involved and deep part of our algorithm and its main component is computing an array $f_{t+1}[\cdot]$, with $n$ elements, such that $f_{t+1}[i]=1$ iff there exists a position $j$ for which $D[j][t]=1$, $w[j+1..i]\in L(A_{t})$, and $\lowerBoundShort{t} \leq |w[j+1..i]|\leq \upperBoundShort{t}$. A full description of this part of the algorithm is given in Appendix \ref{sec:MatchingAlgo}; here we just sketch it.\looseness=-1

We first collect in a list $L_{t+1}=j_1<\ldots <j_r$ (increasingly sorted) all the positions $j$ of $w$ with $D[j][t]=1$.
Then, we compute a graph $G_{t+1}$ that has nodes of the form $(i,q)$, with $i\in [n]$ and $q\in Q_t$, and consists of the union, over $j\in L_{t+1}$, of the (not necessarily disjoint) paths 
$[(j,q_{0,t}), (j+1,q^j_1),\ldots, (n,q^j_{n-j})]$, where $\delta_t(q_{0,t},w[j+1])=q^j_1$ and $\delta_t(q^j_r,w[j+r+1])=q^j_{r+1}$, for all $r\in [n-j-1]$. Intuitively, such a path records the trace of the computation of $A_t$ on the input $w[j+1..n]$. For efficiency, these paths (and, therefore, the graph $G_{t+1}$) can be simultaneously constructed to avoid redundant computations. An important observation is that if two such paths intersect, then they are identical after their first common node; this is, indeed, true because $A_t$ is a deterministic finite automaton. Consequently, $G_{t+1}$ is a collection of disjoint trees $T_1, T_2, \ldots, T_z$. As there are no edges between any pair of nodes $(n,q)$ and $(n,q')$, with $q,q'\in Q_t$, each such tree $T_i$ can be seen as a rooted tree, whose root is its single node of the form $(n,q)$ and whose leaves are some of the nodes $(j,q_{0,t})$, with $j\in L_{t+1}$. \looseness=-1

Then, based on a series of efficient data structures and further insights, we efficiently mark, for each tree $T_i$ and for each leaf $(j,q_{0,t})$ of $T_i$, all the ancestors $(d,q)$ of $(j,q_{0,t})$ such that $\lowerBoundShort{t}\leq |w[j+1..d]|=d-j\leq \upperBoundShort{t}$. Once we have completed the marking for tree $T_i$, a node $(j,q)$ is marked if and only if there exists a path ${\mathcal P}$ of length $\ell$, with $\lowerBoundShort{t}\leq \ell\leq \upperBoundShort{t}$, which connects a leaf $(j',q_{0,t})$ of $T_i$ to $(j,q)$. Or, in other words, $\delta_t(q_{0,t}, w[j'+1..j])=q$. 
The trees $T_i$, with $i\in [z]$, are computed in $O(n|Q_t|)$ time, while the marking takes $O(\sum_{i=1}^p|T_i|)$ time.\looseness=-1

Finally, we simply set, for $i$ from $1$ to $n$, $f_{t+1}[i]=1$ if and only if there exists a state $q\in F_t$ such that the node $(i,q)$ is marked. This means that $f_{t+1}[i]=1$ if and only if there exists a word $w[j+1..i]$ of length $\ell$, with $\lowerBoundShort{t}\leq \ell\leq \upperBoundShort{t}$, such that $j\in L_{t+1}$ and $\delta_t(q_{0,t}, w[j+1..i])$ is a final state (i.\,e., $w[j+1..i]\in L(A_{t})$). 

Coming now back to the computation of the elements of $D$, we set $D[i][t+1]=1$ if and only if $w[i-|p_{t+1}|+1..i]=p_{t+1}$ and $f_{t+1}[i-|p_{t+1}|]=1$. Clearly, $D[\cdot][t+1]$ is correctly computed. 

After $D$ is computed, we decide that $p \subseq_{\gaptuple} w$ if and only if there exists $j$ with $D[j][k]=1$. The whole process can be implemented in $O(|w|\states(\gaptuple) + \size(\gaptuple))$ time.
\end{proof}

The next results are now immediate. Note that for these particular cases (but, to the best of our knowledge, not for their conjunction,  covered in Theorem \ref{constantPatternsRegularLength}) simpler algorithms exist. \looseness=-1

\begin{corollary}\label{constantPatternsRL}
(1). $\matchProb$ with length constraints can be solved in $O(|w|\nz(\gaptuple))$ time.\\ (2). $\matchProb$ with regular constraints can be solved in $O(|w|\states(\gaptuple) + \size(\gaptuple))$ time.
\end{corollary}

It is worth noting that the matching problem can be solved in $O(|w|)$ time when $\gaptuple$ only defines constraints that are $\{\emptyword\}$ or $\Sigma^*$, which covers, e.\,g., the cases of subsequence matching or string matching. 
In particular, the greedy strategy used for matching regular patterns with variables (see, e.\,g., \cite{FernauMMS20}) can be easily adapted to solve $\matchProb$ with length constraints in linear time, when the upper bounds on each gap are trivial (i.\,e., they are all greater or equal to the length of the input word). So, as far as length constraints are concerned, it seems that non-trivial upper bounds lead to an increase in the difficulty of the $\matchProb$ problem; a particularly efficient approach for subsequences with general length constraints is given in \cite{BilleEtAl2012}, but, in the worst case, it still has rectangular complexity. However, even when non-trivial length upper bounds are used, there are still some simpler particular cases. For instance, when working with {\em strings with don't cares} (or partial words), where each gap has a fixed length (i.\,e., the lower and upper bounds are the same), $\matchProb$ can be solved in time $O(|w| \log |p|)$ \cite{CliffordC07}. \looseness=-1\par
A gapped sequence $(p, \gaptuple)$ with reg-len constraints can be represented as a classical regular expression $r_{(p, \gaptuple)}$, so $\matchProb$ can be solved by a textbook algorithm in $O(|w||r_{(p, \gaptuple)}|)$~\cite{Thompson68}, which is optimal w.\,r.\,t. polynomial speed-ups, conditional on $\OVH$~\cite{BackursIndyk2016}. However, including the string $p$ and the length constraints in the regular expression might, once more, lead to a slower algorithm compared to our direct approach, as $|r_{p,\gaptuple}|$ may be much larger than $\states(\gaptuple)$. \looseness=-1\par
To summarise, at an intuitive level, we could say that as long as we have non-trivial length or regular constraints, $\matchProb$ seems to become more difficult than its counterpart for classical subsequences. This intuitive remark is confirmed by our second main result.

\begin{theorem}\label{lowerBoundLength}
$\matchProb$ with length constraints cannot be solved in $\mathcal{O}(|w|^h \nz(\gaptuple)^g)$ time 
with $h+g= 2-\epsilon$ for some $\epsilon>0$, unless $\OVH$ fails. This holds even if $|\Sigma| = 4$ and all length constraints are $(0, \ell)$ with $\ell \leq 6$.
\end{theorem}

\begin{proof}[Proof Sketch]
Let $A=\{\vec{a}_1, \ldots, \vec{a}_n\}$ and $B=\{\vec{b}_1, \ldots, \vec{b}_n\}$, with $A,B \subset \{0,1\}^d$ be an $\OV$-instance. We transform $A$ into a string $w \in \Sigma^* = \{0,1,\#,@\}^*$ and $B$ into a string $p \in \Sigma^*$ and a $(|p|-1)$-tuple $\gaptuple$ of length constraints. For convenience, we represent the gapped sequence $(p, \gaptuple)$ with $p=p[1]\cdots p[m]$ by  writing the length constraints in between the symbols, i.\,e., $p[1] \stackrel{\gaptuple[1]}{\leftrightarrow} p[2] \stackrel{\gaptuple[2]}{\leftrightarrow} \cdots \stackrel{\gaptuple[m-1]}{\leftrightarrow} p[m]$, and we omit $\stackrel{\gaptuple[i]}{\leftrightarrow}$ if $\gaptuple[i]=(0,0)$. For example, if $p=abab$ and $\gaptuple[1]=(0,0)$, $\gaptuple[2]=(1,5)$, and $\gaptuple[3]=(0,6)$, we use the notation $ab \stackrel{(1,5)}{\leftrightarrow}  a\stackrel{\leq 6}{\leftrightarrow} b$. \par
Let $\vec{a}_i=(a_i^1,\ldots, a_i^d)$ and $\vec{b}_i=(b_i^1,\ldots, b_i^d)$, for all $i\in [n]$. We shall represent the vectors from $A$ and $B$ by different encodings $\codeSketch_a(\cdot)$ and $\codeSketch_b(\cdot)$, respectively. The $0$ and $1$ entries in the $A$-vectors are encoded by $\codeSketch_a (0)=010$ and $\codeSketch_a(1)=100$, and the $0$ and $1$ entries in the $B$-vectors are encoded by $\codeSketch_b(0)=10$ and $\codeSketch_b(1)=01$. We note that for every $x, y \in \{0, 1\}$, $\codeSketch_b(x)$ is a factor of $\codeSketch_a(y)$ if and only if $x \cdot y = 0$. This means that the orthogonality of $\vec{a}_i$ and $\vec{b}_{i'}$ is characterised by the situation that, for every $j \in [d]$, $\codeSketch_b(b_{i'}^j)$ is a factor of $\codeSketch_a(a_{i}^j)$.\par
We represent each bit $a_i^j$ of $\vec{a}_i \in A$ as the string $\#\codeSketch_a(0) \# \#\codeSketch_a(a_i^j) \# \#\codeSketch_a(0) \#$, and the whole vector $\vec{a}_i$ as the concatenation
$\codeSketch_a(\vec{a}_i) = \prod_{j=1}^d ([_1\#\codeSketch_a(0) \#]_1 [_2\#\codeSketch_a(a_i^j) \#]_2[_3 \#\codeSketch_a(0) \#]_3)\,,$
where the brackets $[_1 \ldots ]_1$, $[_2 \ldots ]_2$, $[_3 \ldots ]_3$ are not actual symbols of the gadget, but serve the only purpose to illustrate that $\codeSketch_a (\vec{a}_i)$ has three individual \emph{tracks}, where track $1$ and $3$ correspond to $d$ occurrences of $\#\codeSketch_a(0) \#$ (representing the all-$0$ vector), while track $2$ represents the actual vector $\vec{a}_i$. These three tracks play a central role in the correctness of the reduction. \looseness=-1\par
For $i \in [n]$, every vector $\vec{b}_i \in B$ is also represented by listing all bit encodings $\codeSketch_b(b_i^j)$, but in a slightly different way and, most importantly, as a gapped sequence (in the notation defined above):
$(\codeSketch_b(\vec{b}_i), \gaptuple_i) = \left(\prod_{j=1}^{d-1} (\# \stackrel{\leq 1}{\leftrightarrow} \codeSketch_b(b_i^j) \stackrel{\leq 1}{\leftrightarrow} \# \# \stackrel{\leq 3}{\leftrightarrow} \# \# \stackrel{\leq 3}{\leftrightarrow} \#)\right)\# \stackrel{\leq 1}{\leftrightarrow} \codeSketch_b(b_i^d) \stackrel{\leq 1}{\leftrightarrow} \#.$

It can be shown (see Appendix \ref{sec:MatchingLowerBound}) that if $\codeSketch_b(\vec{b}_i) \subseq_{e} \codeSketch_a(\vec{a}_{\ell})$ and $e$ satisfies $\gaptuple_i$, then the embedding $e$ maps each $\codeSketch_b(b_i^j)$ to the $\codeSketch_a(0)$ of $\codeSketch_a(\vec{a}_{\ell})$'s first track, or each $\codeSketch_b(b_i^j)$ to the $\codeSketch_a(a_{\ell}^j)$ of $\codeSketch_a(\vec{a}_{\ell})$'s second track, or each $\codeSketch_b(b_i^j)$ to the $\codeSketch_a(0)$ of $\codeSketch_a(\vec{a}_{\ell})$'s third track. More precisely, due to how we use the symbols $\#$, the factor $\codeSketch_b(b_i^1)$ must be mapped to $[_1\#\codeSketch_a(0) \#]_1$ or to $[_2\#\codeSketch_a(a_i^1) \#]_2$ or to $[_3 \#\codeSketch_a(0) \#]_3$. Since we have $4$ occurrences of $\#$ between each $\codeSketch_b(b_i^{j})$ and $\codeSketch_b(b_i^{j + 1})$, and between two consecutive parts of the same track in $\codeSketch_a(\vec{a}_i)$, all the following factors $\codeSketch_b(b_i^2), \codeSketch_b(b_i^3), \ldots$ must be mapped to the same track $\codeSketch_b(b_i^1)$ is mapped to. This is illustrated in Figure~\ref{fig:fittinggadgets}. Based on these considerations, it is clear that $\codeSketch_b(\vec{b}_i) \subseq_{e} \codeSketch_a(\vec{a}_{\ell})$ with $e$ mapping $\codeSketch_b(\vec{b}_i)$ to $\codeSketch_a(\vec{a}_{\ell})$'s second track is possible if and only if $\vec{a}_{\ell}$ and $\vec{b}_{i}$ are orthogonal.\par

\begin{figure}[h!]
\vspace*{-5pt}
    \centering
  \includegraphics[width=14cm]{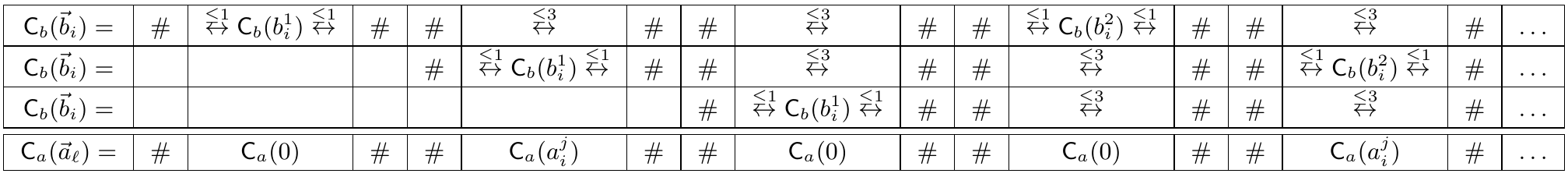}
\caption{Possible embeddings of $\codeSketch_b(\vec{b}_i)$ in $\codeSketch_a(\vec{a}_\ell)$, selecting its first, second, or third track.}
    \label{fig:fittinggadgets}
\end{figure}

The remaining challenge is to combine the gadgets $\codeSketch_a(\vec{a}_i)$ into a string $w$, and the gadgets $(\codeSketch_b(\vec{b}_i), \gaptuple_i)$ into a gapped sequence $(p, \gaptuple)$, such that $p \subseq_{e} w$ for an embedding $e$ satisfying $\gaptuple$ if and only if $e$ is such that every $(\codeSketch_b(\vec{b}_i), \gaptuple_i)$ is mapped to some $\codeSketch_a(\vec{a}_{\ell})$, and there is necessarily at least one pair $i, \ell \in [n]$ such that $(\codeSketch_b(\vec{b}_i), \gaptuple_i)$ is embedded into $\codeSketch_a(\vec{a}_{\ell})$'s second track. We next define $w$ and $(p, \gaptuple)$, and then discuss why they satisfy the property from above:\looseness=-1

\hspace*{1.2cm}
$w = \left (\prod_{i=1}^{n-1} @ \codeSketch_a(\vec{a}_i) \right) @ \codeSketch_a(\vec{a}_n)\left (\prod_{i=1}^{n-1} @ \codeSketch_a(\vec{a}_i) \right) @\,,$\\
\hspace*{1.1cm}$(p,\gaptuple) = @\stackrel{\leq 5}{\leftrightarrow} \left( \prod_{j=1}^{n-1} \codeSketch_b(\vec{b}_j) \stackrel{\leq 1}{\leftrightarrow} \# \stackrel{\leq 3}{\leftrightarrow} \#  \stackrel{\leq 1}{\leftrightarrow} \# \stackrel{\leq 3}{\leftrightarrow} \# \stackrel{\leq 6}{\leftrightarrow} \right) \codeSketch_b(\vec{b}_n) \stackrel{\leq 5}{\leftrightarrow} @.$

If $p \subseq_{e} w$ for an embedding $e$ satisfying $\gaptuple$, then the first $@$-symbol of $p$ is mapped to the $@$-symbol of $w$ occurring before an occurrence of $\codeSketch_a(\vec{a}_{\ell_1})$ for some $\ell_1$, and this occurrence is in the prefix $\left (\prod_{i=1}^{n-1} @ \codeSketch_a(\vec{a}_i) \right)$ of $w$. By reasoning about the occurrences of symbols $\#$ and the length constraints (see Appendix \ref{sec:MatchingLowerBound}), we can show that $\codeSketch_b(\vec{b}_1)$ must be embedded in $\codeSketch_a(\vec{a}_{\ell_1})$ in the way discussed above (i.\,e., $\gaptuple_1$ is satisfied and $\codeSketch_b(\vec{b}_1)$ is entirely mapped to some track $q \in \{1, 2, 3\}$ of $\codeSketch_a(\vec{a}_{\ell_1})$). For simplicity, assume that $\ell_1 \leq n - 1$. The factor $\stackrel{\leq 1}{\leftrightarrow} \# \stackrel{\leq 3}{\leftrightarrow} \#  \stackrel{\leq 1}{\leftrightarrow} \# \stackrel{\leq 3}{\leftrightarrow} \# \stackrel{\leq 6}{\leftrightarrow}$ between $(\codeSketch_b(\vec{b}_1),\gaptuple_1)$ and the next part $(\codeSketch_b(\vec{b}_2),\gaptuple_2)$ will enforce that $\codeSketch_b(\vec{b}_2)$ is embedded in $\codeSketch_a(\vec{a}_{\ell_1 + 1})$, and, moreover, it will be mapped to $\codeSketch_a(\vec{a}_{\ell_1 + 1})$'s track $q$ or $q+1$ (as there can be at most $18$ symbols between $\codeSketch_b(\vec{b}_1)$ and $\codeSketch_b(\vec{b}_2)$, track $3$ cannot be reached in the case $q=1$). \par
By repeating this argument, we can show that if $(\codeSketch_b(\vec{b}_j),\gaptuple_j)$ is embedded in track $s$ of $\codeSketch_a(\vec{a}_{\ell_{j}})$ (with $\ell_j \leq n - 1$), then $(\codeSketch_b(\vec{b}_{j+1}),\gaptuple_{j+1})$ is embedded in track $s$ or $s + 1$ of $\codeSketch_a(\vec{a}_{\ell_{j} + 1})$ in case that $s \in \{1, 2\}$, and it is necessarily embedded in track $s$ of $\codeSketch_a(\vec{a}_{\ell_{j} + 1})$ in case that $s = 3$. If $\ell_j = n$, then analogously $(\codeSketch_b(\vec{b}_{j+1}),\gaptuple_{j+1})$ is mapped to $\codeSketch_a(\vec{a}_{1})$ of $w$'s suffix $(\prod_{i=1}^{n-1} @ \codeSketch_a(\vec{a}_i)) @$. 
Consequently, each $\codeSketch_b(\vec{b}_j)$ is mapped to a track of $\codeSketch_a(\vec{a}_{\ell_j})$, and the tracks to which these $\codeSketch_b(\vec{b}_j)$ are mapped may start with track $1$ or $2$, and then can only increase until we possibly map some $\codeSketch_b(\vec{b}_j)$ to track $3$. However, after having mapped the last occurrence of $\#$ in $\codeSketch_b(\vec{b}_n)$ to an occurrence of $\#$ in $\codeSketch_a(\vec{a}_{\ell_n})$, we can afford a gap of length at most $5$ before mapping the last symbol $@$ of $(p, \gaptuple)$ to an occurence of $@$ in $w$. By the structure of $(p, \gaptuple)$ and $w$, this is only possible if $\codeSketch_b(\vec{b}_n)$ is mapped to track $2$ or $3$ of $\codeSketch_a(\vec{a}_{\ell_n})$. \looseness=-1 \par
We conclude that if $p \subseq_{\gaptuple} w$, then, for some $j, \ell_j \in [n]$, $(\codeSketch_b(\vec{b}_j), \gaptuple_j)$ is mapped to track $2$ of $\codeSketch_a(\vec{a}_{\ell_j})$; thus, $\vec{a}_{\ell_j}$ and $\vec{b}_{j}$ are orthogonal. On the other hand, the explanations from above show that if $\vec{a}_{\ell_j}$ and $\vec{b}_{j}$ are orthogonal vectors, then $p$ can be embedded into $w$ by an embedding that satisfies $\gaptuple$, i.\,e., an embedding that maps $(\codeSketch_b(\vec{b}_j), \gaptuple_j)$ to track $2$ of $\codeSketch_a(\vec{a}_{\ell_j})$, all $(\codeSketch_b(\vec{b}_{j'}), \gaptuple_{j'})$ with $1 \leq j' < j$ to the first tracks of some $\codeSketch_a(\vec{a}_{\ell_{j'}})$, and all $(\codeSketch_b(\vec{b}_{j'}), \gaptuple_{j'})$ with $j < j' \leq n$ to the third tracks of some $\codeSketch_a(\vec{a}_{\ell_{j'}})$. \par
In this reduction, we have $|\Sigma| = 4$, all constraints are $(0, \ell)$ with $\ell \leq 6$, and $|w|, |p| \in \Theta(nd)$. If $\matchProb$ can be solved in $O(|w|^g|p|^h)$ with $g+h=2-\epsilon$ for some $\epsilon>0$, then $\OV$ can be solved in $O(nd + (nd)^{2-\epsilon})$. Since $\nz(\gaptuple) \in \Theta(nd)$, solving $\matchProb$ in $O(|w|^g\nz(\gaptuple)^h)$ with $g+h=2-\epsilon$ for some $\epsilon<0$ also contradicts $\OVH$.
\end{proof}

We emphasise that, according to our proof, these lower bounds hold for $\matchProb$ with length constraints even if we only have constant upper bounds on the length of the gaps.

\begin{corollary}\label{lowerBoundRegular}
$\matchProb$ with regular constraints
cannot be solved in $\mathcal{O}(|w|^h \states(\gaptuple_p)^g)$ time with $h+g= 2-\epsilon$ for some $\epsilon>0$, unless $\OVH$ fails. This holds even if $|\Sigma| = 4$ and all regular constraints are expressed by constant size DFAs.
\end{corollary}

From Theorem~\ref{lowerBoundLength} and Corollary~\ref{lowerBoundRegular} we also get that $\matchProb$ with length, regular, or reg-len constraints cannot be solved in $\mathcal{O}(|w|^h |p|^g)$ time, with $h+g=2-\epsilon$, nor in $\mathcal{O}(|w|^{2-\epsilon})$ time. Moreover (see Appendix \ref{sec:lowerBoundsBinary}) we can show similar lower bounds for $|\Sigma|=2$ as well. \par

Compared to the $\OVH$-bound for regular expression matching of~\cite{BackursIndyk2016}, we provide a lower bound for a much more restricted problem (i.\,e., matching gapped sequences with length constraints, a subclass of regular expressions that still seems to have a significant practical relevance); thus, a stronger lower bound (this is also why our $\OV$-reduction has a significantly different structure and is technically more involved than that of~\cite{BackursIndyk2016}). In particular, our lower bound applies (unlike those from~\cite{BackursIndyk2016}) to the
case of matching variable length gap patterns, and settles the complexity of that problem. 
We wrap up this section by noting that Theorem \ref{lowerBoundLength} and Corollary \ref{lowerBoundRegular} show that (if $\OVH$ holds) the algorithm of Theorem \ref{constantPatternsRegularLength}, also when used for regular constraints or length constraints only, is optimal in the sense that there are no algorithms which can solve $\matchProb$ in the respective settings polynomially faster.\looseness=-1

\section{Analysis Problems for Gapped Subsequences}\label{sec:NonUniversality}

Let us recall that the \emph{universality}, \emph{containment} and \emph{equivalence problem} (denoted by $\uniProb$, $\contProb$ and $\equiProb$ for short) consist in deciding $\subseqSet{\gaptuple}{w} = \Sigma^{k}$, $\subseqSet{\gaptuple}{w} \subseteq \subseqSet{\gaptuple}{w'}$, and $\subseqSet{\gaptuple}{w} =  \subseqSet{\gaptuple}{w'}$, respectively, for a given $(k-1)$-tuple $\gaptuple$ of gap constraints and strings $w, w' \in \Sigma^*$. 
As mentioned before, these problems can be solved in polynomial time for classical subsequences (see Appendix \ref{sec:ConP} for further details). We show next that these problems become much harder for non-trivial length or regular constraints. \par
From Cor.~\ref{constantPatternsRL}~and~Thm.~\ref{constantPatternsRegularLength}, we can directly conclude the following brute-force upper bounds.

\begin{theorem}\label{upperBoundUnboundedAlphabets}
(1) The problems $\uniProb$, $\contProb$ and $\equiProb$ with length (or reg-len) constraints can be solved in time $\bigO(|\Sigma|^{k} \nz(\gaptuple) \ell)$ (respectively, $\bigO(|\Sigma|^{k} \states(\gaptuple) \ell)$), where $\ell = \max\{|w|, |w'|\}$.\\
(2) The problems $\uniProb_{\Sigma}$, $\contProb_{\Sigma}$ and $\equiProb_{\Sigma}$ with length (or reg-len) constraints can be solved in time $2^{\bigO(k)} \nz(\gaptuple) \ell$ (respectively, $2^{\bigO(k)} \states(\gaptuple) \ell$), where $\ell = \max\{|w|, |w'|\}$.
\end{theorem}

We shall next complement these brute-force upper bounds by suitable lower bounds, which demonstrate that significantly faster algorithms are unlikely to exist. For convenience, we state our complexity results for the complement problems, i.\,e., \emph{non-universality problem} ($\nuniProb$), \emph{non-containment problem} ($\ncontProb$), and \emph{non-equivalence problem} ($\nequiProb$). Moreover, we state the lower bounds for the case of length constraints only. By simply interpreting the length constraints as regular constraints, all the lower bounds also apply to the case of regular constraints (this does not cause an exponential size increase of the instances, see Appendix \ref{sec:instances}).\looseness=-1\par 
Our first result establishes the general $\npclass$-completeness (even for small constant alphabets and length constraints), and that the exponent $\bigO(k)$ of Theorem~\ref{upperBoundUnboundedAlphabets}(2) cannot be significantly improved, unless ETH or SETH fail. We will discuss some proof ideas later on. \looseness=-1

\begin{theorem}\label{HardnessNonUniversalityBounded}
For every fixed alphabet $\Sigma$ with $|\Sigma| \geq 3$, $\nuniProb_{\Sigma}$, $\ncontProb_{\Sigma}$ and $\nequiProb_{\Sigma}$ with length constraints are $\npclass$-complete, even if all length constraints are $(1, 5)$. Moreover,
\begin{itemize}
\item they cannot be solved in subexponential time $2^{\smallO(k)} \poly(|w|, k))$ (unless ETH fails), 
\item they cannot be solved in time $\bigO(2^{k(1-\epsilon)} \poly(|w|, k))$ (unless SETH fails).
\end{itemize}
\end{theorem}

This directly leads to the question whether these problems are tractable if $|\Sigma| \leq 2$. This is obviously true for unary alphabet $\Sigma = \{\ta\}$ (note that in this case, $\subseqSet{\gaptuple}{w} = \Sigma^k = \{\ta^k\}$ if $(\sum_{i \in [k]} \lowerBoundShort{i}) + k \leq |w|$, and $\subseqSet{\gaptuple}{w} = \emptyset$ otherwise), but $\npclass$-complete for $|\Sigma| = 2$:

\begin{theorem}\label{binaryCaseHardnessTheorem}
For every fixed alphabet $\Sigma$ with $|\Sigma| = 2$, $\nuniProb_{\Sigma}$, $\ncontProb_{\Sigma}$ and $\nequiProb_{\Sigma}$ with length constraints are $\npclass$-complete even if each length constraint is $(0, 0)$ or $(3, 9)$. 
\end{theorem}

Let us now consider the case where $\Sigma$ is not treated as a constant. Theorem~\ref{upperBoundUnboundedAlphabets} means that $\nuniProb$, $\ncontProb$ and $\nequiProb$ with length constraints are trivially fixed parameter tractable if parameterised by both $|\Sigma|$ and $k$. Moreover, since $\ell = \max\{|w|, |w'|\}$ bounds both $|\Sigma|$ and $k$, we also have fixed parameter tractability with respect to $\ell$ for trivial reasons. Are the problems fixed-parameter tractable with respect to the single parameter $|\Sigma|$ or the single parameter $k$? With respect to $|\Sigma|$, this is answered in the negative by Theorem~\ref{binaryCaseHardnessTheorem} (unless $\pclass = \npclass$). With respect to parameter $k$, the following result gives a negative answer as well.

\begin{theorem}\label{HardnessNonUniversalityUnbounded}
Problems $\nuniProb$, $\ncontProb$ and $\nequiProb$ with length constraints cannot be solved in running time $\bigO(f(k) \poly(|w|, k))$ for any computable function $f$ (unless $\fptclass = \wclass[1]$). 
\end{theorem}

This result only holds for unbounded alphabets and length constraints. Indeed, for constant $\Sigma$ the brute-force algorithm is an fpt-algorithm with respect to $k$. Moreover, if the upper length constraints are bounded by some constant $\ell$, then we only have to enumerate at most $\ell^{k-1}$ candidate tuples of gap sizes and check whether one of them induces an embedding satisfying $\gaptuple$ with respect to $w$, which again would yield an fpt-algorithm with respect to $k$.

\textbf{Proof Ideas for the Lower Bounds.} We present some proof ideas and sketches for the lower bounds mentioned above. For convenience, we only consider the non-universality problem here. Full proof details can be found in Appendix \ref{sec:NonUniversalityAppendix}. \par
Theorem~\ref{HardnessNonUniversalityBounded} can be proven by a reduction from $\SatProb$. In order to get the ETH and SETH lower bounds, this reduction must yield instances with a $(k-1)$-tuple of gap constraints, where $k$ is exactly the number of Boolean variables. Theorem~\ref{HardnessNonUniversalityUnbounded} can be shown by a very similar reduction that starts from the standard parameterisation of the independent set problem. Both reductions can be conveniently defined by using a \emph{meta non-universality problem} ($\metaNuniProb$ for short) as an intermediate step, which we define next.\par
Let $\Gamma = \{b_1, b_2, \ldots, b_m\}$ be some alphabet, and let $q, k \in \mathbb{N}$. An instance of the problem is a $(q \times k)$-matrix with the entries $W_{i, j}$, which are subsets of $\Gamma$. For every $i \in [q]$, we associate with row $i$ of the matrix the language $\lang{W_i} = W_{i, 1} \cdot W_{i, 2} \cdots W_{i, k}$, i.\,e., we simply represent the elements of $W_{i, 1} \times W_{i, 2} \times \ldots \times W_{i, k}$ as length-$k$ strings over $\Gamma$ in the natural way. The question is then to decide whether $\cup_{i \in [q]} \lang{W_i} \neq \Gamma^k$ (see Figure~\ref{figure:exampleReductions} for an example).

\begin{figure}
\centering

\begin{equation*}
\begin{pmatrix}
\{\ta\} & \{\tb\} & \{\ta, \tc, \td\}\\
\{\tc, \td\} & \{\tb\} & \{\ta\}\\
\{\ta\} & \{\tb, \tc, \td, \te\} & \{\td\}\\
\{\te\} & \{\tb\} & \{\tc, \te\}
\end{pmatrix}
\hspace{1cm}
\begin{pmatrix}
\{0\} & \{1\} & \{0\} & \{0, 1\} & \{0, 1\} & \{0, 1\} &\\
\{1\} & \{1\} & \{0, 1\} & \{0, 1\} & \{0, 1\} & \{0\} &\\
\{0, 1\} & \{0, 1\} & \{0\} & \{0\} & \{0\} & \{0, 1\} &
\end{pmatrix}
\end{equation*}

\vspace*{-5pt}   
\caption{Left side: example instance of $\metaNuniProb$ for $\Gamma = \{\ta, \tb, \tc, \td, \te\}$, $q = 4$, and $k = 3$. Note that, e.\,g., $W_{3, 2} = \{\tb, \tc, \td, \te\}$ and $W_{4, 1} = \{\te\}$; moreover, $\lang{W_1} = \{\ta\} \cdot \{\tb\} \cdot \{\ta, \tc, \td\} = \{\ta \tb \ta, \ta \tb \tc, \ta \tb \td\}$. Since $\cup_{i \in [4]} \lang{W_i} \neq \Gamma^3$, this is a negative instance. Right side: the $\SatProb$-instance $c_1 = \{v_1, \neg v_2, v_3\}$, $c_2 = \{\neg v_1, \neg v_2, v_5\}$, $c_3 = \{v_3, v_4, v_5\}$ over the variables $\{v_1, v_2, \ldots, v_6\}$ as an instance of $\metaNuniProb$ . Note that $100010 \notin \cup_{i \in [3]} \lang{W_i}$; thus, $100010$ is a satisfying assignment.}
\label{figure:exampleReductions}
\vspace*{-15pt}
\end{figure}

We next discuss, how we can reduce $\SatProb$ to $\metaNuniProb$. Let $F = \{c_1, c_2, \ldots, c_q\}$ be a Boolean formula in CNF on variables $\{v_1, \ldots, v_k\}$ (i.\,e., $c_i \subseteq \{v_1, \neg v_1, \ldots, v_k, \neg v_k\}$). We define alphabet $\Gamma = \{0, 1\}$ and the $(q \times k)$-matrix with the entries $W_{i, j}$ as follows. For every $i \in [q]$ and $j \in [k]$, we define $W_{i, j} = \{0\}$, if $v_j \in c_i$, $W_{i, j} = \{1\}$, if $\neg v_j \in c_i$, and $W_{i, j} = \{0, 1\}$, if $\{v_j, \neg v_j\} \cap c_i = \emptyset$. It can be verified with moderate effort, that for every $i \in [q]$, $\lang{W_i}$ contains exactly the Boolean assignments that do not satisfy clause $c_i$.
Hence, $\cup_{i \in [q]} \lang{W_i} \neq \{0, 1\}^k$ if and only if $F$ is satisfiable (see Figure~\ref{figure:exampleReductions} for an example). \par
In a rather similar way, we can also phrase the independent set problem in terms of $\metaNuniProb$. For the independent set problem, we get an undirected graph $G = (V, E)$ with $|V| = n$ and $E = \{e_1, e_2, \ldots, e_m\}$, and a $k \in [|V|]$, and the question is whether $G$ has a $k$-independent set, i.\,e., a set $A \subseteq V$ with $|A| = k$ and $\{u, u'\} \notin E$ for every $u, u' \in A$ with $u \neq u'$. This can be expressed in terms of $\metaNuniProb$ as follows. We interpret the set $V$ of vertices as the alphabet $\Gamma$. We fix some bijection $\nu : \{(i, r, s) \in [m] \times [k] \times [k] \mid r \neq s\} \to [m k(k-1)]$. For every $i \in [m]$ with $e_i = (u, v)$, and every $r, s, j \in [k]$ with $r \neq s$, we define $W_{\nu(i, r, s), j} = \{u\}$, if $j = r$, $W_{\nu(i, r, s), j} = \{v\}$, if $j = s$, and $W_{\nu(i, r, s), j} = V$, else. For example, if $e_9 = (v_3, v_7)$ and $k = 4$, then row $\nu(9, 2, 4)$ of the matrix would be $V \:\: \{v_3\}\:\:V\:\:\{v_7\}$. \par
It is a bit more difficult to see why this reduction works. The idea is that we represent sets of vertices of cardinality \emph{at most} $k$ by length-$k$ strings over $V$ (note that sets of cardinality strictly less than $k$ can be represented by strings with repeated symbols). For every edge $(u, v)$ and for all pairs of positions $r, s \in [k]$, the language $\lang{W_{\nu(i, r, s)}} = W_{\nu(i, r, s), 1} W_{\nu(i, r, s), 2} \ldots W_{\nu(i, r, s), k}$ represented by row $\nu(i, r, s)$ of the matrix contains exactly the strings $w \in \Gamma^k$ with $(w[r], w[s]) = (u, v)$, i.\,e., strings that represent non-independent sets with edge $(u, v)$. For the example $e_9 = (v_3, v_7)$ and $k = 4$, we have $\lang{W_{\nu(9, 2, 4)}} = \{v_1 v_3 v_1 v_7,$ $v_2 v_3 v_1 v_7,$ $\ldots,$ $v_n v_3 v_1 v_7,$ $\ldots,$ $v_1 v_3 v_2 v_7,$ $v_2 v_3 v_2 v_7,$ $\ldots\}$. \par
This whole idea works only because, in our setting, we assume that every vertex has a loop since then strings $w$ of $V^k$ contain an edge $(w[r], w[s]) \in E$ for some $r, s \in [k]$ if and only if the corresponding set of vertices is not independent or of cardinality strictly less than $k$ (the latter is represented by a loop, i.\,e., $w[r] = w[s]$).
In summary, $G$ has a $k$-independent set if and only if not all length-$k$ strings are in $\bigcup_{i \in [m], r, s \in [k], r \neq s} \lang{W_{\nu(i, r, s)}}$. \par
The main technical challenge is to show a reduction from $\metaNuniProb$ to $\nuniProb$ with length constraints. We next give a sketch of this reduction.
Let $\Gamma = \{b_1, b_2, \ldots, b_m\}$, $q, k \in \mathbb{N}$, and, for every $i \in [q], j \in [k]$, let $W_{i, j} \subseteq \Gamma$. We transform this $\metaNuniProb$ instance into an instance of $\nuniProb$ with length constraints as follows. We first define the alphabet $\Sigma = \Gamma \cup \{\#\}$ (with $\# \notin \Gamma$). Then we define a $(k-1)$-tuple $\gaptuple = (C_1, C_2, \ldots, C_{k-1})$ of gap constraints with $C_i = (\lowerBoundShort{i}, \upperBoundShort{i}) = (m-1, 3m-1)$ for every $i \in [k-1]$ (recall that $m$ is $\Gamma$'s cardinality). To conclude the reduction, we have to construct a string $K(W_1, \ldots, W_q)$ over $\Sigma$, such that $\subseqSet{\gaptuple}{K(W_1, \ldots, W_q)} = \Sigma^k$ if and only if $\cup_{i \in [q]} \lang{W_i} = \Gamma^k$. We do this in several steps. \looseness=-1\par
For every $i \in [q]$ and $j \in [k]$, let $w_{i, j} \in \Gamma^*$ be some string representation of $W_{i, j}$, i.\,e., $\alphabet{w_{i, j}} = W_{i, j}$ and $|w_{i, j}| = |W_{i, j}| \leq m$. For every $i \in [q]$, we define the string\\
\hspace*{1,2cm}{$
S(W_{i}) = w_{i, 1} (\#)^{m-1} w_{i, 2} (\#)^{m-1} \ldots (\#)^{m-1} w_{i, k}\,.
$}

We can show that those $\gaptuple$-subsequences of $S(W_{i})$ that do not contain occurrences of symbol $\#$ must be mapped to $S(W_{i})$ in such a way that each $j \in [k]$ is mapped to $w_{i, j}$. More precisely, for every $i \in [q]$, we have that {$(\subseqSet{\gaptuple}{S(W_{i})} \cap \Gamma^*) = \lang{W_i}$.} \hfill ($\dagger$)\par

Next, we define a string $T$ whose purpose it is to contain \emph{all} $\gaptuple$-subsequences that contain at least one occurrence of $\#$.
For every $i \in [k]$, let $T_i = T_{i, 1} T_{i, 2} \ldots T_{i, k}$, where, for every $j \in [k] \setminus \{i\}$, $T_{i, j} = b_1 b_2 \ldots b_m \#^{m}$, and $T_{i, i} = \#^{m}$. We define $T$ by $T = T_1 (\#^{3m}) T_2 (\#^{3m}) \ldots (\#^{3m}) T_k$.
The idea here is that any $\gaptuple$-subsequence of $T$ must be mapped entirely into some $T_i$, which, due to the length constraints, forces position $i$ to be mapped to $T_{i, i} = \#^{m}$, i.\,e., to an occurrence of $\#$. More precisely, we have {$\subseqSet{\gaptuple}{T} = \{w \in \Sigma^k \mid |w|_{\#} \geq 1\}$.} \hfill ($\diamond$)

Finally, we set $K(W_{1}, \ldots, W_{q}) = T (\#^{3m}) S(W_1) (\#^{3m}) S(W_2) (\#^{3m}) \ldots (\#^{3m}) S(W_q)$. By using ($\dagger$) and ($\diamond$) from above,
we can now prove $\subseqSet{\gaptuple}{K(W_1, \ldots, W_q)} = \Sigma^{k}$ if and only if $\cup_{i \in [q]} \lang{W_{i}} = \Gamma^k$, which concludes the proof of correctness. \par

For proving Theorem~\ref{binaryCaseHardnessTheorem}, using $\metaNuniProb$ as an intermediate step seems not possible, since it introduces another symbol $\#$ to the alphabet. However, we can devise a similar reduction.
The main difference is that we represent each Boolean variable by two consecutive symbols of the subsequence, i.\,e., we need a $(2k-1)$ tuple of length constraints (therefore, the reduction does not yield a SETH bound as mentioned in Theorem~\ref{HardnessNonUniversalityBounded}). Since we cannot conveniently use a separator $\#$ that is not used for expressing Boolean assignments, the constructed string is more complicated in this reduction (see Appendix~\ref{sec:NonUniversalityAppendix} for full details).




\appendix 

%
%

\section{Regular and Length Constraints}

\subsection{On our choice  of Representing Regular Constraints By DFAs (and potential complications resulting from alternative representations)} \label{NFAvsDFA}

The question on how should one represent the regular constraints (both when they appear alone, and when they appear in conjunction with length constraints) seems a valid and interesting question to us. Natural options would have been NFAs, DFAs, or regular expressions (regexes). We have chosen to represent them with DFAs and, in the following, we argue that this is a reasonable (and not unusual) choice. 

{\bf Impact of the choice.} Before putting forward our argument, we note that this choice impacts only two results, namely Corollary \ref{constantPatternsRL}(2) and Theorem \ref{constantPatternsRegularLength} (the algorithmic results for matching subsequences with regular gap constraints). All the other results hold irrespective of the representation used for the regular language present in the constraints.

{\bf Motivation of the choice.} As said, we have chosen to represent the regular constraints as deterministic finite automata (DFAs), rather than representing them as regular expressions (regex) or non-deterministic finite automata (NFAs). This representation is not unusual when dealing with testing whether factors of words are in a regular language, see, e.\,g., the overview of the results of Imre Simon regarding factorization forests and their application to such problems \cite{Bojanczyk09} or the works related to sliding window algorithms for regular languages \cite{GanardiHKLM18,GanardiHL18,GanardiHL18b,GanardiHLS19}. 

Additionally, we note that in the case of regular constraints of constant complexity (i.\,e., where $O(1)$ states are needed for each regular constraint) there is no difference in the asymptotic complexity of our algorithms w.\,r.\,t. the chosen representation of the constraints: DFA, NFA, or regex. Moreover, all the lower bounds would still hold as stated. To the same end, it is worth noting that in some of the works on sliding window algorithms for regular languages, the DFAs for the respective regular languages are assumed to have constant size \cite{GanardiHLS19}; a deeper discussion of such restriction in that setting is made in \cite{GanardiHKLM18}. 

Finally, regular constraints of constant (and relatively small) complexity are quite usual in practice, as they can model simple constraints such as the presence/absence of some letters in a string, filtering according to the presence/absence of some constant strings, restrictions on the order in which some symbols appear in a word, etc. 

The survey \cite{amadini2021survey} and the references therein provide examples of {\em small} regular constraints (in the sense that they are accepted by automata with a small number of states, or described by short regexes) which appear and are relevant in the area of string solving, and, as such, formal verification. As string solving is an area in which regular constraints on strings play an important role, we have investigated Kaulza \cite{saxena2010symbolic,KaluzaWeb,liang2016efficient}, one the standard benchmarks containing string constraints, developed based on practical applications of string solving (in particular, symbolic execution), and usually used in the evaluation of string solvers (as mentioned in \cite{zaligVinderJournal}). In this investigation, we have focused on the regular constraints (appearing alone or in conjunction with length constraints or other types of constraints) and their complexity (for simplicity, we present here the length of the regex specifying them and the number of states in a minimal DFA accepting them). We have used the BASC tool (\url{https://b4sc.github.io/}) to extract the wanted information from the respective benchmark, and have obtained the following results.
\begin{itemize}
\item The Kaluza benchmark contains 47305 instances, out of which 20740, that is around $43\%$, contain regular constraints. 
\item In total, there are 207038 regular constraints (specified as regular expressions) appearing in these instances (there can, of course, be more constraints in each instance, sometimes even more constraints for the same variable). 
\item All these regular expressions have length at most 20. The average length of the regexes occurring in Kaluza is lower than $8$. The NFAs cannonically obtained from these regexes have, in average, $17$ states. 
\item $99\%$ of the minimal DFAs corresponding to these regexes have at most $20$ states, and the average number of states in these minimal DFAs is lower than $11$.
\end{itemize}

{\bf DFAs, NFAs, regexes: differences in size (in theory).} Folklore results show that there are regular languages $L$ for which the size of the shortest regular expression describing $L$ is logarithmic w.\,r.\,t. the size of the minimal DFA accepting $L$. However, there are also regular languages $L$ for which the size of the DFA accepting $L$ is logarithmic w.\,r.\,t. the size of the shortest regular expression describing $L$ \cite[Example 23 in the Arxiv version]{GruberH15}.  So, among DFAs and regexes there is no representation of a regular language which is guaranteed to be exponentially smaller than the other one. As DFAs are particular cases of NFAs, it is clear that there is no advantage in choosing DFAs over NFAs; however, dealing with unrestricted NFAs induced some complications which we were not able to solve in the case of Theorem \ref{constantPatternsRegularLength} and obtain an algorithm which matches the lower bounds we were able to show (more on this below).

{\bf On the impacted results.} So, let us address now more throughly the two results which are impacted by the choice of the model used to describe regular gap constraints. 

In the case of Corollary \ref{constantPatternsRL}(2), one can easily show that $\matchProb$ with regular constraints can be solved in $O(|w|\size(\gaptuple))$ time when the constraints are specified as NFAs or regexes, using a dynamic programming strategy (this is discussed later in the Appendix). So, the difference is that, unsurprisingly, the total size of the automaton appears here instead of the number of states. 

The case of Theorem \ref{constantPatternsRegularLength} is more complicated (and it would make sense to revisit this paragraph after reading its complete proof, appearing later in this appendix). In that case, the regular constraints appear in conjunction with length constraints. The key idea used in our efficient algorithm solving $\matchProb$ is that the traces of the computations of each automaton encoding the constraints on suffixes of the word are essentially linear lists of pairs (position in the word, state of the automaton); that is, they do not fork, because the automaton is deterministic. This allows us to represent the totality of these traces as a forest of disjoint trees. This property is already important for the first generic step of our algorithm, but it is crucial for the later steps: in the case when the regular constraints would be given as regexes or NFAs, the first generic step would still produce a graph $G_{t+1}$, but this would be a directed acyclic graph, without the nice disjointness properties which are used later in the algorithm to get a good running time (for instance, the sets $S_q$ from the simpler variant described in the proof of Theorem \ref{constantPatternsRegularLength} would not be disjoint and managing them would be more time consuming). Hence, the data structures and ideas we use in the later steps of that algorithm (e.\,g., level ancestor data structures or our marking procedure), and which are custom-designed for such forests of disjoint trees, seem to need non-trivial adaptations and extensions to work for directed acyclic graphs  (as it would be the case when implementing the first generic step of the algorithm in the case of constraints defined by regexes or NFAs) within a similar rectangular complexity. 

It seems that this is also a good point to emphasise that both the more complicated data structures (w.\,r.\,t. those needed to get the results in Corollary \ref{constantPatternsRL}) as well as the usage DFAs in the framework of Theorem \ref{constantPatternsRegularLength} are a consequence of the interplay between regular and length constraints, which complicates the matching problem $\matchProb$ significantly. 

Thus, with respect to Theorem \ref{constantPatternsRegularLength}, we preferred to stick to the representation of regular constraints by DFAs, as, in this setting, we have obtained matching upper and lower bounds (in Theorem \ref{lowerBoundLength}) for the complexity of the matching problem, and we think that this is mathematically interesting. It is an interesting open problem whether a similar result can be obtained when specifying the constraints by NFAs (or regular expressions); there does not seem to be an easy way to modify the construction below to work for nondeterministic automata or regexes.

\subsection{Conjunctions of Regular {\bf and} Length Constraints?}\label{ap:RegAndLen}

After the discussion on the way regular constraints are represented, it is maybe a good point to also briefly mention that dealing with conjunctions of regular and length constraints is also not unusual at all. In our investigation of Kaluza string solving benchmark, where both these types of constraints are usual, we have noted the following.
\begin{itemize}
\item The Kaluza benchmark contains $47305$ string solving instances. From these, $20740$ instances contain regular constraints (approximatively $43\%$) and $21246$ instances contain length constraints (approximatively $44\%$). 
\item There is a large overlap between the instances containing regular and length constraints. There are $19812$ instances which contain both types of constraints. This corresponds to $95\%$ of the instances containing regular constraints, and to $93\%$ of the instances containing length constraints.
\end{itemize}
In general, length constraints can be more general than lower and upper bounds on the length of the variables (which intuitively correspond to the gaps in our setting); they can be linear relations between the lengths of variables and constants. Our result in Theorem \ref{constantPatternsLengthExtended} shows that $\matchProb$ becomes NP-hard in the respecting setting. However, a visual inspection of a random sample of the Kaluza instances showed that in many cases the length constraints already are, or can be immediately reduced, to the simpler case of only ower and upper bounds on the length of the variables. 

To conclude, at least when considering the representative Kaluza string solving benchmark, it is very common to have combinations of regular and length constraints and, at least to a certain extent, these are not very complex.

\section{Some Details Omitted From the Preliminaries}

\subsection{Computational Model}\label{sec:compModel}

The computational model we use is the standard unit-cost RAM with logarithmic word size: for an input of size $n$, each memory word can hold $\log n$ bits. Arithmetic and bitwise operations with numbers in $[1:n]$ are, thus, assumed to take $O(1)$ time. Numbers larger than $n$, with $\ell$ bits, are represented in $O(\ell/\log n)$ memory words, and working with them takes time proportional to the number of memory words on which they are represented. In all the problems, we assume that we are given one word $w$ (respectively, two words $w$ and $w'$), with $|w|=n$ (respectively, $|w|=n\geq |w'|=m$), over an alphabet $\Sigma=\{1,2,\ldots,\sigma\}$, with $|\Sigma|=\sigma\leq n$. Whenever a gapped sequence $(p,\gaptuple)$ is also part of the input of the problem, then $|p|\leq n$ and $p$ is also over the same alphabet $\Sigma$ defined above. The way $\gaptuple$ is given is discussed in the main part of the paper. 
So, in general, we assume that the processed words are sequences of integers (called letters or symbols), each fitting in $O(1)$ memory words. This is a common assumption in string algorithms: the input alphabet is said to be {\em an integer alphabet}. For a more detailed general discussion on this model see, e.\,g.,~\cite{crochemore}. 

\subsection{Basic Stringology Data Structures}

For a length-$n$ word $w$ we can build in $\bigO(n)$ time the \emph{suffix array} structure, as well as the \emph{longest common extension} (also known as {\em longest common prefix}) data structures allowing us to retrieve in constant time the length of the longest common prefix of any two suffxes $w[1..n]$ and $w[1..j]$ of $w$ (for details, see, e.\,g.,~\cite{Gusfield1997, KarkkainenEtAl2006,dinklage_et_al:LIPIcs:2020:12905}, and the references therein).

\section{Some Details About Gap Constrained Subsequences Omitted From The Main Part}

\subsection{Size Measure of Our Instances}\label{sec:instances}

\begin{remark}\label{integersRemark}
Let $(p,\gaptuple)$ be a gapped sequence over $\Sigma$, and let $w \in \Sigma^*$. If $|p| > |w|$, then $p$ is not a $\gaptuple$-subsequence of $w$, and $\subseqSet{\gaptuple}{w} = \emptyset$. Hence, we generally assume that for the matching problem we always have $|p| \leq |w|$.
If $(p,\gaptuple)$ is a gapped sequence with length constraints, then $\lowerBoundShort{i} > |w|$ for some $i \in [|p|-1]$ directly implies that $p$ is not a $\gaptuple$-subsequence of $w$, and $\subseqSet{\gaptuple}{w} = \emptyset$, and if $\upperBoundShort{i} > |w|$ for some $i \in [|p|-1]$, we could replace $\upperBoundShort{i}$ by $|w|$ to obtain an equivalent instance. In particular, note that we can check $\lowerBoundShort{i} > |w|$ and $\upperBoundShort{i} > |w|$ in linear time in $|w|$. Consequently, we shall also assume that $\lowerBoundShort{i} \leq \upperBoundShort{i} \leq |w|$ (or $\lowerBoundShort{i} \leq \upperBoundShort{i} \leq \min\{|w|, |w'|\}$ for the non-equivalence, non-containment and non-universality problems).

In particular, the remark above shows that in all the analysis problems we consider in Section \ref{sec:NonUniversality}, we can replace the length constraints by regular constraints without exponentially increasing the overall size of the input instance, which consists in one word $w$ (or two words $w$ and $w'$), and a gapped sequence $(p,\gaptuple)$ (or simply a tuple of gap constraints $\gaptuple$). Indeed, instead of the constraint $(\lowerBoundShort{i},\upperBoundShort{i})$ we can use the regular-constraint $C_i=\{u\in \Sigma^*\mid  \lowerBoundShort{i}\leq |u| \leq \upperBoundShort{i}\}$, which is accepted by a DFA with $O(\upperBoundShort{i})$ states. Clearly, this would lead to an exponential growth in the size of the string describing the respective gap constraint, but  as $\upperBoundShort{i}$ is always upper bounded by the length of the input string(s), the overall growth in the size of the string describing the input instance is just polynomial. 
\end{remark}

\subsection{Constraints for the Prefix and Suffix Gaps}
\label{prefixSuffix}
For a word $w$, with $|w|=m$, a $(k-1)$-tuple of gap constraints $\gaptuple$, and an embedding $e : [k] \to [m]$ which satisfies $\gaptuple$, we note that $\gaptuple$ only restricts the form of gaps induced by $e$ in $w$, i.\,e., of the factors $w[e(i)+1..e(i+1)-1]$, for $i\in [k-1]$. So, it is natural to ask how one could also restrict the factors $w[1..e(1)]$ and $w[e(k)..m]$. Our formalism can model this case too. 
Let $\gaptuple=(C_0,C_1,\ldots, C_{k-1},C_k)$ be a $(k+1)$-tuple of gap constraints for some $p \in \Sigma^k$, i.\,e., $C_0$ and $C_k$ are interpreted as constraints for the prefix $w[1..e(1)]$ and the suffix $w[e(k)..m]$ induced by some embedding $e$ with $p \subseq_e w$. We can describe this setting by defining $p' = \$ p \$$, where $\$ \notin \Sigma$ is a new symbol, and by interpreting $\gaptuple$ as a tuple $\gaptuple'$ of gap constraints for $p'$, i.\,e., a tuple of gap constraints that constrains the $k+1$ actual gaps of $p'$. Then, there is an embedding $e$ that embeds $p$ into $w$ such that $\gaptuple$ is satisfied (i.\,e., with the prefix and suffix gap constraints) if and only if there is an embedding $e'$ that embeds $p'$ into $w$ such that $\gaptuple'$ is satisfied in the sense defined in Section~\ref{sec:GapConstraints}.

In this context, we can revisit our comments on {\bf regular pattern matching} from Section \ref{sec:GapConstraints}. Let $\pi= w_0 x_1 w_1 \cdots x_k w_k $ be a regular (Angluin-style~\cite{Angluin80}) pattern, where $x_i$, for $i\in [k]$ are variables, and $w_i\in \Sigma^*$, for $i\in \{0\} \cup [k]$, are constant factors. We can determine whether a word $w$ matches $\pi$ exactly (so, whether we can replace $x_1,\ldots,x_k$ with words from $\Sigma^*$ such that $\pi$ becomes equal to $w$) as follows. We define the gapped-sequence $(p,\gaptuple)$ with $p=\$w_0w_1\cdots w_k\$$ where $\$$ is a new letter and $\gaptuple[i]=\Sigma^*$ (or $(0,+\infty)$ as length constraint), for all $i=1+\sum^t_{j=0} |w_j|$ with $0 \leq t \leq k$,  and $\gaptuple[i]=\{\emptyword\}$  (respectively, $(0,0)$, as length constraint) otherwise. Then we need to determine if an embedding $e$ exists which satisfies $\gaptuple$ w.\,r.\,t. $\$w\$$. Note that the addition of $\$$ forces the first (respectively, last) symbols of $p$ and $\$ w\$ $ to be aligned, and, as such, it enforces the exact matching of $w$ to $\pi$. Moreover, it also accounts for the case when one of $w_0$ or $w_k$ is not $\varepsilon$, and essentially allows us to impose gap constraints on the initial and final ``gaps'' occurring before and after the embedding of the constant factors of $\pi$ in $w$.

\subsection{Related Problems Expressed by Subsequences with Gap Constraints} \label{relatedProbs}
Let us discuss some examples that show how classical string matching or formal languages problems can be expressed in terms of subsequences with gap constraints.

\begin{itemize}
\item \textsf{String Matching}: To find all occurrences of a word $p$ in a word $w$, it is enough to define the gapped-sequence $(p,\gaptuple)$ with $\gaptuple[i]=\{\emptyword\}$ (or, $(0,0)$ as length constraint), for all $i\in [|p|-1]$, and find all the embeddings $e$ which satisfy $\gaptuple$ w.\,r.\,t. $w$. 

\item \textsf{Regular Pattern Matching}: 
Let $\pi= w_0 x_1 w_1 \cdots x_k w_k $ be a regular (Angluin-style~\cite{Angluin80}) pattern, where $x_i$, for $i\in [k]$, are variables, and $w_i\in \Sigma^+$, for $i\in \{0\} \cup [k]$. To find all  matches of $\pi$ to factors of $w$, we would need to find all assignments of the variables $x_1,\ldots,x_k$ to words from $\Sigma^*$ such that $\pi$ becomes equal to a factor of $w$. Therefore, we define the gapped-sequence $(p,\gaptuple)$ with $p=w_0w_1\cdots w_k$ and $\gaptuple[i]=\Sigma^*$ (respectively, $\gaptuple[i]=(0,+\infty)$ as length constraint), for all $i=\sum^t_{j=0} |w_j|$ with $0 \leq t \leq k-1$,  and $\gaptuple[i]=\{\emptyword\}$  (respectively, $(0,0)$, as length constraint) otherwise. Then we need to find all the embeddings $e$ which satisfy $\gaptuple$ w.\,r.\,t. $w$. See also Appendix \ref{prefixSuffix}.

\item \textsf{Partial Words} (also known as strings with don't cares \cite{Sadri2008}): Consider a partial word $\pi= w_0 \diamond^{\ell_1} w_1 \cdots \diamond^{\ell_k} w_k $, where $\diamond$ is a wildcard, which can be replaced by any letter of the alphabet $\Sigma$, $\ell_i \geq 1$ for all $i\in [k]$, and $w_i\in \Sigma^+$, for $i\in \{0\} \cup [k]$. To find all matches of $\pi$ to factors of $w$, we need to find a replacement of each $\diamond$ from $\pi$ by a letter of $\Sigma$ such that $\pi$ becomes equal to a factor of $w$. Therefore, we define the gapped-sequence $(p,\gaptuple)$ with $p=w_0w_1\cdots w_k$ and $\gaptuple[i]=\Sigma^{\ell_{t+1}}$ (respectively, $(\ell_{t+1}, \ell_{t+1})$ as length constraint), for all $i=\sum^t_{j=0} |w_i|$ with $0 \leq t \leq k-1$,  and $\gaptuple[i]=\{\emptyword\}$  (respectively, $(0,0)$, as length constraint) otherwise. Then we need to find all the embeddings $e$ which satisfy $\gaptuple$ w.\,r.\,t. $w$. \looseness=-1
\item \textsf{Generalisations of the Downward Closure and Simon's Congruence}: Consider the $(n-1)$-tuple of gap constraints $\gaptuple_n$ such that $\gaptuple_n[i] = \Sigma^*$ with $i \in [n-1]$. Then, $\subseqSet{\gaptuple_n}{w}$ is the set of all length-$n$ subsequences of a string $w$. Thus, in this context, the problem of checking $\subseqSet{\gaptuple}{w} = \subseqSet{\gaptuple}{w'} $ for some $\gaptuple$ is a natural extension of the well-known Simon's congruence (see the introduction).
Moreover, we also note that the \emph{downward closure} of $\{w\}$, i.\,e., the set of subsequences of $w$, is $\cup_{n\leq |w|} \subseqSet{\gaptuple_n}{w}$. This can be extended to languages, and we can define the \emph{downward closure} of $L$ as $\cup_{w\in L} \left( \cup_{n\leq |w|} \subseqSet{\gaptuple_n}{w}\right)$.  Hence, considering a set ${\mathcal C}$ of (more complicated)  gap constraints, the set $\cup_{w\in L} \left( \cup_{\gaptuple \in {\mathcal C}} \subseqSet{\gaptuple}{w}\right)$ can be seen as a generalisation of the downward closure of languages, the \emph{downward closure induced by} ${\mathcal C}$. It is interesting to note that, although sometimes it cannot be computed (see \cite{Zetzsche16}, and the references therein), the downward closure of any language is necessarily regular \cite{HAINES196994}. When considering the downward closure induced by a finite set ${\mathcal C}$ of gap constraints, then $\cup_{w\in L} \left( \cup_{\gaptuple \in {\mathcal C}} \subseqSet{\gaptuple}{w}\right)$ is also finite. However, if ${\mathcal C}$ is not finite, then $\cup_{w\in L} \left( \cup_{\gaptuple \in {\mathcal C}} \subseqSet{\gaptuple}{w}\right)$ is not necessarily regular anymore. Indeed, if ${\mathcal C}$ is the set of $n$-tuples of gap constraints $\gaptuple'_n$ such that $\gaptuple'_n[i] = \{\emptyword\}$, with $i \in [n-1]$, then the downward closure of $L$ induced by ${\mathcal C}$ describes the set of factors of the words of $L$, which, in general, is not necessarily regular. 
We consider it an interesting open question to understand for which sets ${\mathcal C}$ of gap constraints the induced downward closure is a regular language for all languages.
\end{itemize}

\section{Proof Details Omitted From Section~\ref{sec:matching}}\label{sec:MatchingAlgo}

Before starting the proof of Theorem~\ref{constantPatternsRegularLength}, let us state a general result which would serve as an example for a basic approach for $\matchProb$ for all the cases approached in this paper.

\begin{proposition}
$\matchProb$ with polynomial constraints can be solved in polynomial time.
\end{proposition}

\begin{proof}
Assume $|w|=n$, $|p|=m$, and $\gaptuple = (C_1,\ldots, C_{m-1})$, where, for every $i \in [m-1]$, the constraint $C_i$ is given as a (black box) procedure $\checkP(i,u)$, which checks $u\in C_i$ in time $O(P(|u|))$ for some polynomial $P$.

To solve the matching problem for the gapped sequence $(p, \gaptuple)$ and the word $w$ we will first preprocess $w$, by computing a three dimensional $n\times n\times m$ array $M[\cdot][\cdot][\cdot]$ where, for $i,j\in[n]$ and $k\in [m]$, $M[i][j][k]=1$ if $w[i..j]\in C_k$ and $M[i][j][k]=0$, otherwise. The elements of $M$ can be computed na\"ively in $O(n^2mP(n))$ time by simply checking for all $i$ and $j$ whether $\checkP(i,w[i..j])$ returns true. 

Further, we define a two dimensional $n\times m$ array $D[\cdot][\cdot]$, where $D[i][k]=1$ if and only if $p[1..k]$ can be embedded in $w[1..i]$ by an embedding $e$ which satisfies $\gaptuple^k$, for $\gaptuple^k = (C_1,\ldots, C_{k-1})$, and maps position $k$ of $p$ to $w[i]$. Otherwise, $D[i][k]=0$. 

The elements of the matrix $D$ can be computed easily by dynamic programming. We initially set all elements of $D$ to $0$. Then, we set $D[i][1]=1$ if and only if $w[i]=p[1]$. Then, for $k\geq 2$, we set $D[i][k]=1$ if and only if there exists $j<i$ such that $D[j][k-1]=1$ and $M[j+1][i-1][k]=1$ and $w[i]=p[k]$. Clearly, this shows how the matrix $D$ can be computed in $O(n^2 m)$. 

Finally, we decide that $p\subseq_{\gaptuple} w$ if and only if there exists $j\leq n$ such that $D[j][m]=1$. By the above, we can decide this in $O(n^2 mP(n) )$ time.
\end{proof}

In general, the na\"ive result in the previous proposition can be improved, as we will see in the following. 

We now see the actual {\bf Proof of Theorem~\ref{constantPatternsRegularLength}}. For certain steps of the algorithm found at the core of this proof, we will sketch how they can be implemented in a simpler but less efficient manner. After the proof, we will also sketch how our approach can be simplified for the case when {\bf only} regular or length constraints are used. None of these simplifications seem to work in the case when {\bf both} regular {\bf and} length constraints are used.

\begin{proof}
Assume $|w|=n$, $|p|=m$, and $\gaptuple =(C_1,\ldots, C_{m-1})$ such that, for $i\in [m-1]$, $C_i=((\lowerBoundShort{i},$ $\upperBoundShort{i}),A_{i}))$, where $A_i=(Q_i,q_{0,i},F_i,\delta_i)$ are DFAs defining the regular constraints and $(\lowerBoundShort{i}, \upperBoundShort{i})$ are pairs of numbers defining the length constraints. 

We begin with a simple observation: in time linear in the size of the gapped sequence-part of the input (that is, in $O(|p|+\size(\gaptuple))$ time), we can identify  the non-zero constraints of $\gaptuple$. Let $i_1, \ldots, i_{k-1}\in [m-1]$ be numbers such that $C_i\neq \{\emptyword\}$, for all $i\in \{i_1, \ldots, i_{k-1}\}$, and $C_i=\{\emptyword\}$, for all $i\notin \{i_1, \ldots, i_{k-1}\}$. Clearly, $\states(\gaptuple)=\sum_{j=1}^{k-1}|Q_{i_j}|$ and $\states(\gaptuple)\geq k-1$. 

Moreover, with $i_0=0$ and $i_k=m$, we define the words $p_j=p[i_{j-1}+1..i_j]$, for $j\in [k]$. For $i\in [k]$, we denote $n_i=|p_1\cdots p_i|$ and $n_0=0$. From an algorithmic point of view, one can preprocess the word $p$ in linear time to compute the non-empty words $p_1,\ldots,p_k$. 

Further, we can construct in linear time the suffix array and the longest common extension (also known as longest common prefix) data structures for the word $x=wp$. These allow us to check in constant time whether $w[j+1..j+|p_i|]=p_i$, for all $j$ and $i$, by simply checking in $O(1)$ time whether the longest common prefix of $x[j+1..n+m]$ and $x[n+|p_{i-1}|..n+m]$ has at least length $|p_i|$.

All the steps described above are part of a {\bf preprocessing part} of our algorithm.

{\bf The main part of our algorithm} consists in a dynamic programming approach. 

We define and aim to compute a two dimensional $n\times k$ array $D[\cdot][\cdot]$, where $D[i][\ell]=1$ if and only if $p[1..n_\ell]$ can be embedded in $w[1..i]$ and this embedding satisfies the first $\ell-1$ non-zero constraints $C_{i_1}, \ldots, C_{i_{\ell-1}}$ of $\gaptuple$ and maps $p_\ell$ to the suffix of length $|p_\ell|$ of $w[1..i]$. Otherwise, $D[i][\ell]=0$. 

To compute the elements of the matrix $D$, we proceed as follows. 

Firstly, we initialize all the elements of $D$ with $0$. We then set $D[i][1]=1$ if and only if $w[1..i]$ ends with $p_1$, i.\,e., $w[i-|p_1|+1..i]=p_1$.

Further, assume that, for some $t\in [k-1]$, we have computed $D[\cdot][j]$, for all $j\leq t$, and we want to compute $D[\cdot][t+1]$. This part is the most involved part of our algorithm and its aim is computing an array $f_{t+1}[\cdot]$, with $n$ elements, such that $f_{t+1}[i]=1$ if and only if there exists a position $j$ such that $D[j][t]=1$, $w[j+1..i]\in L(A_{t})$, and $\lowerBoundShort{t} \leq |w[j+1..i]|\leq \upperBoundShort{t}$. 

{\bf We now present a procedure for the efficient computation of $f_{t+1}[\cdot]$.} This procedure consists in several generic steps. 

{\em The first generic step} of this procedure is the following.

We first collect in a list $L_{t+1}$, in increasing order, all the positions $i$ of $w$ such that $D[i][t]=1$. Let $j_1<\ldots <j_r$ be the elements of $L_{t+1}$. 

Then, we compute a graph $G_{t+1}$ and a two dimensional $n\times |Q_{t}|$ array $M_{t+1}[\cdot][\cdot]$ as follows. Initially, all elements of $M_{t+1}$ are set to $0$. We then set $M_{t+1}[j][q_{0,t}]=1$ for all $j\in L_{t+1}$. Now, we compute $G_{t+1}$: its nodes are pairs $(i,q)$ with $i\in [n]$ and $q\in Q_t$. Intuitively, the nodes of $G_{t+1}$ will be exactly those pairs $(i,q)$ for which $M[i][q]=1$, and the edges are $[(i,q),(i+1),q']$ where $\delta_t(q,w[i+1])=q'$. The construction of $G_{t+1}$ is immediate: for $i$ from $j_1$ to $n$, and for $q\in Q_t$, if $M_{t+1}[i][q]=1$ and  $\delta_t(q,w[i+1])=q'$, we add the edge $[(i,q),(i+1),q']$ to $G_{t+1}$ and set $M_{t+1}[i+1][q']=1$. 

Intuitively, $G_{t+1}$ consists in the union, over $j\in L_{t+1}$, of the (not necessarily disjoint) paths $[(j,q_{0,t}), (j+1,q^j_1),\ldots, (n,q^j_{n-j})]$, where $\delta_t(q_{0,t},w[j+1])=q^j_1$  and $\delta_t(q^j_r,w[j+r+1])=q^j_{r+1}$, for all $r\in [n-j-1]$. Intuitively, such a path records the trace of the computation of $A_t$ on the input $w[j+1..n]$. An important observation is that if two such paths $[(j,q_{0,t}), (j+1,q^{j}_1),\ldots, (n,q^{j}_{n-j})]$ and $[(j',q_{0,t}), (j'+1,q^{j'}_1),\ldots, (n,q^{j'}_{n-{j'}})]$ intersect, then they stay identical after their first common node; this is, indeed, true because $A_t$ is a deterministic finite automaton. Consequently, $G_{t+1}$ is a collection of disjoint trees (i.\,e., a forest) $T_1, T_2, \ldots, T_z$. As there are no edges between any pair of nodes $(n,q)$ and $(n,q')$, with $q,q'\in Q_t$, we can compute these trees by a series of depth first searches starting in the nodes $(n,q)$, with $q\in Q_t$, of $G_{t+1}$. Clearly, when computing these nodes, we can associate to each node $(j,q)$ the label $\ell_{j,q}$ if and only if $(j,q)$ is contained in the tree $T_{\ell_{j,q}}$. In other words, the label allows us to check quickly which tree contains each node of $G_t$.

From now on, each of the trees $T_i$ will be seen as a rooted tree, whose root is the single node of the form $(n,q)$ contained in that tree; the root of $T_i$, for $i\in [z]$, is denoted $(n,q_i)$ and the leaves of each of these trees are of the form $(j,q_{0,t})$, with $j\in L_{t+1}$. 

So, at the end of the first main step, we have obtained the rooted trees $T_1, \ldots T_z$. For simplicity, we denote by $|T_i|$ the number of nodes in $T_i$ (the size of $T_i$). Clearly, $p\leq |Q_t|$ and $\sum_{i=1}^p|T_i|\leq n|Q_t|$.

{\em Sketch of a simple but less efficient variant:} At this point, we could use a rather direct approach to compute $f_{t+1}[\cdot]$. We give an informal sketch of this idea, without too many implementation details as, anyway, we have a more efficient variant below. The role of this short interlude is to explain the need for a heavier data structures machinery in the efficient variant. The reader not interested in this simpler but less efficient variant, can skip directly to the description of the second generic step of our algorithm. 

In general, we want to identify each node $(j,q)$ (of a tree $T_i$) in our forest for which there exists a path ${\mathcal P}$ of length $\ell$, with $\lowerBoundShort{t}\leq \ell\leq \upperBoundShort{t}$, and a leaf $(j',q_{0,t})$ of $T_i$ such that the path ${\mathcal P}$ starts with $(j',q_{0,t})$ and ends with $(j,q)$. Or, in other words, $\delta_t(q_{0,t}, w[j'+1..j])=q$ and $\lowerBoundShort{t}\leq |w[j'+1..j]|\leq \upperBoundShort{t}$. For such a node, we can set $f_{t+1}[j]=1$, as there exists a word $w[j'+1..j]$ of length $\ell$, with $\lowerBoundShort{t}\leq \ell\leq \upperBoundShort{t}$, such that $j'\in L_{t+1}$ and $\delta_t(q_{0,t}, w[j'+1..j])$ is a final state. 

So, we will consider $j$ from $1$ to $n$ in increasing order and, while doing this, we maintain a collection of disjoint sets $S_q$, with $q\in Q_t$, included in $\{1,\ldots,n\}$.  Initially, all the sets in the collection are empty. When considering $j=i$, we construct the sets $S'_{q}=\cup_{q'\in Q_t, \delta(q',w[i])=q} S_{q'}$, for $q\in Q_{t}$; afterwards, we insert $i$ in $S'_{q_{0,t}}$ if and only if $i\in L_{t+1}$. Then, we simply set $S_q\gets S'_q$, for all $q\in Q_{t}$. Intuitively, $S_q$ contains those positions $j'\in\{1,\ldots,i\}$ such that there is a path between $(j',q_0)$ and $(i,q)$; note now that these sets are pairwise disjoint because the automaton $A_t$ is deterministic. Therefore, we set $f_{t+1}[i]=1$ if and only if there exists a final state $f\in F_t$ such that $S_f$ contains a position $j'$ with $\lowerBoundShort{t}\leq |w[j'+1..j]|\leq \upperBoundShort{t}$. A careful implementation of the sets $S_q$ from our collection (e.\,g., based on AVL-trees, which allows us fast searching, and an implementation of tree-union in which simply the elements of the smaller tree are inserted, one by one, in the larger tree) would lead to a computation of $f_{t+1}[\cdot]$ in $O(n |Q_{t}| \poly\!\log(n))$. In the following, we present a more efficient variant for the computation of $f_{t+1}[\cdot]$. 

{\em In the second generic step of our efficient procedure}, we construct {\em level ancestor data structures} for the trees $T_i$, with $i\in [p]$.  

More precisely, the {\em Level Ancestor Problem} is defined as follows (see \cite{BenderFarachTCS}). In a rooted tree $T$, $v$ is an ancestor of $u$ if the shortest path (i.\,e., the only simple path) from the root to $u$ goes through $v$; the depth of a node $u$, denoted $\depth(u)$, is the number of edges on the shortest path from $u$ to the root of $T$. 
For a rooted tree $T$, let $\LA_T(u,d)=v$, where $v$ is an ancestor of $u$  and $\depth(v)=d$, if such a node exists, or $\uparrow$ otherwise. The {\em Level Ancestor Problem} consists in a preprocessing phase and a querying phase:
\begin{itemize}
	\item Preprocessing: A rooted tree $T$ with $N$ vertices.
	\item Querying: For a node $u$ in the rooted tree $T$, query $\levelAncestor_{T} (u,d)$ returns $\LA_{T} (u,d)$, if it exists, and false otherwise. 
\end{itemize}
A simple and elegant solution for this problem which has $O(N)$ preprocessing time and $O(1)$ time for query-answering can be found in, e.\,g., \cite{BenderFarachTCS} (see also \cite{BenAmram} for a more involved discussion). 

So, for each tree $T_i$ we can compute in $O(|T_i|)$ time data structures allowing us to answer $\levelAncestor_{T_i}$ queries in $O(1)$ time. 

{\em In the third generic step}, we have a procedure of marking nodes in the trees $T_i$, for $i\in [z]$. So, let us consider one of these trees $T_i$ and explain how we process it. In this step, the marking can be maintained, e.\,g., using a boolean $n\times |Q_t|$ matrix which simply keeps track for each node whether it was marked or not; we do not go in further details with this, at it would simply make the exposition more heavy. 

In this step we would like to mark, for any leaf $(j,q_{0,i})$ of $T_i$, all the ancestors $(d,q)$ of $(j,q_{0,i})$ such that $\lowerBoundShort{t}\leq |w[j+1..d]|=d-j\leq \upperBoundShort{t}$.

At a high level, this is done as follows. For each such leaf $(j,q_{0,t})$ of $T_i$, we compute $(j+\lowerBoundShort{t},q_{-}^j)=\levelAncestor_{T_i}((j,q_{0,t}), n - j- \lowerBoundShort{t} )$. We mark this node and all its ancestors while going upwards on the path from $(j+\lowerBoundShort{t},q{-}^j)$ towards the root of $T_i$, until we reach a node $(j+\upperBoundShort{t},q^j_{+})$. We also mark $(j+\upperBoundShort{t},q^j_{+})$ and then we run the marking procedure for the next leaf of $T_i$. 

To avoid marking the same nodes multiple times, we can implement the marking procedure as follows.

By going through the elements of $L_{t+1}$ (which were ordered increasingly), we can produce for each node $T_i$ the stack $L_{t+1,i}$ containing all leaves of $T_i$, ordered decreasingly (top to bottom) w.\,r.\,t. the first component of each node $(j,q)$. Note that the sets/stacks $L_{t+1,i}$, for $i\in [z]$, define a partition of $L_{t+1}$. 

Assume, now, that the leaves of $T_i$ are $(j_{1,i},q_{0,t})$, $(j_{2,i},q_{0,t}), \ldots$, $(j_{s_i,i},q_{0,t})$, with  $j_{1,i}> j_{2,i}> \ldots > j_{s_i,i}$ (and note that $s_i$ is the number of leaves of $T_i$). 

Now, for $g$ from $1$ to $s_i$, we compute $(j_{g,i}+\lowerBoundShort{t},q_{-}^g)=\levelAncestor_{T_i}((j_{g,i},q_{0,t}), n - j_{g,i}- \lowerBoundShort{t} )$. We mark this node and all its ancestors while going upwards on the path from $(j_{g,i}+\lowerBoundShort{t},q_{-}^g)$ towards the root of $T_i$, until we reach a node $(j_{g,i}+\upperBoundShort{t},q_{+}^g)$ or we meet an already marked node. If we reach $(j_{g,i}+\upperBoundShort{t},q^g_{+})$, we also mark $(j_{g,i}+\upperBoundShort{t},q^g_{+})$, and then we continue with the marking procedure for the next leaf of $T_i$; otherwise, if we meet a marked node, we stop the marking for this leaf and continue with the marking procedure for the next leaf of $T_i$. 

To understand why it is correct to stop the marking for a leaf once we meet a marked node $(c,q)$, it is enough to note that this marked node must have been marked when another leaf $(j_{f,i},q_{0,t})$ with a greater first component (i.\,e., $j_{f,i}>j_{g,i}$) was considered. In that case, the path of marked nodes between $(j_{f,i}+\lowerBoundShort{t},q_{-}^f)$ and $(j_{f,i}+\upperBoundShort{t},q_{+}^f)$, ends with the node $(j_{f,i}+\upperBoundShort{t},q_{+}^f)$ for which we have $j_{f,i}+\upperBoundShort{t}>j_{g,i}+\upperBoundShort{t}$. Thus, the node $(j_{g,i}+\upperBoundShort{t},q_{+}^g)$ is actually on the path from $(c,q)$ to $(j_{f,i}+\upperBoundShort{t},q_{+}^f)$, so every node between $(c,q)$ and $(j_{g,i}+\upperBoundShort{t},q_{+}^g)$ is already marked and we do not need to mark them again.

Once we have completed the marking for all trees $T_i$, with $i\in [z]$, we have achieved the following: a node $(j,q)$ is marked if and only if there exists a path ${\mathcal P}$ of length $\ell$, with $\lowerBoundShort{t}\leq \ell\leq \upperBoundShort{t}$, and a leaf $(j',q_{0,t})$ of $T_i$ such that the path ${\mathcal P}$ starts with $(j',q_{0,t})$ and ends with $(j,q)$. Or, in other words, $\delta_t(q_{0,t}, w[j'+1..j])=q$ and $\lowerBoundShort{t}\leq |w[j'+1..j]|\leq \upperBoundShort{t}$. 

{\em In the fourth, and final, generic step}, we simply set, for $i$ from $1$ to $n$, $f_{t+1}[i]=1$ if and only if there exists a state $q\in F_t$ such that the node $(i,q)$ is marked. This means that $f_{t+1}[i]=1$ if and only if there exists a word $w[j+1..i]$ of length $\ell$, with $\lowerBoundShort{t}\leq \ell\leq \upperBoundShort{t}$, such that $j\in L_{t+1}$ and $\delta_t(q_{0,t}, w[j+1..i])$ is a final state (i.\,e., $w[j+1..i]\in C_{t}$). 

Clearly, $f_{t+1}[\cdot]$ is correctly computed with respect to the definition given when we introduced this array. 
  
{\bf After we compute the array $f_{t+1}[\cdot]$, we return to the main algorithm}, and move on to the computation of the elements $D[\cdot][t+1]$ from our main array $D$. 

We simply set $D[i][t+1]=1$ if and only if $w[i-|p_{t+1}|+1..i]=p_{t+1}$ and $f_{t+1}[i-|p_{t+1}|]=1$. It is not hard to see that $D[\cdot][t+1]$ is correctly computed.

{\bf To conclude, in the final step of our algorithm}, we decide that the input gapped sequence $(p, \gaptuple)$ matches the word $w$ if and only if there exists $j$ such that $D[j][k]=1$. 

We now discuss the {\bf complexity of our approach}. The preprocessing part is done in $O(n + \size(\gaptuple) )$ time. The initialization of the array $D[\cdot][\cdot]$ and the computation of its first column $D[\cdot][1]$ can be done in linear time $O(n)$. Now, for each $t\geq 1$, the computation of $f_{t+1}$ takes $O(n|Q_t|)$ time. Indeed, the first generic step of this computation takes  $O(n|Q_t|)$ time. The second generic step takes $O(\sum_{i=1}^z|T_i|)=O(n|Q_t|)$ time. The third generic step takes time proportional to the number of marked nodes. As each node of each tree is marked at most once, then the time needed to complete this step is also $O(\sum_{i=1}^z|T_i|)=O(n|Q_t|)$. Finally, the fourth generic step can be performed in $O(n|F_t|)$ time. So, overall, the computation of $f_{t+1}[\cdot]$ takes, as claimed, $O(n|Q_t|)$ time. Therefore, computing all the elements of $D[\cdot][t+1]$ can be done in $O(n|Q_t|)$ time (including here the computation of $f_{t+1}[\cdot]$). Consequently, all the elements of $D[\cdot][\cdot]$ can be computed in $O(\sum_{i=1}^{k-1} n|Q_i|)=O(n\states(\gaptuple))$ time. Finally, deciding whether the input gapped sequence $(p, \gaptuple)$ matches the word $w$ based on the array $D[\cdot][\cdot]$ can be done in $O(n)$ time. Therefore, the considered matching problem can be solved in $O(|w|\states(\gaptuple) + |w| + \size(\gaptuple) )=O(|w|\states(\gaptuple) + \size(\gaptuple) )$.
\end{proof}

Regarding Corollary \ref{constantPatternsRL}, as each length constraint can be represented as reg-len constraints by adding the regular constraint $\Sigma^*$ (specified as DFA with one state), statement (1) follows immediately. Statement (2) follows from Theorem \ref{constantPatternsRegularLength}.

However, let us now comment more on how the case of $\matchProb$ with either regular constraints or length constraints can be handled.

In the case of length constraints only, we use the same approach and, just like in the proof above, assume that, for some $t\in [k-1]$, we have computed $D[\cdot][j]$, for all $j\leq t$, and we want to compute $D[\cdot][t+1]$. We again compute an array $f_{t+1}[\cdot]$, with $n$ elements, such that $f_{t+1}[i]=1$ if and only if there exists a a position $j$ such that $D[j][t]=1$ and $\lowerBoundShort{t} \leq |w[j+1..i]|\leq \upperBoundShort{t}$. This can be done in a simple way: we go through the positions $i$ from $1$ to $n$ and maintain in a dequeue, in increasing order, the positions $j$ such that $D[j][t]=1$ and $\lowerBoundShort{t} \leq |w[j+1..i]|\leq \upperBoundShort{t}$. For each $i$, we need to simply check if the dequeue is empty or not; if not, then we can set $f_{t+1}[i]=1$ (and otherwise leave $f_{t+1}[i]=0$). We then continue as in the algorithm above, after $f_{t+1}$ is computed. We obtain the time complexity stated in Corollary \ref{constantPatternsRL}(1).

In the case of regular constraints only, we use the same general approach and, once more, assume that, for some $t\in [k-1]$, we have computed $D[\cdot][j]$, for all $j\leq t$, and we want to compute $D[\cdot][t+1]$. We again compute an array $f_{t+1}[\cdot]$, with $n$ elements, such that $f_{t+1}[i]=1$ if and only if there exists a a position $j$ such that $D[j][t]=1$ and $w[j+1..i]\in L(A_{t})$. 
This can be done in a simple way: we go through the positions $i$ from $1$ to $n$ and maintain an array $S[\cdot]$ of size $|Q_t|$, where $S[q]=\min\{j\in L_{t+1}\mid \delta(q_{0,t},w[j+1..i])=q\}$ (or $S[q]=\infty$ if the respective set is emtpy). Clearly, if we have computed $S[\cdot]$ for position $i-1$, we can immediately update it for position $i$. Then, once $S$ is updated for position $i$, we need to simply check if there exists $f\in F_{t}$ such that $S[f]\neq \infty$; if yes, then we can set $f_{t+1}[i]=1$ (and otherwise leave $f_{t+1}[i]=0$). We then continue as in the algorithm above, after $f_{t+1}$ is computed. We obtain trivially the time complexity stated in Corollary \ref{constantPatternsRL}(2). In the case when the regular constraints are given as NFAs (or regexes) instead of DFAs, the computation of $f_{t+1}[\cdot]$ can be done in time $O(n |A_t|)$, where $|A_t|$ is the number of transitions of the automaton $A_t$ (or the size of the regex representing $L(A_t)$). 

It is worth noting that none of these simple approaches can be extended to work for the case when regular and length constraints are combined, and, in fact, our approach for that case proposes a data structure that creates the environment in which the main ideas in these simple approaches can interact efficiently. 

\subsection{Proof of Theorem~\ref{lowerBoundLength}}\label{sec:MatchingLowerBound}

Before stating the proof, note that the reduction is presented in a slightly different way compared to the proof sketch from the main part of this paper. More precisely, what is called $\code_b (\vec{b}_i)$ in the following corresponds to the part $\codeSketch_b (\vec{b}_i) \stackrel{\leq 1}{\leftrightarrow} \# \stackrel{\leq 3}{\leftrightarrow} \#  \stackrel{\leq 1}{\leftrightarrow} \# \stackrel{\leq 3}{\leftrightarrow} \#$ of the proof sketch, while the part $\codeSketch_b (\vec{b}_i)$ from the proof sketch is defined as $u_i$ later in the proof.

\begin{proof}
We reduce $\OV$ to the matching problem for gapped sequences with length constraints. We consider an instance of $\OV$: $A=\{\vec{a}_1, \ldots, \vec{a}_n\}$ and $B=\{\vec{b}_1, \ldots, \vec{b}_n\}$, with $A,B \subset \{0,1\}^d$; we can assume that $d\geq 2$. We transform this $\OV$-instance into an instance of the matching problem for gapped sequences with length constraints. We need to define a word $w$ (which, intuitively, corresponds to the set $A$), and a gapped sequence $(p, \gaptuple)$ with length constraints (which corresponds to the set $B$). Interestingly, the gapped sequence $(p, \gaptuple)$ will additionally fulfil the property that $\gaptuple[i]=(0,k_i)$, where $0\leq k_i \leq 6$, for all $i\in [|p|-1]$. We will show that $p \subseq_{\gaptuple} w$ if and only if there exist two vectors $\vec{a}_i$ and $\vec{b}_j$ which are orthogonal. 

To simplify the exposition, when representing the gapped sequence $(p, \gaptuple)$ with $p=p[1]\cdots p[m]$, we will use the notation $(p,\gaptuple) = p[1] \stackrel{\gaptuple[1]}{\leftrightarrow} p[2] \stackrel{\gaptuple[2]}{\leftrightarrow} \cdots \stackrel{\gaptuple[m-1]}{\leftrightarrow} p[m]$. Moreover, we omit the symbol $\stackrel{\gaptuple[i]}{\leftrightarrow}$ from the notation of $(p,\gaptuple)$ if and only if $\gaptuple[i]=(0,0)$. If $\gaptuple[i]=(0,k)$, then we write $\stackrel{\leq k}{\leftrightarrow}$ instead of the symbol $\stackrel{\gaptuple[i]}{\leftrightarrow}=\stackrel{(0,k)}{\leftrightarrow}$ when writing the notation for $(p,\gaptuple)$. 
For instance, if $p=abab$ and $\gaptuple[1]=(0,0)$, $\gaptuple[2]=(1,5)$, and $\gaptuple[2]=(0,6)$, we denote $(p,\gaptuple)=ab \stackrel{(1,5)}{\leftrightarrow}  a\stackrel{\leq 6}{\leftrightarrow} b$. 

So, let us now {\bf define our reduction}. Once more, we start with two sets $A=\{\vec{a}_1, \ldots, \vec{a}_n\}$ and $B=\{\vec{b}_1, \ldots, \vec{b}_n\}$, with $A,B \subset \{0,1\}^d$. 
We assume $\vec{a}_i=(a_i^1,\ldots, a_i^d)$ and $\vec{b}_i=(b_i^1,\ldots, b_i^d)$, for all $i\in [n]$. 
As usually, we define a series of gadgets, which will be combined to produce the word $w$ and the gapped sequence $(p,\gaptuple)$. The alphabet over which we define the words $w$ and $p$ is $\Sigma = \{0,1,\#,@\}$

{\bf The first set of gadgets} is defined as follows: 
\begin{itemize}
\item $\code_a (0)=010$, $\code_a(1)=100$;
\item $\code_b(0)=10$, $\code_b(1)=01$.
\end{itemize} 

{\bf The second set of gadgets} is obtained based on the first set defined above:
\begin{itemize}
\item For $i\in [n]$, $\code_a (\vec{a}_i)$ is the word\\
$\left(\prod\limits_{j=1}^d (\#\code_a(0) \# \#\code_a(a_i^j) \# \#\code_a(0) \#)\right).$
\item For $i\in [n-1]$, $\code_b (\vec{b}_i)$ is the gapped sequence defined as \\
$\left(\prod\limits_{j=1}^{d-1} (\# \stackrel{\leq 1}{\leftrightarrow} \code_b(b_i^j) \stackrel{\leq 1}{\leftrightarrow} \# \# \stackrel{\leq 3}{\leftrightarrow} \# \# \stackrel{\leq 3}{\leftrightarrow} \#)\right)\# \stackrel{\leq 1}{\leftrightarrow} \code_b(b_i^d) \stackrel{\leq 1}{\leftrightarrow} \#  \stackrel{\leq 1}{\leftrightarrow} \# \stackrel{\leq 3}{\leftrightarrow} \#  \stackrel{\leq 1}{\leftrightarrow} \# \stackrel{\leq 3}{\leftrightarrow} \#.$
\item For $i=n$, $\code_b (\vec{b}_n)$ is the gapped sequence defined as \\
$\left(\prod\limits_{j=1}^{d-1} (\# \stackrel{\leq 1}{\leftrightarrow} \code_b(b_n^j) \stackrel{\leq 1}{\leftrightarrow} \# \# \stackrel{\leq 3}{\leftrightarrow} \# \# \stackrel{\leq 3}{\leftrightarrow} \#)\right)\# \stackrel{\leq 1}{\leftrightarrow} \code_b(b_n^d) \stackrel{\leq 1}{\leftrightarrow} \# \stackrel{\leq 5}{\leftrightarrow}.$
\end{itemize} 
We say that each $\code_a (\vec{a}_i)$ contains three tracks. The first track of the concatenation 

\begin{equation*}
\left(\prod\limits_{j=1}^d (\#\code_a(0) \# \#\code_a(a_i^j) \# \#\code_a(0) \#)\right) 
\end{equation*}

consists in the prefix $\#\code_a(0) \#$ of each factor of the concatenation defining $\code_a (\vec{a}_i)$. The second track consists in the factors $\#\code_a(a_i^j) \#$ occurring in the middle of each factor in the concatenation above. Finally, the third track consists in the factors $\#\code_a(0) \#$ occurring as a suffix of each factor of the concatenation defining $\code_a (\vec{a}_i)$. To ease the understanding, we can highlight these tracks in the concatenation by placing the factors defining the $i^{th}$ track between brackets $[_i \cdots ]_i$: 

\begin{equation*}
\code_a (\vec{a}_i)=\left(\prod\limits_{j=1}^d ([_1\#\code_a(0) \#]_1 [_2\#\code_a(a_i^j) \#]_2[_3 \#\code_a(0) \#]_3)\right)\,. 
\end{equation*}

Clearly, just like the case of usual parentheses, these indexed brackets are not part of the string $\code_a (\vec{a}_i)$. 

{\bf The final set of gadgets}, corresponding to the word $w$ and the gapped sequence $(p,\gaptuple)$ are defined as follows.
\begin{itemize}
\item $w= \left (\prod\limits_{i=1}^{n-1} @ \code_a(\vec{a}_i) \right) @ \code_a(\vec{a}_n)\left (\prod\limits_{i=1}^{n-1} @ \code_a(\vec{a}_i) \right) @$.
\item $(p,\gaptuple) = @\stackrel{\leq 5}{\leftrightarrow} \left( \prod\limits_{j=1}^{n-1} \code_b(\vec{b}_j) \stackrel{\leq 6}{\leftrightarrow} \right) \code_b(\vec{b}_n) @$.
\end{itemize}

We continue with the {\bf correctness proof} for our reduction, i.\,e., the proof of the claim that the instance of $\OV$ defined by $A$ and $B$ contains two orthogonal vectors $\vec{a}_i$ and $\vec{b}_j$ if and only if $p\subseq_{\gaptuple} w$. Note that $\vec{a}_i$ and $\vec{b}_j$ are orthogonal is equivalent to $a_i^f b_j^f=0,$ for $f\in [d]$. 

{\bf We first note} that, for $g,h\in \{0,1\}$, we have $g\cdot h = 0$ if and only if $\code_b(h)$ is a factor of $\code_a(g)$. Indeed $\code_b(0)=10$ is a factor of $\code_a(0)=010$, $\code_b(1)=01$ is a factor of $\code_a(1)=100$, $\code_b(1)=01$ is a factor of $\code_a(0)=010$, but $\code_b(1)=01$ is not a factor of $\code_a(1)=100$. See Figure \ref{fig:basicgadgets}.
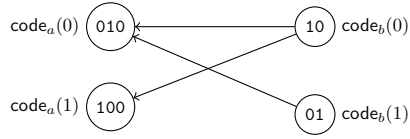
\begin{figure}[h!] 
    \centering
    \begin{tikzpicture}[scale=0.7, transform shape]
        \node[draw, circle, label=left:{$\code_a(0)$}] (A0) at (0,0) {$\mathtt{010}$};
        \node[draw, circle, below=6mm of A0, label=left:{$\code_a(1)$}] (A1) {$\mathtt{100}$};
        \node[draw, circle, right=30mm of A0, label=right:{$\code_b(0)$}] (B0) {$\mathtt{10}$};
        \node[draw, circle, below=9mm of B0, label=right:{$\code_b(1)$}] (B1) {$\mathtt{01}$};

        \draw[<-] (A0) -- node[below]{} (B0);
        \draw[<-] (A0) -- node[left, xshift=-5mm, yshift=-1mm]{} (B1);
        \draw[<-] (A1) -- node[right, xshift=5mm, yshift=-1mm]{} (B0);
        
    \end{tikzpicture}
    \caption{Gadgets for the encoding of single bits. The edges represent ``factor of'' relations.}
    \label{fig:basicgadgets}
\end{figure}

{\bf Secondly}, for $i\in [n]$, we define the gapped  sequence\\
$(u_i,\gaptuple_{i})=\left(\prod\limits_{j=1}^{d-1} (\# \stackrel{\leq 1}{\leftrightarrow} \code_b(b_i^j) \stackrel{\leq 1}{\leftrightarrow} \# \# \stackrel{\leq 3}{\leftrightarrow} \# \# \stackrel{\leq 3}{\leftrightarrow} \#)\right)\# \stackrel{\leq 1}{\leftrightarrow} \code_b(b_i^d) \stackrel{\leq 1}{\leftrightarrow} \#  $. \\
Please note that $\gaptuple_i$ can be seen as the restriction of $\gaptuple$ on the factor $u_i$ of $p$. 

Let us now consider some $i, \ell\in [n]$. We analyse the possible embeddings $e$ of $u_i$ in the word $\code_a(\vec{a}_\ell)$, which satisfy $\gaptuple_i$.  Note that $\code_a(\vec{a}_\ell)$ contains $6d$ $\#$-symbols, while $u_i$ contains $6d-4$  $\#$-symbols. Thus, in an embedding of $u_i$ in $\code_a(\vec{a}_\ell)$, the first three $\#$-symbols of $u_i$ must be aligned to some of the first seven $\#$-symbols of $\code_a(\vec{a}_\ell)$. By the fact that, in an embedding $e$ which satisfies $\gaptuple_i$, we have to embed the two symbols of $\{0,1\}$ occurring between the first two $\#$-symbols of $u_i$ (i.\,e., the string $\code_b(b_i^1) $) into a factor of $\code_a(\vec{a}_\ell)$ which ends before the seventh $\#$ of $\code_a(\vec{a}_\ell)$, and also because the second and third $\#$-symbols of $u_i$ have no gap between them, we get that the symbols $\code_b(b_i^1)$ are embedded as a factor of either the first factor $\code_a(0)$ of the first track, or the first factor $\code_a(0)$ of the third track of $\code_a(\vec{a}_\ell)$, or of the factor  $\code_a(a_\ell^1)$ of the second track of $\code_a(\vec{a}_\ell)$. Now, if $\code_b(b_i^1)$ is embedded, as just described, in the factor corresponding to track $q \in \{1, 2, 3\}$ of $\code_a(\vec{a}_\ell)$, then, by the fact that there are exactly six $\#$ symbols between $\code_b(b_i^j)$ and $\code_b(b_i^{j+1})$, we get that all the factors $\code_b(b_i^j)$ will be embedded in the corresponding factors of $\code_a(\vec{a}_\ell)$ of the same track $q$. Such embeddings are always possible when $\gaptuple_i$ is satisfied, and no other embeddings with these properties are possible. So, in an embedding of $u_i$ in the word $\code_a(\vec{a}_\ell)$, the way we embed $\code_b(b_i^1)$ selects a track of $\code_a(\vec{a}_\ell)$, in which all the factors $\code_b(b_i^j)$ will be embedded. We say that $e$ selects track $q$ of $\code_a(\vec{a}_\ell)$ if $\code_b(b_i^1)$ is embedded in track $q$ of $\code_a(\vec{a}_\ell)$. 

We now claim that $\vec{a}_\ell$ and $\vec{b}_i$ are orthogonal if and only if there exists an embedding $e$ which satisfies $\gaptuple_i$ and, for all $j \in [d]$, maps the factor $ \code_b(b_i^j)$ of $u_i$ to a factor of $ \code_a(a_\ell^j)$ of $\code_a(\vec{a}_\ell)$. In other words, $\vec{a}_\ell$ and $\vec{b}_i$ are orthogonal if and only if there is an embedding $e$ which satisfies $\gaptuple_i$ and that selects the second track of $\code_a(\vec{a}_\ell)$. Clearly, if there exists an embedding $e$ which satisfies $\gaptuple_i$ and maps the factor $ \code_b(b_i^j)$ of $u_i$ to a factor of $ \code_a(a_\ell^j)$ of $\code_a(\vec{a}_\ell)$, for all $j \in [d]$, respectively, then $b_i^j a_{\ell}^j=0$, for all $j \in [d]$. So, $a_\ell$ and $b_i$ are orthogonal. 
For the other implication, we proceed as follows. If $\vec{a}_\ell$ and $\vec{b}_i$ are orthogonal, then $\code_b(b_i^j) $ is a factor of $ \code_a(a_\ell^j)$ for $j\in [d]$. As mentioned above, we can construct the embedding $e$ as follows. To begin with, for $j\in [d]$, $e$ maps the symbols of the factor $ \code_b(b_i^j)$ of $u_i$ to a factor of $\code_a(a_\ell^j)$, and the $\#$ symbol occurring before (respectively, after) $ \code_b(b_i^j)$ in $u_i$ to the $\#$ symbol occurring before (respectively, after)  $\code_a(a_\ell^j)$ in $\code_a(\vec{a}_\ell)$. Then, for $j\in [d-1]$, $e$ maps the four $\#$ symbols occurring between  $\# \code_b(b_i^j)\#$ and  $\#\code_b(b_i^{j+1})\#$ in $u_i$ to the four $\#$ symbols occurring between  $\#\code_a(a_\ell^j)\#$ and  $\#\code_a(a_\ell^{j+1})\#$ in $\code_a(\vec{a}_\ell)$. This embedding clearly satisfies the gap constraints $\gaptuple_i$. 

This proves our claim. 

Before moving on with our proof, it is worth recalling what we have shown. 
The above claim states that $\vec{a}_\ell$ and $\vec{b}_i$ are orthogonal if and only if there exists an embedding $e$ which satisfies $\gaptuple_i$ and maps the factors $ \code_b(b_i^j)$ of $u_i$ to the factors defining the second track of $\code_a(\vec{a}_\ell)$, respectively (i.\,e., $e$ selects the second track of  $\code_a(\vec{a}_\ell)$). But, as mentioned, there exist other embeddings $e$ of $u_i$ in $\code_a(\vec{a}_\ell)$, which satisfy $\gaptuple_i$, even if $a_\ell$ and $b_i$ are not orthogonal. In all these other embeddings either the first of the third track of $\code_a(\vec{a}_\ell)$ are selected. See Figure \ref{fig:fittinggadgetsAppendix}.

\begin{figure}[h!]
    \centering
  \includegraphics[width=14cm]{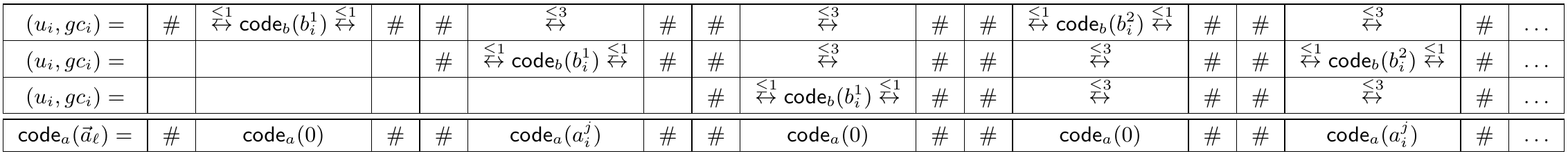}
    \caption{Possible embeddings of $u_i$ in $\code_a(\vec{a}_\ell)$, selecting the first, second, and, respectively, third track of $\code_a(\vec{a}_\ell)$.}
    \label{fig:fittinggadgetsAppendix}
\end{figure}

{\bf Thirdly}, we analyse when there exists an embedding $e$ of $p$ in $w$ which satisfies $\gaptuple$. 

Cleary, in such an embedding $e$, we map the first $@$-symbol of $p$ to an $@$-symbol of $w$. So, let us assume that the first $@$-symbol of $p$ is mapped to the $@$-symbol of $w$ occurring before an occurrence of $\code_a(\vec{a}_\ell)$. 
As $u_1$ contains $6d$ $\#$-symbols, and $\code_a(\vec{a}_\ell)$ contains $6d-4$ $\#$-symbols, and the length of the gap between $@$ and the first $\#$-symbol of $u_1$ is at most $5$, it follows that $u_1$ is then embedded in $\code_a(\vec{a}_\ell)$, and this embedding satisfies $\gaptuple_i$ (as defined above). Therefore, $u_1$ selects a track $q$ of $\code_a(\vec{a}_\ell)$. It is not hard to see that $i\leq 2$ holds, as, according to $\gaptuple$, the gap between $@$ and the first $\#$-symbol of $u_1$ is at most $5$. Then, between the suffix  $\# \stackrel{\leq 1}{\leftrightarrow} \code_b(b_1^d) \stackrel{\leq 1}{\leftrightarrow} \#$ of $(u_1,\gaptuple_1)$ and the factor $(u_2,\gaptuple_2)$, the gapped sequence $(p, \gaptuple)$ contains the factor $\stackrel{\leq 1}{\leftrightarrow} \# \stackrel{\leq 3}{\leftrightarrow} \#  \stackrel{\leq 1}{\leftrightarrow} \# \stackrel{\leq 3}{\leftrightarrow} \# \stackrel{\leq 6}{\leftrightarrow}$. This means that $u_2$ will be embedded in $\code_a(\vec{a}_{\ell+1})$, and, moreover, it will select one of the tracks $q$ or $q+1$ (as there can be at most $18$ symbols between $u_1$ and $u_2$, so track $3$ cannot be reached in the case $q=1$). See Tables \ref{tab:table1} and \ref{tab:table2}.
Now, for $i\geq 2$, as long as $u_i$ is embedded in track $q\in \{1,2\}$ of $\code_a(\vec{a}_{\ell+i-1})$, the process continues as above, and we reach the conclusion that $u_{i+1}$ must be embedded in track $q$ or $q+1$ of $\code_a(\vec{a}_{\ell+i})$. Assume now that $u_i$ is embedded in track $3$ of $\code_a(\vec{a}_{\ell+i-1})$. Then, between the factor  $\# \stackrel{\leq 1}{\leftrightarrow} \code_b(b_i^d) \stackrel{\leq 1}{\leftrightarrow} \#$ of $(u_i,\gaptuple_i)$ and the first symbol of $(u_{i+1},\gaptuple_{i+i})$, the gapped sequence $(p, \gaptuple_i)$ contains the factor $\stackrel{\leq 1}{\leftrightarrow} \# \stackrel{\leq 3}{\leftrightarrow} \#  \stackrel{\leq 1}{\leftrightarrow} \# \stackrel{\leq 3}{\leftrightarrow} \#\stackrel{\leq 6}{\leftrightarrow}$. This means that $u_{i+1}$ can be embedded either in the third track of $\code_a(\vec{a}_{\ell+i})$, or in the first track of the string $w'=\left(\prod\limits_{j=2}^d ([_1\#\code_a(0) \#]_1 [_2\#\code_a(a_{\ell+i}^j) \#]_2[_3 \#\code_a(0) \#]_3)\right)$ $@[_1\#\code_a(0) \#]_1[_2\#\code_a(a_{\ell+i+1}^1) \#]_2[_3\#\code_a(0) \#]_3$. Since $u_{i+1}$ ends with the string $\#\# \stackrel{\leq 1}{\leftrightarrow} \code_b(b_{i+1}^d) \stackrel{\leq 1}{\leftrightarrow} \#  $, and $\# \stackrel{\leq 1}{\leftrightarrow} \code_b(b_{i+1}^d) \stackrel{\leq 1}{\leftrightarrow} \#  $ must be embedded in $\# \code_a(a_{\ell+i+1}^1) \#  $, we note that this embedding of $(u_{i+1},\gaptuple_{i+1})$ is impossible: we would have to embed the factor $\#\# \stackrel{\leq 1}{\leftrightarrow} \code_b(b_{i+1}^d) \stackrel{\leq 1}{\leftrightarrow} \#  $ into the factor $@\# \code_a(a_{\ell+i+1}^1) \#  $ of $w'$, which is not possible. Thus, $u_{i+1}$ is embedded in the third track of $\code_a(\vec{a}_{\ell+i})$. \par
From now on, we can continue as above, and note that all sequences $u_j$, with $j>i+1$, will be embedded only in the third track of the corresponding factors of $w$. We conclude that in an embedding $e$ of $p$ in $w$ we simply embed each of the gapped sequences $(u_i,\gaptuple_i)$ into factors $\code_a(\vec{a}_{\ell})$ of $w$. However,  $\code_b (\vec{b}_n)$ must be embedded in such a way that the gap between the symbol of $\code_a(\vec{a}_{\ell})$ to which the last $\#$-symbol of $\code_b (\vec{b}_n)$ is mapped and an $@$ symbol is at most $5$. This means that $\code_b (\vec{b}_n)$ must have been embedded in track $2$ or $3$ of some $\code_a(\vec{a}_{\ell'})$. Therefore, at least one of $\code_b (\vec{b}_i)$ must have been embedded in track $2$ of some $\code_a(\vec{a}_{\ell''})$. Indeed, if $\code_b (\vec{b}_n)$ is embedded in track $2$ of $\code_a(\vec{a}_{\ell'})$, then this claim holds. If $\code_b (\vec{b}_n)$ is embedded in track $3$ of $\code_a(\vec{a}_{\ell'})$, then let $\code_b (\vec{b}_p)$ be such that $p$ is maximum with the property that $\code_b (\vec{b}_p)$ was not embedded in track $3$ of some $\code_a(\vec{a}_{r})$. By the arguments given above, $\code_b (\vec{b}_p)$ must have been embedded in track $2$ of  $\code_a(\vec{a}_{r})$. So, our claim that at least one of the gapped sequences $\code_b (\vec{b}_i)$, say $\code_b (\vec{b}_j)$, must be embedded in track $2$ of some $\code_a(\vec{a}_{\ell''})$ holds. 

This shows that if an embedding $e$ of $p$ in $w$, which satisfies $\gaptuple$, exists then there exist $j$ and $\ell''$ such that $\code_b (\vec{b}_j)$ is embedded in track $2$ of $\code_a(\vec{a}_{\ell''})$. That is, $\vec{b}_j$ and $\vec{a}_{\ell''}$ are orthogonal.

\begin{table}
\begin{tabular}{|c|c|c|c|c|c|c|c|c|c|}
\hline
$w=$                     & $@$ &$\varepsilon $ &  $w_1$ & $\#\code_a(a_1^d)\#$                                 &  $\#\code_a(0)\#$                                                                                                             & $@ \#\code_a(0)\#$                                      & $w_2$ \\
\hline
$(p,\gaptuple)=$ & $@$ & $\stackrel{\leq 5}{\leftrightarrow}$ & $u_1$ &  $\stackrel{\leq 1}{\leftrightarrow} \#  \stackrel{\leq 3}{\leftrightarrow} \#\stackrel{\leq 1}{\leftrightarrow}$ & $\#  \stackrel{\leq 3}{\leftrightarrow} \# $ & $\stackrel{\leq 6}{\leftrightarrow}$ & $u_2$ \\
\hline
\end{tabular}
\caption{In this example, $u_1$ is embedded in track $1$ of $\code_a(\vec{a}_1)$, and then $u_2$ is embedded in track $2$ of $\code_a(\vec{a}_2)$; we cannot embed $u_2$ in track $3$ of $\code_a(\vec{a}_2)$, as it cannot be reached. Here, $w_1$ (respectively, $w_2$) is the factor of $\code_a(\vec{a}_1)$ (respectively, $\code_a(\vec{a}_2)$)  found between the symbol to which the first symbol of $u_1$ is mapped and  the symbol to which the last symbol of $u_1$ (respectively, $u_2$) is mapped in the respective embedding.}
\label{tab:table1}
\end{table}

\begin{table}
\begin{tabular}{|c|c|c|c|c|c|c|c|c|}
\hline
$w=$                      & $@$ & $\#\code_a(0)\#$                                 & $w_1$ & $\#\code_a(0)\#$                                                                                                             & $@$                                                                               & $\#\code_a(0)\#$                                      & $\#\code_a(a_2^1)\#$                      & $w_2$ \\
\hline
$(p,\gaptuple)=$ & $@$ & $\stackrel{\leq 5}{\leftrightarrow}$ & $u_1$ &  $\stackrel{\leq 1}{\leftrightarrow} \#  \stackrel{\leq 3}{\leftrightarrow} \#$  & $\stackrel{\leq 1}{\leftrightarrow}$ & $\#  \stackrel{\leq 3}{\leftrightarrow} \# $ & $\stackrel{\leq 6}{\leftrightarrow}$ & $u_2$ \\
\hline
\end{tabular}
\caption{In this example, $u_1$ is embedded in track $2$ of $\code_a(\vec{a}_1)$, and then $u_2$ is embedded in track $3$ of $\code_a(\vec{a}_2)$. Here, $w_1$ (respectively, $w_2$) is the factor of $\code_a(\vec{a}_1)$ (respectively, $\code_a(\vec{a}_2)$)  found between the symbol to which the first symbol of $u_1$ is mapped and  the symbol to which the last symbol of $u_1$ (respectively, $u_2$) is mapped in the respective embedding.}
\label{tab:table2}
\end{table}

\begin{table}
\begin{tabular}{|c|c|c|c|c|c|c|c|c|}
\hline
 $@$                                                                               & $\# \code_a(0)\# $                                      & $\#\code_a(a_2^1)\#$                      & $w_2$ &  $@$                                                                               & $\#\code_a(0)\#$                                      & $\#\code_a(a_3^1)\#$                      &  $\emptyword$ & $w_3$ \\
\hline
  $\stackrel{\leq 1}{\leftrightarrow}$ & $\#  \stackrel{\leq 3}{\leftrightarrow} \# $ & $\stackrel{\leq 6}{\leftrightarrow}$ 
& $u_2$ & $\stackrel{\leq 1}{\leftrightarrow}$ & $\# \stackrel{\leq 3}{\leftrightarrow} \#$ & $\stackrel{\leq 1}{\leftrightarrow}\# \stackrel{\leq 3}{\leftrightarrow} \#$ &$\stackrel{\leq 6}{\leftrightarrow} $  & $u_3$ \\
\hline
\end{tabular}
\caption{Following the example in Table \ref{tab:table1}, $u_2$ is embedded in track $3$ of $\code_a(\vec{a}_2)$, and then $u_3$ is embedded in track $3$ of $\code_a(\vec{a}_3)$. Here, $w_2$ (respectively, $w_3$) is the factor of $\code_a(\vec{a}_2)$ (respectively, $\code_a(\vec{a}_3)$)  found between the symbol to which the first symbol of $u_2$ (respectively, $u_3$) is mapped and  the symbol to which the last symbol of $u_2$ (respectively, $u_3$) is mapped in the respective embedding.}
\label{tab:table3}
\end{table}

\begin{table}
\begin{tabular}{|c|c|c|c|c|c|c|c|c|}
\hline
 $@$                                                                               & $\#\code_a(0)\#$                                      & $\#\code_a(a_2^1)\#$                      & $w_2$ &  $@$                                                                               & $\#\code_a(0)\#$                                      & $\#\code_a(a_3^1)\#$                      &  $\# \code_a(0) \#$  & $w'_3$ \\
\hline
  $\stackrel{\leq 1}{\leftrightarrow}$ & $\#  \stackrel{\leq 3}{\leftrightarrow} \# $ & $\stackrel{\leq 6}{\leftrightarrow}$ 
& $u_2$ & $\stackrel{\leq 1}{\leftrightarrow}$ & $\# \stackrel{\leq 3}{\leftrightarrow} \#$ & $\stackrel{\leq 1}{\leftrightarrow}\# \stackrel{\leq 3}{\leftrightarrow} \#$ &$\stackrel{\leq 6}{\leftrightarrow} $  & $u_3$ \\
\hline
\end{tabular}
\caption{Following the example in Table \ref{tab:table1}, $u_2$ is embedded in track $3$ of $\code_a(\vec{a}_2)$, and then $u_3$ could be embedded in track $1$ of $w'_3=\left(\prod\limits_{j=2}^d ([_1\#\code_a(0) \#]_1 [_2\#\code_a(a_3^j) \#]_2[_3 \#\code_a(0) \#]_3)\right)$ $@[_1\#\code_a(0) \#]_1[_2\#\code_a(a_4^1) \#]_2[_3\#\code_a(0) \#]_3$. This is, however, impossible, as it would imply the mapping of a $\#\#$ factor of $u_3$ to a factor $\#@$ or $@\#$ of $w'_3$. Here, $w_2$ is the factor of $\code_a(\vec{a}_2)$ found between the symbol to which the first symbol of $u_2$ is mapped and  the symbol to which the last symbol of $u_2$ is mapped in the respective embedding.}
\label{tab:table4}
\end{table}

Now assume that there exist $\vec{b}_j$ and $\vec{a}_{\ell''}$ which are orthogonal. We can construct an embedding of $p$ in $w$ as follows: we embed $\code_b (\vec{b}_i)$ in track $1$ of $\code_a(\vec{a}_{\ell''-j+i})$, for $i\in [j-1]$, $\code_b (\vec{b}_j)$ in track $2$ of $\code_a(\vec{a}_{\ell''})$, and $\code_b (\vec{b}_{j+i})$ in track $3$ of $\code_a(\vec{a}_{\ell''+i})$, for $i\in [n-j]$. By our explanations, this is clearly possible.

{\bf In conclusion}, an embedding $e$ of $p$ in $w$, which satisfies $\gaptuple$, exists if and only if there exist $j$ and $\ell''$ such that $\code_b (\vec{b}_j)$ is embedded in track $2$ of $\code_a(\vec{a}_{\ell''})$. This is equivalent to saying that $\vec{b}_j$ and $\vec{a}_{\ell''}$ are orthogonal. 

This shows that our reduction is correct. The instance of $\OV$ defined by $A$ and $B$ contains two orthogonal vectors if and only if the instance of the matching problem for gapped sequences with length constraints defined by $w$ and $(p,\gaptuple)$ can be answered positively, or, equivalently, $p\subseq_{\gaptuple} w$. 
Moreover, the word $w$ and the gapped sequence $(p,\gaptuple)$ can be constructed in $O(nd)$ time and we have that $|w|, |p| \in \Theta(nd)$, and the number of bits needed to describe $\gaptuple$ is also $\Theta(nd)$. 

Assume now that there exists a solution for the matching problem for gapped sequences with length constraints running in $O(|w|^g|p|^h)$ with $g+h=2-\epsilon$ for some $\epsilon<0$. This would lead to a solution for $\OV$ running in $O(nd + (nd)^{2-\epsilon})$, a contradiction to the $\OV$-conjecture. The same argument holds for solutions running in $O(|w|^{2-\epsilon})$. 

In fact, as $\nz(\gaptuple) \in \Theta(nd)$, one can show that if there exists a solution of the considered matching problem running in $O(|w|^g\nz(\gaptuple)^h)$ with $g+h=2-\epsilon$ for some $\epsilon<0$, then there exists a solution for $\OV$ running in $O(nd + (nd)^{2-\epsilon})$, a contradiction to the $\OV$-conjecture. This proves our statement. 
\end{proof}

\subsection{Proof of Corollary~\ref{lowerBoundRegular}}

\begin{proof}
The result follows in the same way, and based on the same reduction, as Theorem \ref{lowerBoundLength}. The only observation is that   each of the length gap constraints used in that proof (which are all of the form $(0,\ell)$ for $\ell\leq 6$, so constant) can be expressed as a regular constraint, and encoded using a DFA with a constant number of states. Indeed, if we have strings over $\Sigma$ and the length constraint is $(k, \ell)$, for some constants $k$ and $\ell$ with $k \leq \ell$, then we have a DFA with $\ell + 2$ states $\{0,1,\ldots,\ell,\ell + 1\}$ with the transition $\delta(i,a)=i+1$, for all $a\in \Sigma$ and $i\in \{0,\ldots,\ell\}$, and $\delta(\ell + 1, a)=\ell+1$, for all $a\in \Sigma$. The final states of this DFA are $\{k,k + 1,\ldots, \ell\}$. Clearly, this DFA has a constant number of states. Therefore, $\states(\gaptuple_p) \in \Theta(nd)$. The result follows. 
\end{proof}

\subsection{Lower bounds for $|\Sigma|=2$}\label{sec:lowerBoundsBinary}

\begin{proof}[Proof Sketch]We can actually adapt the proof of Theorem \ref{lowerBoundLength} by simply using a block encoding of the symbols used in that reduction over a binary alphabet $\{\ta,\tb\}$. 

We rewrite the symbols of $w$ and $(p,\gaptuple)$ according to the rules $0\gets \ta^5 \tb \ta \ta \ta \tb^5$, $1\gets \ta^5 \tb \ta \tb \ta \tb^5$, $\#\gets \ta^5 \tb \tb \ta \ta \tb^5$, and $@\gets \ta^5 \tb \tb \tb \ta \tb^5$; the strings used to replace the symbols $\{0,1,\#,@\}$ are called code-blocks in the following. 
Very importantly, when rewriting the gapped sequence $(p,\gaptuple)$ by these rules, the gaps between any two consecutive symbols from $\{\ta,\tb\}$ inside the encoding of one of $\{0,1,\#,@\}$ (i.\,e., the gaps between consecutive symbols inside a code-block) are zero. 

Moreover, we adapt the non-zero gaps previously used in our reduction by multiplying all the bounds with $14$ (i.e., the length of the block codes). 

Now, the proof of Theorem \ref{lowerBoundLength} works exactly as described above, as the prefixes $\ta^5$ and $\tb^5$ used in each encoding ensure that code-blocks can only align with corresponding code-blocks in an embedding of $(p,\gaptuple)$ in $w$. 
\end{proof}
We preferred, for the sake of accessibility, to state and prove Theorem \ref{lowerBoundLength} for $|\Sigma|=4$, rather than for $|\Sigma|=2$. That is, we preferred to not complicate further a (stable) proof which is already quite technically involved, with the added benefit being relatively minor. We think that the ideas we gave in this sketch are convincing enough to support the claim that the result holds for $|\Sigma|=2$ as well, and a full proof will be added in the full journal version of this paper.

\section{The Analysis Problems for Subsequences}\label{sec:ConP}

We briefly overview what is known about the three analysis problems for simple types of gap constraints.

In the case of constraints $\gaptuple=(\{\emptyword\}, \ldots, \{\emptyword\})$ or $\gaptuple=(\Sigma^{\ell_1}, \ldots, \Sigma^{\ell_k})$, i.\,e., when the corresponding gapped sequences are words or partial words, all three problems can be solved in polynomial time, as $\subseqSet{\gaptuple}{w}$ has polynomial size in $|w|$. More efficient solutions can be easily obtained, based on, e.\,g., string processing data structures like suffix arrays. In the case of constraints $\gaptuple=(\Sigma^*, \ldots, \Sigma^*)$, i.\,e., when we deal with classical subsequences, $\subseqSet{\gaptuple}{w}$ is no longer of polynomial size, but the equivalence problem corresponds to the problem of testing the Simon congruence of two words, which can be tested in linear time \cite{dlt2020,GawrychowskiEtAl2021}. Universality can also be solved in linear time \cite{TCS::Hebrard1991,dlt2020}, while containment can be solved in polynomial time. 

\begin{proposition}
$\contProb$ for gap constraints $\gaptuple=((0,+\infty), \ldots, (0,+\infty))$ can be solved in polynomial time.
\end{proposition}
\begin{proof}
We start with our input words $w$ and $w'$. For $w$ we construct $A_w$, the subsequence-automaton \cite{CrochemoreMT03} which accepts all the subsequences of $w$. 

Assume $|w|=n$. Then, $A_w$ is a deterministic finite automaton, which has n+2 states $\{0,...,n,n+1\}$. The initial state of this automaton is $0$, and all states $i$, with $i\in [n]$ are final; the state $n+1$ is an error state. The transition are defined as follows: for all $i\in [n], k\in [n-i]$, and $a\in \Sigma$, we have a transition from state $i$ to $i+k$, labelled with $a$, if and only if $w[i+k]=a$ and $w[i+1..i+k-1]$ does not contain the letter $a$. Moreover, for all $i\in [n]$ and $a\in \Sigma$, we have a transition from state $i$ to state $n+1$, labelled with letter $a$, if and only if $w[i+1..n]$ does not contain $a$. For state $n+1$ we have loop-transitions for all letters $a\in \Sigma$. It is straightforward that $A_w$ accepts exactly the non-empty subsequences of $w$ and can be constructed in $O(n|\Sigma|)$ time.

For the word $w'$, with $|w'|=m$, we construct the automaton $A_{w'}$, as above, and then modify it to obtain the automaton $B_{w'}$ by simply making the state $m+1$ the single final state. Clearly, $B_{w'}$ accepts all strings which are not subsequences of $w'$. It can be constructed in $O(m|\Sigma|)$ time.

Now, if $\gaptuple$ consists of $k-1$ constraints, we can observe that  $\subseqSet{\gaptuple}{w} \subseteq \subseqSet{\gaptuple}{w'}$ if and only if there exists a word of length at most $k$ accepted of length by $A_w$ which is also accepted by $B_{w'}$. This can be checked in $O(nm|\Sigma|)$ time by simply computing the shortest word in the intersection of the language accepted by $A_w$ with the language accepted by $B_{w'}$. Our statement follows.
\end{proof}

It remains an interesting open problem whether a linear time algorithm exists for the problem approached in the proposition above. 

\section{Proof Details Omitted From Section~\ref{sec:NonUniversality}}\label{sec:NonUniversalityAppendix}

\subsection{Proof of Theorem~\ref{upperBoundUnboundedAlphabets}}

\begin{proof}
We first consider the case of length constraints. Let $\gaptuple$ be a $(k-1)$-tuple of length constraints, let $w, w' \in \Sigma^*$. In order to check $\subseqSet{\gaptuple}{w} \subseteq \subseqSet{\gaptuple}{w'}$, it is sufficient to check for every $v \in \Sigma^k$ whether $v \subseq_{\gaptuple} w$ implies $v \nsubseq_{\gaptuple} w'$. This can be done by solving $|\Sigma|^k$ times the matching problem for $p$, $\gaptuple$, and $w$ and $w'$, respectively. Since, by Corollary~\ref{constantPatternsRL}(1), the matching problem with length constraints can be solved in time $O(|w|\nz(\gaptuple))$, the statement of the theorem for problem $\contProb$ follows. Moreover, since $\equiProb$ can be solved by solving two $\contProb$-instances, the statement of the theorem for problem $\equiProb$ follows. Finally, for $\uniProb$, it is sufficient to check for every $v \in \Sigma^k$ whether $v \subseq_{\gaptuple} w$. Hence, analogously as before, the statement of the theorem for problem $\uniProb$ follows.\par
The case with reg-len or just regular constraints can be handled analogously, just by applying the matching upper bounds $O(|w|\states(\gaptuple))$ of Theorem \ref{constantPatternsRegularLength} and Corollary~\ref{constantPatternsRL}(2).
\end{proof}

\subsection{Decision Problems Used for Lower Bounds}

We first define two well-known decision problems. \par
The problem $\SatProb$ is defined as follows: The input is a Boolean formula $F$ in CNF, i.\,e., a set of clauses $F = \{c_1, c_2, \ldots, c_q\}$ over some set of variables $V = \{v_1, v_2, \ldots, v_k\}$, i.\,e., for every $i \in [q]$, we have $c_i \subseteq \{v_1, \neg v_1, \ldots, v_k, \neg v_k\}$. The question is whether $F$ is satisfiable, i.\,e., whether there is an assignment $\pi : V \to \{0, 1\}$ that makes at least one literal of each clause $c_i$ true.\par
The (parameterised) problem $\kISProb$ is defined as follows: The input is an undirected graph $G = (V, E)$ with $V = \{v_1, v_2, \ldots, v_n\}$ and some $k \in [n]$. The question is whether $G$ has a $k$-independent set, i.\,e., a set $A \subseteq V$ with $|A| = k$ and $\{u, u'\} \notin E$ for every $u, u' \in V$ with $u \neq u'$. For convenience, we assume in the following that undirected graphs $G = (V, E)$ are represented by symmetric directed graphs, i.\,e., $E$ is a symmetric binary relation over $V$. \par
For the sake of convenience, in the following we define reductions that will prove the statements of Theorems~\ref{binaryCaseHardnessTheorem},~\ref{HardnessNonUniversalityBounded},~and~\ref{HardnessNonUniversalityUnbounded} only for the non-universality problem $\nuniProb$. We will discuss in Section~\ref{sec:ContEquiAppendix} how the statements of Theorems~\ref{binaryCaseHardnessTheorem},~\ref{HardnessNonUniversalityBounded},~and~\ref{HardnessNonUniversalityUnbounded} also follow for the problems $\ncontProb$ and $\nequiProb$.

\subsection{The Meta Non-Universality Problem}

We now define a \emph{meta non-universality problem} ($\metaNuniProb$ for short) that can be easily used to express other intractable decision problems. An instance of this meta non-universality problem is defined as follows. \par
Let $\Gamma = \{b_1, b_2, \ldots, b_m\}$ be some set of size $m$, and let $q, k \in \mathbb{N}$. An instance of the problem is a $(q \times k)$-matrix the entries $W_{i, j}$ of which are subsets of $\Gamma$. For every $i \in [q]$, we associate with row $i$ of the matrix the language $\lang{W_i} = W_{i, 1} \cdot W_{i, 2} \cdots W_{i, k}$, i.\,e., we simply represent the elements of $W_{i, 1} \times W_{i, 2} \times \ldots \times W_{i, k}$ as length-$k$ strings over $\Gamma$ in the natural way. The question is then to decide whether $\cup_{i \in [q]} \lang{W_i} \neq \Gamma^k$.\par
As an example, let $\Gamma = \{\ta, \tb, \tc, \td, \te\}$, $q = 4$, and $k = 3$. Then the following matrix is a possible instance:
\begin{equation*}
\begin{pmatrix}
\{\ta\} & \{\tb\} & \{\ta, \tc, \td\}\\
\{\tc, \td\} & \{\tb\} & \{\ta\}\\
\{\ta\} & \{\tb, \tc, \td, \te\} & \{\td\}\\
\{\te\} & \{\tb\} & \{\tc, \te\}
\end{pmatrix}
\end{equation*}

We observe that, e.\,g., $W_{3, 2} = \{\tb, \tc, \td, \te\}$ and $W_{4, 1} = \{\te\}$. Moreover, $\lang{W_1} = \{\ta\} \cdot \{\tb\} \cdot \{\ta, \tc, \td\} = \{\ta \tb \ta, \ta \tb \tc, \ta \tb \td\}$ and $\lang{W_4} = \{\te\tb\tc, \te\tb\te\}$. It can be easily seen that for this instance we have $\cup_{i \in [4]} \lang{W_i} \neq \Gamma^3$, e.\,g., $\cup_{i \in [4]} \lang{W_i}$ does not contains strings that start with $\tb$.\par

Next, we will show how the problem $\SatProb$ and the parameterised problem $k$-CLIQUE can be reduced to $\metaNuniProb$, which is comparatively simple. After this, we will show how $\metaNuniProb$ reduces to $\nuniProb$ with length constraints, which requires more work.

\subsection{From $\SatProb$ to $\metaNuniProb$}\label{sec:SatReduction}

Let $F = \{c_1, c_2, \ldots, c_q\}$ be a Boolean formula in CNF on variables $\{v_1, v_2, \ldots, v_k\}$. We define alphabet $\Gamma = \{0, 1\}$ and the $(q \times k)$-matrix with the entries $W_{i, j}$ as follows (note that $q$ and $k$ are already defined as the number of clauses and the number of variables, respectively). For every $i \in [q]$ and $j \in [k]$, we define
\begin{equation*}
W_{i, j} = \begin{cases}
\{0\}& \text{if $v_j \in c_i$},\\
\{1\}& \text{if $\neg v_j \in c_i$},\\
\{0, 1\}& \text{if $\{v_j, \neg v_j\} \cap c_i = \emptyset$}.
\end{cases}
\end{equation*}

It can be verified with moderate effort, that for every $i \in [q]$, $\lang{W_i}$ contains exactly the Boolean assignments that do not satisfy clause $c_i$ (where Boolean assignments are represented as length-$k$ strings over $\{0, 1\}$ in the obvious way). This means that $\cup_{i \in [q]} \lang{W_i}$ is the set of all non-satisfying assignments. Thus, $\cup_{i \in [q]} \lang{W_i} \neq \{0, 1\}^k$ if and only if $F$ is satisfiable, i.\,e., the constructed $\metaNuniProb$ instance is positive if and only if $F$ is satisfiable. 

As an example, consider the clauses $c_1 = \{v_1, \neg v_2, v_3\}$, $c_2 = \{\neg v_1, \neg v_2, v_5\}$, $c_3 = \{v_3, v_4, v_5\}$ over the variables $\{v_1, v_2, \ldots, v_6\}$, which yields the following $\metaNuniProb$ instance:
\begin{equation*}
\begin{pmatrix}
\{0\} & \{1\} & \{0\} & \{0, 1\} & \{0, 1\} & \{0, 1\} &\\
\{1\} & \{1\} & \{0, 1\} & \{0, 1\} & \{0, 1\} & \{0\} &\\
\{0, 1\} & \{0, 1\} & \{0\} & \{0\} & \{0\} & \{0, 1\} &
\end{pmatrix}
\end{equation*}

In this case, $100010 \notin \cup_{i \in [3]} \lang{W_i}$, which means that $100010$ is a satisfying assignment.

\subsection{From $\kISProb$ to $\metaNuniProb$}\label{sec:kCliqueReduction}

Let $G = (V, E)$ be an undirected graph with $V = \{v_1, v_2, \ldots, v_n\}$ and $E = \{e_1, e_2, \ldots, e_m\}$, and let $k \in [n]$. For technical reasons, we assume that every $v_i \in V$ has a loop, i.\,e., $E$ is reflexive (note that this modification does not change whether a given set is an independent set or not); in particular, these loops are explicitly included in the set $E$.
We define $\Gamma = V$.\par
We generally interpret strings $w \in \Gamma^k$ as sets $V(w) = \{w[j] \mid j \in [k]\}$ of at most $k$ vertices. For every $r, s \in [k]$ with $r \neq s$, we say that \emph{$w$ contains edge $e \in E$ at positions $r, s \in [k]$} if $e = (w[r], w[s])$.
Every string $w \in \Gamma^k$ with $|V(w)| < k$ satisfies $w[r] = w[s]$ for some $r, s \in [k]$ with $r \neq s$, and, since every vertex has a loop, this means that $w$ necessarily contains an edge. 
Hence, for every $w \in \Gamma^k$, we have that $w$ contains at least one edge if and only if $V(w)$ is not a $k$-independent set (i.\,e., it is either not an independent set or it is a set of cardinality strictly less than $k$).\par
We fix some bijection $\nu : \{(i, r, s) \in [m] \times [k] \times [k] \mid r \neq s\} \to [m k(k-1)]$. For every $i \in [m]$ with $e_i = (u, v)$, and every $r, s, j \in [k]$ with $r \neq s$, we define sets 
\begin{equation*}
W_{\nu(i, r, s), j} = \begin{cases}
\{u\}& \text{if $j = r$},\\
\{v\}& \text{if $j = s$},\\
V& \text{else}.
\end{cases}
\end{equation*}
This concludes the definition of the reduction. Note that here the number of columns of the constructed matrix corresponds to the number $k$, while the number of rows is $q := m k(k-1)$. Let us explain this reduction intuitively.\par
For a fixed $e_i \in E$ and fixed positions $r, s \in [k]$ with $r \neq s$, the set 

\begin{equation*}
\lang{W_{\nu(i, r, s)}} = W_{\nu(i, r, s), 1} W_{\nu(i, r, s), 2} \ldots W_{\nu(i, r, s), k} 
\end{equation*}

contains exactly the strings $w \in \Gamma^k$ that contain the edge $e_i$ at positions $r$ and $s$. This means that $\lang{W_{\nu(i, r, s)}}$ represents exactly the sets of cardinality at most $k$ that contain edge $e_i$ (and are therefore not $k$-independent sets). Note that if $e_i$ is a loop, then $V(w)$ might be an independent set, but one of cardinality strictly less than $k$. \par
As an example, assume that $G$ contains an edge $e_9 = (v_3, v_7)$ and that $k = 3$, then this edge is represented in the matrix by the following rows:
\begin{align*}
&\text{row } \nu(9, 1, 2):& &\{v_3\}& &\{v_7\}& &V&\\
&\text{row } \nu(9, 1, 3):& &\{v_3\}& &V& &\{v_7\}&\\
&\text{row } \nu(9, 2, 1):& &\{v_7\}& &\{v_3\}& &V&\\
&\text{row } \nu(9, 2, 3):& &V& &\{v_3\}& &\{v_7\}&\\
&\text{row } \nu(9, 3, 1):& &\{v_7\}& &V&  &\{v_3\}&\\
&\text{row } \nu(9, 3, 2):& &V& &\{v_7\}& &\{v_3\}&
\end{align*}

For example, $\lang{W_{\nu(9, 3, 1)}} = \{v_7 v_1 v_3, v_7 v_2 v_3, v_7 v_3 v_3, v_7 v_4 v_3, \ldots\}$ represents all vertex sets of cardinality at most $3$ that contain edge $(v_3, v_7)$ at positions $3$ and $1$; observe that $v_7 v_3 v_3$ represents the vertex set $\{v_3, v_7\}$. Moreover, if $e_{11} = (v_7, v_3)$, then rows $\nu(9, 1, 2)$ and $\nu(11, 2, 1)$ are the same. This redundancy could be avoided, but it would unnecessarily complicate the reduction.

\begin{lemma}
$\bigcup_{i \in [m], r, s \in [k], r \neq s} \lang{W_{\nu(i, r, s)}} \neq \Gamma^k$ if and only if $G$ has a $k$-independent set.
\end{lemma}

\begin{proof}
If $G$ has a $k$-independent set $A$, then $A = \{v_{t_1}, v_{t_2}, \ldots, v_{t_k}\}$ with $|A| = k$ such that, for every $\ell, \ell' \in [k]$ with $\ell \neq \ell'$, $(v_{t_{\ell}}, v_{t_{\ell'}}) \notin E$, which means that every $w \in \Gamma^k$ with $V(w) = A$ has no edge at any positions. This means that $w \notin \bigcup_{i \in [m], r, s \in [k], r \neq s} \lang{W_{\nu(i, r, s)}}$. Hence, $w \notin \bigcup_{i \in [m], r, s \in [k], r \neq s} \lang{W_{\nu(i, r, s)}}$ for every $w \in \Gamma^k$ with $V(w) = A$, and therefore $\bigcup_{i \in [m], r, s \in [k], r \neq s} \lang{W_{\nu(i, r, s)}} \neq \Gamma^k$ .\par
Next, assume that $G$ has no $k$-independent set and let $w \in \Gamma^k$ be arbitrarily chosen. If $|V(w)| < k$, then there are $r, s \in [k]$ with $r \neq s$ and $w[r] = w[s] = v_{\ell}$ for some $\ell \in [n]$. This means that $w \in \lang{W_{\nu(i, r, s)}}$ with $e_i = (v_{\ell}, v_{\ell})$. If $|V(w)| = k$, then, since $V(w)$ is not an independent set, there is an edge $e_i = (v_{\ell}, v_{\ell'})$ with $\{v_{\ell}, v_{\ell'}\} \subseteq V(w)$, and therefore there are $r, s \in [k]$ with $r \neq s$ and $(w[r], w[s]) = e_i$, which means that $w \in \lang{W_{\nu(i, r, s)}}$. Consequently, $\Gamma^k \subseteq \bigcup_{i \in [m], r, s \in [k], r \neq s} \lang{W_{\nu(i, r, s)}}$, which means that $\bigcup_{i \in [m], r, s \in [k], r \neq s} \lang{W_{\nu(i, r, s)}} = \Gamma^k$.
\end{proof}

\subsection{Reduction From $\metaNuniProb$ to $\nuniProb$}\label{sec:theReduction}

We no define a reduction from $\metaNuniProb$ to $\nuniProb$ with length constraints. By doing so, we obtain a reduction from $\SatProb$ to $\nuniProb$ with length constraints, and a (parameterised) reduction from $\kISProb$ to $\nuniProb$ with length constraints.\par
Let $\Gamma = \{b_1, b_2, \ldots, b_m\}$, $q, k \in \mathbb{N}$, and, for every $i \in [q], j \in [k]$, let $W_{i, j} \subseteq \Gamma$. We transform this $\metaNuniProb$ instance into an instance of $\nuniProb$ with length constraints (i.\,e., a string and a $(k-1)$-tuple of gap constraints), as follows. We first define the alphabet $\Sigma = \Gamma \cup \{\#\}$ (with $\# \notin \Gamma$). Then we define a $(k-1)$-tuple $\gaptuple = (C_1, C_2, \ldots, C_{k-1})$ of gap constraints with $C_i = (\lowerBoundShort{i}, \upperBoundShort{i}) = (m-1, 3m-1)$ for every $i \in [k-1]$ (recall that $m$ is $\Gamma$'s cardinality). In order to conclude the reduction, we have to construct a string $K(W_1, \ldots, W_q)$ over $\Sigma$, such that $\subseqSet{\gaptuple}{K(W_1, \ldots, W_q)} = \Sigma^k$ if and only if $\cup_{i \in [q]} \lang{W_i} = \Gamma^k$.

For every $i \in [q]$ and $j \in [k]$, let $w_{i, j} \in \Gamma^*$ be some string representation of $W_{i, j}$, i.\,e., $\alphabet{w_{i, j}} = W_{i, j}$ and $|w_{i, j}| = |W_{i, j}| \leq m$. 
For every $i \in [q]$, we define the string
\begin{equation*}
S(W_{i}) = w_{i, 1} (\#)^{m-1} w_{i, 2} (\#)^{m-1} \ldots (\#)^{m-1} w_{i, k}\,.
\end{equation*}

Any $\gaptuple$-subsequence of $S(W_{i})$ may or may not contain occurrences of symbol $\#$. However, those $\gaptuple$-subsequence of $S(W_{i})$ that do not contain occurrences of symbol $\#$ are exactly the strings in $\lang{W_i}$. This is stated by the next lemma.

\begin{lemma}\label{SingleWiLemmaAppendix}
For every $i \in [q]$, $(\subseqSet{\gaptuple}{S(W_{i})} \cap \Gamma^*) = \lang{W_i}$.
\end{lemma}

\begin{proof}
We first prove that $(\subseqSet{\gaptuple}{S(W_{i})} \cap \Gamma^*) \subseteq \lang{W_i}$. Let $p \in (\subseqSet{\gaptuple}{S(W_{i})} \cap \Gamma^*)$. This means that there is an embedding $e$ that satisfies $\gaptuple$ and $p \subseq_e S(W_{i})$.
If, for some $\ell \in [k - 1]$, $\ell$ and $\ell + 1$ are mapped by $e$ to positions of the same factor $w_{i, j}$ for some $j \in [k]$, then $|\gap{S(W_{i})}{e}{\ell}| \leq |w_{i, j}| - 2 \leq m - 2$ and therefore the lower length constraint $\lowerBoundShort{\ell} = m-1$ would be violated. Consequently, since $p \in \Gamma^*$, the embedding $e$ must be such that, for every $j \in [k]$, $j$ is mapped to a position of $w_{i, j}$. This directly implies that $p \in \lang{W_i}$. \par
Next, we prove that $\lang{W_i} \subseteq (\subseqSet{\gaptuple}{S(W_{i})} \cap \Gamma^*)$. Let $u \in \lang{W_i}$. By construction of $S(W_{i})$, for every $j \in [k]$, there is an occurrence of $u[j]$ in $w_{i, j}$. Thus, there is an embedding $e$ that maps each $j \in [k]$ to such an occurrence of $u[j]$ in $w_{i, j}$. Consequently, $u \subseq_{e} S(W_{i})$. It only remains to observe that $e$ satisfies the length constraints. To this end, let $j \in [k - 1]$. This means that $e$ maps $j$ to a position of the prefix $w_{i, j}$ of the factor $w_{i, j} (\#)^{m-1} w_{i, j + 1}$, and $e$ maps $j+1$ to a position of the suffix $w_{i, j + 1}$ of the factor $w_{i, j} (\#)^{m-1} w_{i, j + 1}$. Since $|w_{i, j}| \leq m$ and $|w_{i, j + 1}| \leq m$, we have that $m - 1 \leq |\gap{S(W_{i})}{e}{j}| \leq 3(m-1)$. This means that $e$ satisfies the length constraints. Hence, $u \in (\subseqSet{\gaptuple}{S(W_{i})} \cap \Gamma^*)$.
\end{proof}

The above lemma shows that $\gaptuple$-subsequences over $\Gamma$ of the strings $S(W_{i})$ corresponds to the strings of $\cup_{i \in [q]} \lang{W_i}$. We next define a string $T$ that satisfies that all $\subseqSet{\gaptuple}{T}$ is the set of all length-$k$ strings over $\Sigma$ with at least one occurrence of the symbol $\#$.\par
For every $i \in [k]$, let $T_i = T_{i, 1} T_{i, 2} \ldots T_{i, k}$, where, for every $j \in [k] \setminus \{i\}$, $T_{i, j} = b_1 b_2 \ldots b_m \#^{m}$, and $T_{i, i} = \#^{m}$. We define the string $T$ by
\begin{equation*}
T = T_1 (\#^{3m}) T_2 (\#^{3m}) \ldots (\#^{3m}) T_k\,. 
\end{equation*}

\begin{lemma}\label{sharpUnivStringLemmaAppendix}
$\subseqSet{\gaptuple}{T} = \{w \in \Sigma^k \mid |w|_{\#} \geq 1\}$.
\end{lemma}

\begin{proof}
We first prove that $\subseqSet{\gaptuple}{T} \subseteq \{w \in \Sigma^k \mid |w|_{\#} \geq 1\}$, i.\,e., every string from $\subseqSet{\gaptuple}{T}$ contains at least one occurrence of symbol $\#$. 
For contradiction, we assume that there is some $p \in \subseqSet{\gaptuple}{T}$ without any occurrence of $\#$, which means that there is an embedding $e$ that satisfies the length constraints $\gaptuple$ such that $p \subseq_e T$. Since $e$ satisfies the upper length bounds $\upperBoundShort{j} = 3m - 1$ for every $j \in [k-1]$, there must be an $\ell \in [k]$ such that all positions $j \in [k]$ are mapped by $e$ to positions of the factor $T_{\ell}$ (note that otherwise either $p[j] = \#$ for some $j \in [k]$, or $|\gap{T}{e}{j}| \geq 3m$ for some $j \in [k - 1]$). Consequently, $p \subseq_{e'} T_{\ell}$ for some embedding $e'$ that satisfies the length constraints $\gaptuple$. By construction, $T_{\ell}$ contains $k - 1$ maximal factors over alphabet $\Gamma$. By assumption, $p \in \Gamma^*$, which, by the pigeonhole principle, means that there is at least one $j \in [k - 1]$ such that both $j$ and $j+1$ are mapped by $e'$ to the same maximal factor over $\Gamma$. Since each such maximal factor over $\Gamma$ has size at most $m$, it follows that $|\gap{T_{\ell}}{e'}{j}| = m-2$, which violates the lower length bound $\lowerBoundShort{j} = m - 1$ and is therefore a contradiction. \par
Next, we prove that $\{w \in \Sigma^k \mid |w|_{\#} \geq 1\} \subseteq \subseqSet{\gaptuple}{T}$. To this end, let $p \in \{w \in \Sigma^k \mid |w|_{\#} \geq 1\}$ be arbitrarily chosen, and assume that $u[r] = \#$ with $r \in [k]$ (since $|u|_{\#} \geq 1$, such an $r$ must exist). We recall that $T_r = T_{r,1} T_{r,2} \ldots T_{r,k}$, where, for every $j \in [k] \setminus \{r\}$, $T_{r, j} = b_1 b_2 \ldots b_m \#^{m}$, and $T_{r, r} = \#^{m}$. 
We will construct an embedding $e : [k] \to [|T_r|]$ that satisfies the length constraints $\gaptuple$ such that $p \subseq_{e} T_r$. Since $T_r$ is a factor of $T$, this implies that $p \in \subseqSet{\gaptuple}{T}$, and therefore $\{w \in \Sigma^k \mid |w|_{\#} \geq 1\} \subseteq \subseqSet{\gaptuple}{T}$. \par
We map $r$ to the first position of $T_{r, r} = \#^m$, and, for every $j \in [k] \setminus \{r\}$, we map $j$ to the $t^{\text{th}}$ position of $T_{r, j}$ if $p[j] = b_t$, and to the $(m + 1)^{\text{st}}$ position of $T_{r, j}$ if $p[j] = \#$. We observe that this embedding $e$ satisfies that every $j$ is mapped to an occurrence of $p[j]$ of $T_{r, j}$, i.\,e., $p \subseq_{e} T_r$. It remains to prove that $e$ satisfies the length constraints $\gaptuple$. \par
Since every $j \in [k] \setminus \{r\}$ is mapped to one of the first $m + 1$ occurrences of the length-$2m$ factor $T_{r, j}$, and position $r$ is mapped to the first possition of the length-$m$ factor $T_{r, r}$, it can be easily seen that $m - 1 \leq |\gap{T_{r}}{e}{j}| \leq 3m - 1$. Hence, $e$ satisfies the length constraints.
\end{proof}

We are now ready to append the gadgets developed above in order to define the complete string $K(W_{1}, \ldots, W_{q})$ as follows:
\begin{equation*}
K(W_{1}, \ldots, W_{q}) = T (\#^{3m}) S(W_1) (\#^{3m}) S(W_2) (\#^{3m}) \ldots (\#^{3m}) S(W_q)\,.
\end{equation*}

Before concluding the proof of correctness of the reduction, we show the following lemma, which states that every $\gaptuple$-sequence of $K(W_{1}, \ldots, W_{q})$ without occurrences of $\#$ must be mapped into one of the factors $S(W_i)$. 

\begin{lemma}\label{limitedToSJLemmaAppendix}
Let $p \in \subseqSet{\gaptuple}{K(W_1, \ldots, W_q)}$ with $p \in \Gamma^k$. Then, for some $i \in [q]$, $p \in \subseqSet{\gaptuple}{S(W_i)}$.
\end{lemma}

\begin{proof}
Let $p \in \subseqSet{\gaptuple}{K(W_1, \ldots, W_q)}$ and let $e$ be an embedding that satisfies the length constraints $\gaptuple$, and $p \subseq_e K(W_1, \ldots, W_q)$. Since $e$ satisfies the upper length bounds $\upperBoundShort{j} = 3m-1$ for every $j \in [k-1]$, all positions $j \in [k]$ are mapped by $e$ to positions of some factor $S(W_i)$ with $i \in [q]$, or all positions $j \in [k]$ are mapped by $e$ to positions of the factor $T$ (note that otherwise either $p[j] = \#$ for some $j \in [k]$, or $|\gap{K(W_{1}, \ldots, W_{q})}{e}{j}| \geq 3m$ for some $j \in [k - 1]$). This means that, for some embedding $e'$ that satisfies $\gaptuple$, $p \subseq_{e'} S(W_i)$ for some $i \in [q]$, or $p \subseq_{e'} T$. By Lemma~\ref{sharpUnivStringLemmaAppendix}, we know that $\subseqSet{\gaptuple}{T} = \{w \in \Sigma^k \mid |w|_{\#} \geq 1\}$; thus, since $|p|_{\#} = 0$, $p \subseq_{e'} T$ is not possible. Hence, $p \subseq_{e'} S(W_i)$ for some $i \in [q]$, which means that $p \in \subseqSet{\gaptuple}{S(W_i)}$.
\end{proof}

Finally, the following lemma concludes the proof of correctness.

\begin{lemma}\label{mainCorrectnessLemmaAppendix}
$\subseqSet{\gaptuple}{K(W_1, \ldots, W_q)} = \Sigma^{k} \iff \cup_{i \in [q]} \lang{W_{i}} = \Gamma^k$.
\end{lemma}

\begin{proof}
We first show that
\begin{equation*}
\subseqSet{\gaptuple}{K(W_1, \ldots, W_q)} = \bigcup_{i \in [q]} (\subseqSet{\gaptuple}{S(W_{i})} \cap \Gamma^*) \cup \subseqSet{\gaptuple}{T}\,.
\end{equation*}
The ``$\supseteq$''-direction holds, since $T$ and every $S(W_{i})$ for every $i \in [q]$ are factors of $K(W_1, \ldots, W_q)$. Now let $p \in \subseqSet{\gaptuple}{K(W_1, \ldots, W_q)}$  be arbitrarily chosen. If $|p_{\#}| \geq 1$, then, by Lemma~\ref{sharpUnivStringLemmaAppendix}, $p \in \subseqSet{\gaptuple}{T}$. If $|p_{\#}| = 0$, then $p \in \subseqSet{\gaptuple}{K(W_1, \ldots, W_q)}$ and $p \in \Gamma^k$. Thus, Lemma~\ref{limitedToSJLemmaAppendix} implies that $p \in \subseqSet{\gaptuple}{S(W_i)}$ for some $i \in [q]$, which means that $p \in \subseqSet{\gaptuple}{S(W_i)} \cap \Gamma^*$. This shows that the ``$\subseteq$''-direction holds as well.\par

Now we prove the \emph{if} direction of the statement of the lemma and assume that $\cup_{i \in [q]} \lang{W_i} = \Gamma^k$. By Lemma~\ref{SingleWiLemmaAppendix}, we know that 
$(\subseqSet{\gaptuple}{S(W_i)} \cap \Gamma^*) = \lang{W_i}$ for every $i \in [q]$. With Lemma~\ref{sharpUnivStringLemmaAppendix} and the observation from above, this means that 
\begin{align*}
\subseqSet{\gaptuple}{K(W_1, \ldots, W_q)} &= \bigcup_{i \in [q]} (\subseqSet{\gaptuple}{S(W_{i})} \cap \Gamma^*) \cup \subseqSet{\gaptuple}{T}\\
&= \left(\bigcup_{i \in [q]} \lang{W_i}\right) \cup \subseqSet{\gaptuple}{T}\\
&= \Gamma^k \cup \{w \in \Sigma^k \mid |w|_{\#} \geq 1\} = \Sigma^{k}\,.
\end{align*}\par
In order to prove the \emph{only if} direction, we assume that $\subseqSet{\gaptuple}{K(W_1, \ldots, W_q)} = \Sigma^{k}$. With our observation from above, this means that
\begin{equation*}
\Sigma^{k} = \bigcup_{i \in [q]} (\subseqSet{\gaptuple}{S(W_{i})} \cap \Gamma^*) \cup \subseqSet{\gaptuple}{T} 
\end{equation*}
Applying Lemma~\ref{sharpUnivStringLemmaAppendix} yields
\begin{equation*}
\Sigma^{k} = \bigcup_{i \in [q]} (\subseqSet{\gaptuple}{S(W_{i})} \cap \Gamma^*) \cup \{w \in \Sigma^k \mid |w|_{\#} \geq 1\} \,.
\end{equation*}
Since $\bigcup_{i \in [q]} (\subseqSet{\gaptuple}{S(W_{i})} \cap \Gamma^*)$ and $\{w \in \Sigma^k \mid |w|_{\#} \geq 1\}$ are clearly disjoint, we can conclude that
\begin{align*}
\Sigma^{k} \setminus \{w \in \Sigma^k \mid |w|_{\#} \geq 1\} &= \bigcup_{i \in [q]} (\subseqSet{\gaptuple}{S(W_{i})} \cap \Gamma^*) &\iff \\
\Gamma^{k} &= \bigcup_{i \in [q]} (\subseqSet{\gaptuple}{S(W_{i})} \cap \Gamma^*)\,. &
\end{align*}
Finally, Lemma~\ref{SingleWiLemmaAppendix} implies that $\Gamma^k = \bigcup_{i \in [q]} \lang{W_i}$.
\end{proof}

\subsection{Direct Reduction from $\SatProb$ to $\nuniProb$}\label{sec:binaryCaseReduction}

Let $F = \{c_1, c_2, \ldots, c_q\}$ be a Boolean formula in CNF on variables $\{v_1, v_2, \ldots, v_k\}$. Let $\Sigma = \{\ta, \tb\}$. For every $i \in [q]$ and $j \in [k]$, we define
\begin{equation*}
W_{i, j} = \begin{cases}
\{\ta \ta\}& \text{if $v_j \in c_i$},\\
\{\tb \tb\}& \text{if $\neg v_j \in c_i$},\\
\{\ta \ta, \tb \tb\}& \text{if $\{v_j, \neg v_j\} \cap c_i = \emptyset$},
\end{cases}
\end{equation*}

By interpreting $\ta \ta$ and $\tb \tb$ as Boolean values \emph{false} and \emph{true}, we can interpret the words form $L_{\Sigma \Sigma} = \lang{(\ta \ta \vee \tb \tb)^k}$ as assignments for $F$'s variables. Moreover, the set $\lang{W_i} = W_{i, 1} W_{i, 2} \ldots W_{i, k}$ represents all assignments that do not satisfy clause $c_i$, and therefore $\bigcup_{i \in [q]}\lang{W_i}$ is the set of all non-satisfying assignments. In particular, $F$ is satisfying if and only if $\bigcup_{i \in [q]}\lang{W_i} \neq L_{\Sigma \Sigma}$. \par
We define a $2k$-tuple $\gaptuple = (C_1, C'_1, C_2, C'_2, \ldots, C_k)$ of length constraints, where $C_{j} = (0, 0)$ for every $j \in [k]$ and $C'_{j} = (3, 9)$ for every $j \in [k-1]$, i.\,e., every $\gaptuple$-sequence corresponds to $k$ pairs of consecutive symbols with a gap of length at least $3$ and at most $9$ in between.

For every $i \in [q]$ and $j \in [k]$, we define the string 
\begin{equation*}
w_i = w_{i, 1} \tb \ta \tb w_{i, 2} \tb \ta \tb \ldots \tb \ta \tb w_{i, k}\,,
\end{equation*}
where 
\begin{equation*}
w_{i, j} = \begin{cases}
\ta \ta& \text{if $v_j \in c_i$},\\
\ta \tb \tb \ta& \text{if $\neg v_j \in c_i$},\\
\ta \ta \tb \tb \ta& \text{if $\{v_j, \neg v_j\} \cap c_i = \emptyset$}.
\end{cases}
\end{equation*}

\begin{lemma}\label{SingleWiLemmaBinaryCaseAppendix}
$(\subseqSet{\gaptuple}{w_i} \cap L_{\Sigma \Sigma}) = \lang{W_i}$.
\end{lemma}

\begin{proof}
We first prove $(\subseqSet{\gaptuple}{w_i} \cap L_{\Sigma \Sigma}) \subseteq \lang{W_i}$. To this end, let $p \in (\subseqSet{\gaptuple}{w_i} \cap L_{\Sigma \Sigma})$. This means that there is an embedding $e$ that satisfies $\gaptuple$ and $p \subseq_e w_i$, and that, for every $j \in [k]$, $p[2j-1]p[2j] \in \{\ta \ta, \tb \tb\}$. Since all factors $w_{i, j}$ start and end with symbol $\ta$ and the separating factor $\tb \ta \tb$ starts and ends with symbol $\tb$, for every $j \in [k]$, $p[2j-1]p[2j]$ is mapped by $e$ completely inside a factor $w_{i, j}$. Moreover, due to the lower bounds $3$ of the length constraints $C'_j$ for every $j \in [k-1]$, it is not possible that both $p[2j-1]p[2j]$ and $p[2(j+1)-1]p[2(j+1)]$ are mapped by $e$ inside the same factor $w_{i, \ell}$. Hence, for every $j \in [k]$, $p[2j-1]p[2j]$ is mapped by $e$ to a factor $\ta \ta$ or $\tb \tb$ inside of $w_{i, j}$. By construction, for every $j \in [k]$, the factor $w_{i, j}$ contains factor $\ta \ta$, but not $\tb \tb$, if $W_{i, j} = \{\ta \ta\}$, it contains factor $\tb \tb$, but not $\ta \ta$, if $W_{i, j} = \{\tb \tb\}$, and it contains factor $\ta \ta$ and factor $\tb \tb$, if $W_{i, j} = \{\ta \ta, \tb \tb\}$. Consequently, $p \in W_1 W_2 \ldots W_k = \lang{W_i}$.\par
Next, we prove that $\lang{W_i} \subseteq (\subseqSet{\gaptuple}{w_i} \cap L_{\Sigma \Sigma})$. To this end, let $p \in \lang{W_i}$, which means that $p[2j-1]p[2j] \in W_{i, j}$ for every $j \in [k]$, which also means that $p \in L_{\Sigma \Sigma}$. Therefore, by construction of $w_{i}$, $p[2j-1]p[2j]$ is a factor of $w_{i, j}$ for every $j \in [k]$. Thus, there is an embedding $e$ that, for every $j \in [k]$, maps each $2j-1$ and $2j$ to consecutive positions in $w_{i, j}$ that correspond to a factor $p[2j-1]p[2j]$, which means that $p \subseq_e w_i$. In particular, $e$ satisfies all length constraints $C_j = (0, 0)$ with $j \in [k]$. Finally, we note that by construction of the factors $w_{i, j}$, we have $3 \leq |\gap{w_i}{e}{2j}| \leq 8$ for every $j \in [k-1]$, which means that $e$ also satisfies the length constraints $C'_j$ with $j \in [k-1]$, and therefore $p \in (\subseqSet{\gaptuple}{w_i} \cap L_{\Sigma \Sigma})$.
\end{proof}

For every $i \in [k]$, let $T_i = T_{i, 1} \tb \ta \tb T_{i, 2} \tb \ta \tb \ldots \tb \ta \tb T_{i, k}$, where, for every $j \in [k] \setminus \{i\}$, $T_{i, j} = \ta \ta \tb \tb \ta$, and $T_{i, i} = \ta \tb \ta$. We define the string $T$ by
\begin{equation*}
T = T_1 \tb (\ta \tb)^5 T_2 \tb (\ta \tb)^5 \ldots \tb (\ta \tb)^5 T_k\,. 
\end{equation*}

We define $\overline{L_{\Sigma \Sigma}} = \{\ta, \tb\}^{2k} \setminus L_{\Sigma \Sigma}$.

\begin{lemma}\label{sharpUnivStringLemmaBinaryCaseAppendix}
$\subseqSet{\gaptuple}{T} = (\overline{L_{\Sigma \Sigma}})$.
\end{lemma}

\begin{proof}
We first prove $\subseqSet{\gaptuple}{T} \subseteq (\overline{L_{\Sigma \Sigma}})$. To this end, let $p \in \subseqSet{\gaptuple}{T}$, which means that there is an embedding $e$ that satisfies $\gaptuple$ and $p \subseq_e T$. 
For contradiction, we assume that $p \notin (\overline{L_{\Sigma \Sigma}})$, which means that $p \in L_{\Sigma \Sigma}$ and therefore $p[2j-1]p[2j] \in \{\ta \ta, \tb \tb\}$ for every $j \in [k]$. Next we observe that all factors $T_{i, j}$ start and end with symbol $\ta$, while all separating factors $\tb \ta \tb$ and $\tb (\ta \tb)^5$ start and end with symbol $\tb$. In particular, this means that, for every $j \in [k]$, both $2j-1$ and $2j$ are mapped by $e$ to positions of some $T_{i, j}$ factor. Moreover, due to the upper bound $9$ of length constraints $C'_j$ for every $j \in [k-1]$, it is also not possible that $e$ maps positions of $p$ to different factors $T_{i}$ and $T_{i + 1}$ for some $i \in [k]$, i.\,e., there is some $r \in [k]$ such that all positions of $p$ are mapped to positions of $T_{r}$. As observed above, for every $j \in [k]$, factor $p[2j-1]p[2j]$ must be mapped to positions of some $T_{r, \ell}$, i.\,e., to a factor $\ta \ta$ or $\tb \tb$ in some $T_{r, \ell}$. Due to the lower bounds $3$ of the length constraints $C'_j$ for every $j \in [k-1]$, it is not possible that both $p[2j-1]p[2j]$ and $p[2(j+1)-1]p[2(j+1)]$ are mapped by $e$ to the same factor $T_{r, \ell}$, which means that, for every $j \in [k]$, $p[2j-1]p[2j]$ is mapped to $T_{r, j}$. This is a contradiction, since $T_{r, r}$ does not contain any occurrence of factor $\ta \ta$ or $\tb \tb$.\par
Next, we prove $(\overline{L_{\Sigma \Sigma}}) \subseteq \subseqSet{\gaptuple}{T}$. To this end, let $p \in (\overline{L_{\Sigma \Sigma}})$, which means that there is some $r \in [k]$ such that $p[2r-1]p[2r] \in \{\ta \tb, \tb \ta\}$. Let the embedding $e$ be defined as follows. For every $j \in [k]$ with $j \neq r$, we map $2j-1$ and $2j$ to a factor $p[2j-1]p[2j]$ in $T_{r, j}$ (note that since $T_{r, j} = \ta \ta \tb \tb \ta$ and $p[2j-1]p[2j] \in \{\ta \ta, \ta \tb, \tb \ta, \tb \tb\}$, $T_{r, j}$ must contain the factor $p[2j-1]p[2j]$). Furthermore, we map $2r-1$ and $2r$ to the factor $p[2r-1]p[2r]$ in $T_{r, r}$ (note that this must be possible since $T_{r, r} = \ta \tb \ta$ and $p[2r-1]p[2r] \in \{\ta \tb, \tb \ta\}$). Finally, we note that by construction of the factors $T_{i, j}$, we have $3 \leq |\gap{w_i}{e}{j}| \leq 9$, which means that $e$ satisfies the length constraints and therefore $p \in \subseqSet{\gaptuple}{T}$.
\end{proof}

Finally, we define the string
\begin{equation*}
S = T \tb (\ta \tb)^5 w_1 \tb (\ta \tb)^5 w_2 \tb (\ta \tb)^5\ldots \tb (\ta \tb)^5 w_q\,.
\end{equation*}

\begin{lemma}\label{limitedToSJLemmaBinaryCaseAppendix}
Let $p \in (\subseqSet{\gaptuple}{S} \cap L_{\Sigma\Sigma})$. Then, for some $i \in [q]$, $p \in \subseqSet{\gaptuple}{w_i}$.
\end{lemma}

\begin{proof}
We assume that $p \in (\subseqSet{\gaptuple}{S} \cap L_{\Sigma\Sigma})$, i.\,e., there is an embedding $e$ that satisfies $\gaptuple$ and $p \subseq_e S$. All factors $w_i$ and $T$ start and end with occurrences of symbol $\ta$, while the separating factors $\tb (\ta \tb)^5$ in between start and end with occurrences of symbol $\tb$. This means that all positions of $p$ must be mapped by $e$ to a position of some factor $w_i$ or to a position of factor $T$. Due to the upper bounds $9$ of the length constraints $C'_j$ for every $j \in [k-1]$, it is not possible that $e$ maps positions of $p$ to different such factors; thus, all positions of $p$ are mapped to positions of some factor $w_i$ or to factor $T$. However, due to Lemma~\ref{sharpUnivStringLemmaBinaryCaseAppendix}, $\subseqSet{\gaptuple}{T} = (\overline{L_{\Sigma \Sigma}})$, so since $p \in L_{\Sigma \Sigma}$ it is not possible that $e$ maps $p$ to $T$. Hence, $p \in \subseqSet{\gaptuple}{w_i}$.
\end{proof}

\begin{lemma}\label{mainCorrectnessLemmaBinaryAppendix}
$\subseqSet{\gaptuple}{S} = \{\ta, \tb\}^{2k} \iff \cup_{i \in [q]} \lang{W_{i}} = L_{\Sigma \Sigma}$.
\end{lemma}

\begin{proof}
We first show that
\begin{equation*}
\subseqSet{\gaptuple}{S} = \bigcup_{i \in [q]} (\subseqSet{\gaptuple}{w_i} \cap L_{\Sigma \Sigma}) \cup \subseqSet{\gaptuple}{T}\,.
\end{equation*}
The ``$\supseteq$''-direction holds, since $T$ and every $w_i$ for every $i \in [q]$ are factors of $S$. Now let $p \in \subseqSet{\gaptuple}{S}$  be arbitrarily chosen. If $p \notin L_{\Sigma \Sigma}$, then, by Lemma~\ref{sharpUnivStringLemmaBinaryCaseAppendix}, $p \in \subseqSet{\gaptuple}{T}$. If $p \in L_{\Sigma \Sigma}$, then Lemma~\ref{limitedToSJLemmaBinaryCaseAppendix} implies that $p \in \subseqSet{\gaptuple}{w_i}$ for some $i \in [q]$, which means that $p \in \subseqSet{\gaptuple}{w_i} \cap L_{\Sigma \Sigma}$. This shows that the ``$\subseteq$''-direction holds as well.\par 
Now we prove the \emph{if} direction of the statement of the lemma and assume that $\cup_{i \in [q]} \lang{W_i} = L_{\Sigma\Sigma}$. By Lemma~\ref{SingleWiLemmaBinaryCaseAppendix}, we know that 
$(\subseqSet{\gaptuple}{w_{i}} \cap L_{\Sigma \Sigma}) = \lang{W_i}$ for every $i \in [q]$. With Lemma~\ref{sharpUnivStringLemmaBinaryCaseAppendix} and the observation from above, this means that 
\begin{align*}
\subseqSet{\gaptuple}{S} &= \bigcup_{i \in [q]} (\subseqSet{\gaptuple}{w_i} \cap L_{\Sigma \Sigma}) \cup \subseqSet{\gaptuple}{T}\\
&= \left(\bigcup_{i \in [q]} \lang{W_i}\right) \cup \subseqSet{\gaptuple}{T}\\
&= L_{\Sigma \Sigma} \cup (\overline{L_{\Sigma \Sigma}}) = \{\ta, \tb\}^{2k}\,.
\end{align*}\par
In order to prove the \emph{only if} direction, we assume that $\subseqSet{\gaptuple}{S} = \{\ta, \tb\}^{2k}$. With our observation from above, this means that
\begin{equation*}
\{\ta, \tb\}^{2k} = \bigcup_{i \in [q]} (\subseqSet{\gaptuple}{w_i} \cap L_{\Sigma \Sigma}) \cup \subseqSet{\gaptuple}{T} 
\end{equation*}
Applying Lemma~\ref{sharpUnivStringLemmaBinaryCaseAppendix} and the fact that $\bigcup_{i \in [q]} (\subseqSet{\gaptuple}{w_i} \cap L_{\Sigma \Sigma})$ and $\overline{L_{\Sigma \Sigma}}$ are clearly disjoint yields
\begin{align*}
\{\ta, \tb\}^{2k} &= \bigcup_{i \in [q]} (\subseqSet{\gaptuple}{w_i} \cap L_{\Sigma \Sigma}) \cup (\overline{L_{\Sigma \Sigma}}) & \iff\\
L_{\Sigma \Sigma} &= \bigcup_{i \in [q]} (\subseqSet{\gaptuple}{w_i} \cap L_{\Sigma \Sigma})\,. &
\end{align*}
Finally, Lemma~\ref{SingleWiLemmaBinaryCaseAppendix} implies that $L_{\Sigma \Sigma} = \bigcup_{i \in [q]} \lang{W_i}$.
\end{proof}

\subsection{Proof of Theorem~\ref{HardnessNonUniversalityBounded}}

\begin{theorem}\label{HardnessNonUniversalityBoundedAppendix}
For every fixed alphabet $\Sigma$ with $|\Sigma| \geq 3$, $\nuniProb_{\Sigma}$ with length constraints is $\npclass$-complete, even if all length constraints are $(1, 5)$. Moreover,
\begin{itemize}
\item it cannot be solved in subexponential time $2^{\smallO(k)} \poly(|w|, k))$ (unless ETH fails), 
\item it cannot be solved in time $\bigO(2^{k(1-\epsilon)} \poly(|w|, k))$ (unless SETH fails).
\end{itemize}
\end{theorem}

\begin{proof}
We first show that $\nuniProb$ with length constraints is in $\npclass$. Let $w \in \Sigma^*$, let $k \in \mathbb{N}$ and let $\gaptuple$ be some tuple of length constraints of size $k-1$. In order to check whether $\subseqSet{\gaptuple}{w} \neq \Sigma^k$, we guess a string $p \in \Sigma^k$ in polynomial time and then check whether $p \subseq_{\gaptuple} w$. According to Corollary~\ref{constantPatternsRL}, this can be done in polynomial time. \par
The $\npclass$-hardness and the conditional lower bounds follow from the reduction obtained by plugging together the reductions described in Sections~\ref{sec:SatReduction}~and~\ref{sec:theReduction}. Let $F = \{c_1, c_2, \ldots, c_q\}$ be an instance of SAT over some set of variables $V = \{v_1, v_2, \ldots, v_k\}$. We can then use the reduction from Section~\ref{sec:SatReduction} in order to reduce this SAT instance into a $\metaNuniProb$ instance $(W_{i, j})_{1 \leq i \leq q, 1 \leq j \leq k}$ over a binary alphabet $\{0, 1\}$. After that, we use the reduction of Section~\ref{sec:theReduction} in order to reduce the $\metaNuniProb$ instance $(W_{i, j})_{1 \leq i \leq q, 1 \leq j \leq k}$ over $\{0, 1\}$ into a $\nuniProb_{\Sigma}$ instance $\gaptuple = (C_1, C_2, \ldots, C_{k-1})$ and $w$ over alphabet $\Sigma = \{0, 1, \#\}$. Moreover, $C_i = (\lowerBoundShort{j}, \upperBoundShort{j}) = (m-1, 3m-1)$ for every $j \in [k-1]$, where $m$ is the alphabet size of the $\metaNuniProb$ instance, which is $|\{0, 1\}| = 2$. This means that all length constraints are $(m-1, 3m-1) = (1, 5)$. This proves that $\nuniProb_{\Sigma}$ with length constraints is $\npclass$-complete, even if $|\Sigma| = 3$ and all length constraints are $(1, 5)$. We can furthermore note that $|w| = \poly(k, m, q) = \poly(k + q)$. \par
In order to prove the ETH bound, assume that $\nuniProb_{\Sigma}$ with $|\Sigma| = 3$ and only length constraints $(1, 5)$ can be solved in time $2^{\smallO(k)} \poly(|w|, k))$. Let $F = \{c_1, c_2, \ldots, c_q\}$ be an instance of 3-CNF-SAT over some set of variables $V = \{v_1, v_2, \ldots, v_{k'}\}$. If we use the above reduction on $F$, then we get a $\nuniProb_{\Sigma}$ instance $(\gaptuple, w)$ with $|\gaptuple| = {k'}-1$ and $|w| = \poly({k'} + q)$. By assumption, this instance can be solved in time $2^{\smallO(k')} \poly(|w|, k'))$, which means that 3-CNF-SAT can be solved in time $2^{\smallO(k')} \poly(k' + q))$. Hence, ETH fails.\par
In order to prove the SETH bound, assume that $\nuniProb_{\Sigma}$ with $|\Sigma| = 3$ and only length constraints $(1, 5)$ can be solved in time$\bigO(2^{k(1-\epsilon)} \poly(|w|, k))$ for some $\epsilon > 0$. Let $F = \{c_1, c_2, \ldots, c_q\}$ be an instance of CNF-SAT over some set of variables $V = \{v_1, v_2, \ldots, v_{k'}\}$. If we use the above reduction on $F$, then we get a $\nuniProb_{\Sigma}$ instance $(\gaptuple, w)$ with $|\gaptuple| = k'-1$ and $|w| = \poly(k' + q)$. By assumption, this instance can be solved in time $\bigO(2^{k'(1-\epsilon)} \poly(|w|, k'))$ for some $\epsilon > 0$, which means that CNF-SAT can be solved in time $\bigO(2^{k'(1-\epsilon)} \poly(k' + q))$. Hence, SETH fails.
\end{proof}

\subsection{Proof of Theorem~\ref{binaryCaseHardnessTheorem}}

\begin{theorem}\label{binaryCaseHardnessTheoremAppendix}
For every fixed alphabet $\Sigma$ with $|\Sigma| = 2$, $\nuniProb_{\Sigma}$ with length constraints is $\npclass$-complete, even if each length constraint is $(0, 0)$ or $(3, 9)$. 
\end{theorem}

\begin{proof}
This follows directly from the reduction described in Sections~\ref{sec:binaryCaseReduction}. 
\end{proof}

\subsection{Proof of Theorem~\ref{HardnessNonUniversalityUnbounded}}

\begin{theorem}\label{HardnessNonUniversalityUnboundedAppendix}
The problem $\nuniProb$ with length constraints cannot be solved in running time $\bigO(f(k) \poly(|w|, k))$ for any computable function $f$ (unless $\fptclass = \wclass[1]$). 
\end{theorem}

\begin{proof}
We observe that the reduction that results from plugging together the reductions described in Sections~\ref{sec:kCliqueReduction}~and~\ref{sec:theReduction}) is a parameterised reduction from $\kISProb$ to $\nuniProb$ with length constraints parameterised by $k$. Indeed, the reduction from Section~\ref{sec:kCliqueReduction} transforms an instance $(G = (V, E), k)$ of $\kISProb$ into an $\metaNuniProb$ instance over an alphabet $\Gamma$ of cardinality $|V|$ and with a matrix of dimensions $|V|k(k-1)$ and $k$. Then, the reduction from Section~\ref{sec:theReduction} transforms this $\metaNuniProb$ instance into an instance of $\nuniProb$ with length constraints, where parameter $k$ corresponds to parameter $k$ of the original $\kISProb$ instance. Consequently, the reduction is a parameterised reduction from $\kISProb$ to $\nuniProb$ with length constraints parameterised by $k$. Therefore, $\nuniProb$ with length constraints parameterised by $k$ is $\wclass[1]$, which yields the statement of the theorem.
\end{proof}

\subsection{Containment and Equivalence}\label{sec:ContEquiAppendix}

We now show that the statements of Theorems~\ref{binaryCaseHardnessTheorem},~\ref{HardnessNonUniversalityBounded},~and~\ref{HardnessNonUniversalityUnbounded} also hold for the problems $\ncontProb$ and $\nequiProb$.\par
We recall the definition of the \emph{con-containment} ($\ncontProb$) and the \emph{non-equivalence problem for gap constrained subsequences} ($\nequiProb$): Given a $(k-1)$-tuple of gap constraints $\gaptuple$ and strings $w, w'$, decide $\subseqSet{\gaptuple}{w} \not \subseteq \subseqSet{\gaptuple}{w'}$, or $\subseqSet{\gaptuple}{w} \neq \subseqSet{\gaptuple}{w'}$, respectively.\par
We now extend the reduction from Section~\ref{sec:theReduction} to a reduction from $\metaNuniProb$ to $\nequiProb$. Let $\Gamma = \{b_1, b_2, \ldots, b_m\}$, $q, k \in \mathbb{N}$, and, for every $i \in [q], j \in [k]$, let $W_{i, j} \subseteq \Gamma$. Moreover, let $\gaptuple$ be the tuple of length constraints and let $K(W_1, \ldots, W_q)$ be the string over $\Sigma = \Gamma \cup \{\#\}$ constructed by the reduction from Section~\ref{sec:theReduction}. \par
We now define a string that is universal for $\gaptuple$. For every $i \in [k]$, let $T' = T'_{1} T'_{2} \ldots T'_{k}$, where, for every $j \in [k]$, $T'_{j} = b_1 b_2 \ldots b_m \#^{m}$.

\begin{lemma}\label{UnivStringLemmaAppendix}
$\subseqSet{\gaptuple}{T'} = \Sigma^k$.
\end{lemma}

\begin{proof}
Since $T'$ is a string over $\Sigma$, we have $\subseqSet{\gaptuple}{T'} \subseteq \Sigma^k$. Now let $p \in \Sigma^k$ be arbitrarily chosen, and let $e$ be the embedding that is defined as follows. For every $j \in [k]$, we map $j$ to the $t^{\text{th}}$ position of $T'_{j}$ if $p[j] = b_t$, and to the $(m + 1)^{\text{st}}$ position of $T'_{j}$ if $p[j] = \#$. We observe that this embedding $e$ satisfies that every $j$ is mapped to an occurrence of $\pi[j]$ of $T'_{j}$. Moreover, since every $j \in [k]$ is mapped to one of the first $m + 1$ occurrences of the length-$2m$ factor $T'_{j}$, it can be easily seen that $m - 1 \leq |\gap{T'}{e}{j}| \leq 3m - 1$. Hence, $e$ satisfies the length constraints, and therefore $p \in \subseqSet{\gaptuple}{T'}$. This means that $\Sigma^k \subseteq \subseqSet{\gaptuple}{T'}$.
\end{proof}

\begin{lemma}\label{mainCorrectnessLemmaEquiAppendix}
$\subseqSet{\gaptuple}{K(W_1, \ldots, W_q)} = \subseqSet{\gaptuple}{T'} \iff \cup_{i \in [q]} \lang{W_{i}} = \Gamma^k$.
\end{lemma}

\begin{proof}
\begin{align*}
&\subseqSet{\gaptuple}{K(W_1, \ldots, W_q)} = \subseqSet{\gaptuple}{T'}& &\overset{\text{Lem.~\ref{UnivStringLemmaAppendix}}}{\iff}&\\
&\subseqSet{\gaptuple}{K(W_1, \ldots, W_q)} = \Sigma^k& &\overset{\text{Lem.~\ref{mainCorrectnessLemmaAppendix}}}{\iff}&\\
&\cup_{i \in [q]} \lang{W_{i}} = \Gamma^k\,.& &&
\end{align*}
\end{proof}

\begin{theorem}\label{HardnessNonUniversalityBoundedAppendixEqui}
For every fixed alphabet $\Sigma$ with $|\Sigma| \geq 3$, $\nequiProb_{\Sigma}$ and $\ncontProb_{\Sigma}$ with length constraints are $\npclass$-complete, even if all length constraints are $(1, 5)$. Moreover,
\begin{itemize}
\item it cannot be solved in subexponential time $2^{\smallO(k)} \poly(|w|, |w'|, k))$ (unless ETH fails), 
\item it cannot be solved in time $\bigO(2^{k(1-\epsilon)} \poly(|w|, |w'|, k))$ (unless SETH fails).
\end{itemize}
\end{theorem}

\begin{proof}
This result follows analogously as in the proof of Theorem~\ref{HardnessNonUniversalityBoundedAppendix}, i.\,e., by the reduction obtained by plugging together the reductions described in Sections~\ref{sec:SatReduction} and the extension of the reduction of Section~\ref{sec:theReduction} described above. That this reduction satisfies that $|\Sigma| \geq 3$ and that all length constraints are $(1, 5)$ follows in exactly the same way as in the proof of Theorem~\ref{HardnessNonUniversalityBoundedAppendix}.\par
The conditional lower bounds also follow analogously as in the proof of Theorem~\ref{HardnessNonUniversalityBoundedAppendix}. We only have to observe that the universal string $T'$ is polynomial in $k$, and that the adapted reduction still produces a $\nequiProb_{\Sigma}$ with $|\gaptuple| = k-1$, where $k$ is the number of Boolean variables.
\end{proof}

\begin{theorem}\label{HardnessNonUniversalityUnboundedAppendixEqui}
The problems $\nequiProb$ and $\ncontProb$ with length constraints cannot be solved in running time $\bigO(f(k) \poly(|w|, k))$ for any computable function $f$ (unless $\fptclass = \wclass[1]$). 
\end{theorem}

\begin{proof}
Analogously to the proof of Theorem~\ref{HardnessNonUniversalityUnboundedAppendix}, we observe that plugging together the reduction described in Sections~\ref{sec:kCliqueReduction} and the extension of the reduction from Section~\ref{sec:theReduction}) described above yields a parameterised reduction from $\kISProb$ to $\nequiProb$ with length constraints parameterised by $k$. 
\end{proof}

Next, we extend in a similar way the reduction from Section~\ref{sec:binaryCaseReduction} to a reduction from $\SatProb$ to $\nequiProb$. Let $F = \{c_1, c_2, \ldots, c_q\}$ be a Boolean formula in CNF on variables $\{v_1, v_2, \ldots, v_k\}$. Moreover, let $\gaptuple$ be the tuple of length constraints and let $S$ be the string over $\Sigma = \{\ta, \tb\}$ constructed by the reduction from Section~\ref{sec:binaryCaseReduction}. We now prove that the string $T' = ((\ta \ta \tb \tb \ta) \tb^3)^k$ is universal for $\gaptuple$. 

\begin{lemma}\label{UnivStringLemmaBinaryAppendix}
$\subseqSet{p}{T'} = \{\ta, \tb\}^{2k}$.
\end{lemma}

\begin{proof}
We first note that $\subseqSet{p}{T'} \subseteq \{\ta, \tb\}^{2k}$ obviously holds. In order to show the other direction, let $p \in \{\ta, \tb\}^{2k}$. We define an embedding $e$ as follows. For every $j \in [k]$, $e$ maps $2j-1$ and $2j$ to the factor $p[2j-1]p[2j]$ in the $j^{\text{th}}$ occurrence of factor $\ta \ta \tb \tb \ta$. We observe that $3 \leq \gap{T'}{e}{j} \leq 9$, which means that $e$ satisfies the length constraints, and therefore $p \in \subseqSet{p}{T'}$. 
\end{proof}

\begin{lemma}\label{mainCorrectnessLemmaBinaryAppendixEqui}
$\subseqSet{\gaptuple}{S} = \subseqSet{\gaptuple}{T'} \iff \cup_{i \in [q]} \lang{W_{i}} = L_{\Sigma \Sigma}$.
\end{lemma}

\begin{proof}
\begin{align*}
&\subseqSet{\gaptuple}{S} = \subseqSet{\gaptuple}{T'}& &\overset{\text{Lem.~\ref{UnivStringLemmaBinaryAppendix}}}{\iff}&\\
&\subseqSet{\gaptuple}{S} =  \{\ta, \tb\}^{2k}& &\overset{\text{Lem.~\ref{mainCorrectnessLemmaBinaryAppendix}}}{\iff}&\\
&\cup_{i \in [q]} \lang{W_{i}} = L_{\Sigma \Sigma}\,.& &&\qedhere
\end{align*}
\end{proof}

\begin{theorem}\label{binaryCaseHardnessTheoremAppendixEqui}
For every fixed alphabet $\Sigma$ with $|\Sigma| = 2$, $\nequiProb_{\Sigma}$ and $\ncontProb_{\Sigma}$ with length constraints is $\npclass$-complete, even if each length constraint is $(0, 0)$ or $(3, 9)$. 
\end{theorem}

\begin{proof}
This follows directly from the extension of the reduction from Sections~\ref{sec:binaryCaseReduction} described above. 
\end{proof}

\section{Special Variants}\label{sec:Special}

We next consider two natural variants of our setting that have a substantial impact on the complexity of the problems investigated above. \par

\textbf{Gap Length Equalities}: We investigate whether the polynomiality of the matching problem (see Section~\ref{sec:matching}) is preserved under adding \emph{gap length equalities} to the gap constraints, i.\,e., constraints of the form $|\gapName_i|=|\gapName_j|$ which are satisfied by an embedding $e$ with respect to $w$ if $|\gap{w}{e}{i}|=|\gap{w}{e}{j}|$. Our main motivation is that such length equality constraints (and more complex ones, e.\,g., described by linear inequalities like $2|\gapName_7| + |\gapName_3| \leq |\gapName_2|$) are of interest in the theory of string solving \cite{amadini2021survey}. Unfortunately, the matching problem becomes immediately NP-hard (the following result can be shown by adapting the NP-completeness proof for matching patterns with variables from~\cite{Angluin80}).

\begin{theorem}\label{constantPatternsLengthExtended}
$\matchProb$ with length constraints and gap length equalities is $\npclass$-complete, even for binary alphabets and length constraints $(0,+\infty)$.
\end{theorem}
\begin{proof}

We first define the investigated problem variant more formally. The \emph{matching problem with length constraints and gap length equalities} is defined as follows: Given a gapped sequence $(p, \gaptuple)$ with length constraints, a word $w$, and a finite set of equations ${\mathcal S}$ of the form $|\gapName_i|=|\gapName_j|$ with $i, j \in [|\gaptuple|-1]$, decide whether there exists an embedding $e$ satisfying $\gaptuple$ such that $p \subseq_e w$ and $|\gap{w}{e}{i}|=|\gap{w}{e}{j}|$ for each equation $|\gapName_i|=|\gapName_j|$ of ${\mathcal S}$. \par

	We can now continue with the proof of the theorem. The containment in NP is immediate. We only show the lower bound.
	
	To simplify the exposition, just as in the case of Theorem \ref{lowerBoundLength}, when representing the gapped sequence $(p, \gaptuple)$ with $p=p[1]\cdots p[m]$ and a finite set of equations ${\mathcal S}$ of the form $|\gapName_i|=|\gapName_j|$, we will use the following notations. Firstly, $(p,\gaptuple) = p[1] \stackrel{\gaptuple[1], \ell_1}{\leftrightarrow} p[2] \stackrel{\gaptuple[2], \ell_2}{\leftrightarrow} \cdots \stackrel{\gaptuple[m-1], \ell_n}{\leftrightarrow} p[m]$, where $\ell_1, \ldots, \ell_n$ are symbols from a finite set $X$ of labels. We omit the symbol $\stackrel{\gaptuple[i]}{\leftrightarrow}$ from the notation of $(p,\gaptuple)$ if and only if $\gaptuple[i]=(0,0)$. In this proof, however, all length constraints are trivial, i.\,e., they have either the form $\gaptuple[i]=(0,+\infty)$; in this case, we simply write $\leftrightarrow$ for $\stackrel{\gaptuple[i]}{\leftrightarrow}$. In the context of this proof, and differently from the proof of Theorem \ref{lowerBoundLength}, the gaps will have labels. We use these labels to encode the equations in the set ${\mathcal S}$: the $i^{th}$ gap and $j^{th}$ gap have the same label if and only if $|\gapName_i|=|\gapName_j|$ is a gap length equality of $S$. When, $\stackrel{\gaptuple[i]}{\leftrightarrow}$ has label $x$ we write $\stackrel{\gaptuple[i],x}{\leftrightarrow}$; as such, in this proof we only need to keep track of the labels of the gaps, and denote a gap as $\stackrel{x}{\leftrightarrow}$, where $x$ is the label.

	We reduce 3-CNF-SAT to the extended matching problem.
	Let $\Phi = C_1 \wedge C_2 \wedge \ldots \wedge C_m$ be a Boolean formula in CNF
	with variables $V_1, V_2, \ldots, V_n$,
	where each clause is a conjunction of three literals,
	each of which is either a variable or the negation of a variable.
	Without loss of generality we choose two symbols $0, 1\in \Sigma$.
	Now, we construct from $\Phi$ a gapped sequence $(p,\gaptuple)$ with length constraints,
	a string $s$ and a set of equations $\mathcal{S}$ as follows.

	We first define a set of $4n+2m$ labels: $x_i, x'_i, y_i, y'_i, z_j, z'_j$, 
	where $1 \leq i \leq n$ and $1 \leq j \leq m$.
	We define for $j \in [m]$ and $k \in [3]$
	
	$f(j, k)$ = $\begin{cases}
		x_i \qquad \text{ if the $k$th literal in $C_j$ is $V_i$, } \\
		y_i \qquad \text{ if the $k$th literal in $C_j$ is $\overline{V_i}$.}
	\end{cases}$
	
	We first define the gapped sequence $(p',\gaptuple')$ as 
	\[(p',\gaptuple') = \left(\prod\limits_{i=1}^{n}0 \stackrel{x_i}{\leftrightarrow}1\stackrel{x'_i}{\leftrightarrow}\right) \left(\prod\limits_{i=1}^{n}0 \stackrel{y_i}{\leftrightarrow}1\stackrel{y'_i}{\leftrightarrow}\right)
	\left(\prod\limits_{j=1}^{m}0 \stackrel{z_j}{\leftrightarrow}1\stackrel{z'_j}{\leftrightarrow}\right)
	\]
	
	Now, we define the gapped sequence $(p'',\gaptuple'')$ as 
	\[(p'',\gaptuple'') = 0 p_1 0 p_2 \cdots 0 p_n 0 q_1 0 q_2 \cdots 0 q_m 0\]
	with $p_i = \stackrel{x_i}{\leftrightarrow} 1\stackrel{y_i}{\leftrightarrow}$ and
	$q_j = \stackrel{f(j, 1)}{\leftrightarrow} 1 \stackrel{f(j, 2)}{\leftrightarrow} 1 \stackrel{f(j, 3)}{\leftrightarrow} 1 \stackrel{z_j}{\leftrightarrow}$. 
	for $1 \leq i \leq n$ and $1 \leq j \leq m$.
	
	Let $(p,\gaptuple)=(p',\gaptuple')(p'',\gaptuple'')$. 
	
	The system ${\mathcal S}$ is induced by the labels of the gaps: two gaps have the same label if and only if they have the same length.

	We now define the string $w'$ as \[w' = \left(\prod\limits_{i=1}^{n}0 d_i\right) \left(\prod\limits_{i=1}^{n}0 e_i \right)\left(\prod\limits_{j=1}^{m}0 f_j\right),
	\]
	where $d_i=e_i=111=1^3$,
	and $f_j=1111=1^4$, 
	for $1 \leq i \leq n$ and $1 \leq j \leq m$.

	Finally, the string $w''$ is constructed as
	\[w'' = 0 s_1 0 s_2 \cdots 0 s_n 0 t_1 0 t_2 \cdots 0 t_m 0\]
	with $s_i = 1111=1^4$ and
	$t_j = 1111111111=1^{10}$, 
	for $1 \leq i \leq n$ and $1 \leq j \leq m$.

Let $w=w'w''$.
	
	It is easy to see that the construction of gapped sequence $(p,\gaptuple)$, $w$, and equation system ${\mathcal S}$ 
	can be done in polynomial time.
	We now show that there exists an embedding $e$ satisfying $\gaptuple$ 
	such that $p \subseq_e w$
	and $\len{\gap{w}{e}{i}} = \len{\gap{w}{e}{j}}$
	for each equation $\len{\gapName_i} = \len{\gapName_j}$ of $\mathcal{S}$ (or in other words, if the respective gaps have the same label) if and only if the formula $\Phi$ is satisfiable.
	
	Firstly, assume that $\Phi$ is satisfiable.
	Let $\alpha \colon \{V_1, V_2, \ldots, V_n\} \rightarrow \{0, 1\}$
	be an assignment satisfying $\Phi$.
	
	We define the following mapping which defines an embedding of $p$ into $w$. 
	
	First of all, the $i^{th}$ symbol $0$ of $p$ will be mapped to the $i^{th}$ symbol $0$ of the word $w$. 
	
We then explain how we map the symbols of $p'$ found between occurrences of $0$. The factor $\stackrel{x_i}{\leftrightarrow}1 \stackrel{x'_i}{\leftrightarrow}$ is mapped to the factor $d_i=1^3$ of $w$, and the  factor $\stackrel{y_i}{\leftrightarrow}1 \stackrel{y'_i}{\leftrightarrow}$ is mapped to the factor $e_i$ of $w$. If $\alpha(V_i) = 1$, then let the length of the gap with the label $x_i$ in the embedding we construct be $2$ and the length of the gap with label $y_i$ be $1$. Clearly, in this case, the gap labelled $x'_i$ will have length $0$, and the gap with label $y'_i$ will have length $1$.
If $\alpha(V_i) = 0$, then let the length of the gap with the label $x_i$ be exactly $1$ and the length of the gap with label $y_i$ be exactly $2$. The gap labelled $x'_i$ will then have length $1$, and the gap with label $y'_i$ will have length $0$. The way we assign the length of the gaps labelled with $z_j$ will follow from the explanations below. 

Now, we see how to map the symbols of $p''$ found between $0$-symbols. 

For $1 \leq i \leq n$, each gapped sequence $p_i$ is mapped in the string $s_i=1^4$ of $w$. From the way the symbols of $p'$ are mapped to symbols of $w$, the mapping of $p_i$ to $s_i$ is already determined: the total length of the gaps labelled with $x_i$ and $y_i$ is $3$, which of these gaps has length $2$ and which has length $1$ is determined by the value of $\alpha(V_i)$, and the exact mapping of the single $1$-symbol of $p_i$ to a $1$-symbol of $w$ is, as such, also determined by $\alpha(V_i)$. 
	
Further, we see how the gapped sequences $q_j$ are embedded in $w$. These gapped sequences are embedded in the factors $t_j=1^{10}$ of $w$. So let us consider some  $j \in [m]$, and let $C_j = (t_1 \vee t_2 \vee t_3)$ be the corresponding clause with the three literals $t_1, t_2, t_3$. Since $\alpha$ satisfies $\Phi$, at least one of the three literals $t_1, t_2, t_3$ is assigned the value $1$. Therefore, at least one of the gaps with labels $f(j, 1), f(j, 2), f(j, 3)$ should have been assigned the length $2$ in the previous steps  (where the length of each of the gaps with labels $x_i$ or $y_i$ was decided). This means that, altogether,
the length of the gaps with labels $f(j, 1), f(j, 2), f(j, 3)$ from $q_j$, together with the $1$-symbols occurring after each of these gaps (so three $1$ symbols in total), is either $7, 8,$ or $9$. To complement, the length of the gap with the label $z_j$ is either $3$, $2$, or $1$, respectively. The embedding maps the $1$-symbols occurring in $q_j$ according to the length of the gaps, which once more is essentially determined by the values of the boolean-variables occurring in $\Phi$. Moreover, the length we assign here to the gap with the label $z_j$ induces the length of the gaps labelled with $z'_j$ (and the position of $w$ to which the $1$ symbol from $\stackrel{z_j}{\leftrightarrow}1 \stackrel{z'_j}{\leftrightarrow}$ is mapped in $f_j$). 
	
So, if $\Phi$ is satisfiable, then an embedding of $p$ in $w$ which satisfies both the constraints $\gaptuple$ and the system ${\mathcal S}$ (as induced by the labels of the gaps) is possible.

Now suppose that we have an embedding $e$ such that $p \subseq_e w$ and $e$ satisfies both the gap constraints $\gaptuple$ and the equations of ${\mathcal S}$ (induced by the labels in the definition of $(p,\gaptuple)$). We want to show that there exists an assignment of the variables $\{V_1,\ldots,V_n\}$ which satisfies $\Phi$. 

It is not hard to note that, due to the fact that $w$ and $p$ have exactly the same number $0$ symbols, the embedding $e$ maps the $i^{th}$ $0$-symbol of $p$ to the $i^{th}$ $0$-symbol of the word $w$. 

Therefore, for all $i\in [n]$, $\stackrel{x_i}{\leftrightarrow}1 \stackrel{x'_i}{\leftrightarrow}$ must be mapped to $d_i=1^3$. This means that the gaps labelled with $x_i$ may have length $0, 1$ or $2$.  Similarly, $\stackrel{y_i}{\leftrightarrow}1 \stackrel{y'_i}{\leftrightarrow}$ must be mapped to $e_i=1^3$, so the gaps labelled with $y_i$ may have length $0, 1$ or $2$, for $i\in [n]$. Finally, $\stackrel{z_i}{\leftrightarrow}1 \stackrel{z'_i}{\leftrightarrow}$ must be mapped to $f_i=1^4$. This means that the gaps labelled with $z_i$ may have length $0, 1, 2$ or $3$.  

Also, we have that $p_i$ must be embedded in $s_i$, for $i\in [n]$, and $q_j$ must be embedded in $t_j$, for $j\in [m]$. 

From the fact that $p_i$ is embedded in $s_i$, for $i\in [n]$, we essentially get an assignment $\ell_{x_i}$ for the length of the gaps with label $x_i$ and an assignment $\ell_{y_i}$ for the length of the gaps with label $y_i$, such that $\ell_{x_i}+\ell_{y_i}=3$. We have already seen that $\ell_{x_i}\leq 2$ and $\ell_{y_i}\leq 2$. So, $\ell_{x_i}\in \{1,2\}$ and $\ell_{y_i}=3-\ell_{x_i}$. If $\ell_{x_i}=2$ then we set $V_i=1$ and, otherwise, we set $V_i=0$. 

Further, we have that $q_j$ is embedded in $t_j$, and we already have an assignment for the lengths $\ell_{x_i}$ and $\ell_{y_i}$, for all $i\in [n]$. In order for $q_j$ to be embedded in $t_j=1^{10}$, as the length $\ell_{z_j}$ of the gap with label $z_j$ is at most $3$, we obtain that the gapped sequence $\stackrel{(1,2),f(j, 1)}{\leftrightarrow} 1 \stackrel{(1,2),f(j, 2)}{\leftrightarrow} 1 \stackrel{(1,2),f(j, 3)}{\leftrightarrow} 1 $ is mapped to a string $1^\ell$ with $\ell\geq 7$. This means that at least one of the gaps  $\stackrel{(1,2),f(j, 1)}{\leftrightarrow}, \stackrel{(1,2),f(j, 2)}{\leftrightarrow},$ or $\stackrel{(1,2),f(j, 3)}{\leftrightarrow} $ is mapped to a string of length $2$. As these gaps have labels from the set $\{x_i,y_i\mid i\in [n]\}$, corresponding to the variables occurring in $C_j$, it follows easily that at least one of the literals that appear in $C_j$ should be set to $1$. As such, the disjunction $C_j$ of $\Phi$ evaluates to 1.

Since the previous remark holds for all $j$, it follows that the assignment of the variables induced by the choice of the lengths $\ell_{x_i}$ and $\ell_{y_i}$, for all $i\in [n]$, satisfies $\Phi$. 

In conclusion, there exists an embedding $e$ satisfying $\gaptuple$ such that $p \subseq_e w$ and $\len{\gap{w}{e}{i}} = \len{\gap{w}{e}{j}}$ for each equation $\len{\gapName_i} = \len{\gapName_j}$ of $\mathcal{S}$ (as induced by the labels of the gaps) if and only if the formula $\Phi$ is satisfiable. This shows that our reduction is correct, and, thus, the statement follows. 
\end{proof}

\textbf{Gap Constrained Subsequences With Multiplicities:} With respect to the equivalence problem, we can show a positive result for the following modified setting. Let us consider the $\gaptuple$-subsequences of the sets $\subseqSet{\gaptuple}{w}$ with multiplicities. For example, for $w_1 = \ta \tb \tb \ta$ and $w_2 = \ta \tb \ta \tb$, we have $\subseqSet{\gaptuple_2}{w_1} = \subseqSet{\gaptuple_2}{w_2} = \{\ta \ta, \ta \tb, \tb \ta, \tb \tb\}$ with $\gaptuple_2 = (\Sigma^*)$. There is exactly one way of embedding $\ta \ta$ and $\tb \tb$ into both $w_1$ and $w_2$. On the other hand, $\ta \tb$ can be embedded into $w_1$ in two different ways and into $w_2$ in three different ways. More precisely, the sets of $\gaptuple_2$-subsequences of $w_1$ and $w_2$ with multiplicities are $\{(\ta \ta, 1), (\ta \tb, 2), (\tb \ta, 2), (\tb \tb), 1\}$ and $\{(\ta \ta, 1), (\ta \tb, 3), (\tb \ta, 1), (\tb \tb), 1\}$, respectively.\par 
Let us now formalise this setting. For strings $u$ and $v$, and a $(|u|-1)$-tuple $\gaptuple$ of gap constraints, we denote by $\distinctgcsubseq{u}{v}{\gaptuple}$ the number of distinct embeddings $e : |u| \to |v|$ that satisfy $\gaptuple$ and $v \subseq_e u$. For example, $\distinctgcsubseq{\tb \tb \ta \ta}{\tb \ta}{\gaptuple_2} = 4$, as $u[1]u[3] = u[2]u[3] = u[1]u[4] = u[2]u[4] = \tb \ta$. For any $(k-1)$-tuple $\gaptuple$ of gap constraints, we define the function $\parikhK{\cdot}{\gaptuple} \colon \Sigma^* \to \mathbb{N}^{(\Sigma^k)}$ by $\parikhK{w}{\gaptuple}[p] = \distinctgcsubseq{w}{p}{\gaptuple}$ for every $p \in \Sigma^k$. 

The \emph{equivalence problem with multiplicities} is to decide, for a given $(k-1)$-tuple $\gaptuple$ of gap constraints, and strings $w, w' \in \Sigma^*$, whether $\parikhK{w}{\gaptuple} = \parikhK{w'}{\gaptuple}$. Note that for the case $\gaptuple=(\Sigma^*, \ldots, \Sigma^*)$ this is called the $k$-binomial equivalence, and was studied in the area of combinatorics on words (see, e.\,g., \cite{RigoS15,Rigo19,LeroyRS17a,Freydenberger2015}). \par
We show that equivalence with multiplicities can be decided in polynomial time (in contrast to the $\npclass$-completeness of the case without multiplicities).

\begin{theorem}\label{thm:binom-length-constr}
If $\gaptuple=(C_1,\ldots,C_{k-1})$ and $C_i$ can be decided in polynomial time then the equivalence problem with multiplicities can be solved in polynomial time.\looseness=-1
\end{theorem}

\begin{proof}
We adapt the main idea from \cite{Freydenberger2015} and implement the following approach. 

We first define an algorithm constructing, for a word $w$ with $\len{w} = n$ and a $(k-1)$-tuple $\gaptuple$ of gap constraints, a non-deterministic finite automaton $A_{w,\gaptuple}$ that accepts exactly the gapped subsequences $p \in \subseqSet{\gaptuple}{w}$. Moreover, this automaton has exactly $\Psi_{\gaptuple}(w)[p]$ accepting paths labelled with the subsequence $p$ of $w$.

Then, we will use this algorithm for the input words of the equivalence with multiplicities problem and obtain two automata $A_{w,\gaptuple}$ and $A_{w',\gaptuple}$. 

Finally, we use the algorithm of \cite{Tzeng1992} to test whether $A_{w,\gaptuple}$ and $A_{w',\gaptuple}$ are path equivalent, i.\,e., for each word $p$, the number of accepting paths of $A_{w,\gaptuple}$ labelled with $p$ equals the number of accepting paths of $A_{w',\gaptuple}$ labelled with $p$. If this algorithm returns a positive answer, then we can conclude that $\parikhK{w}{\gaptuple} = \parikhK{w'}{\gaptuple}$. Otherwise, we conclude that $\parikhK{w}{\gaptuple} \neq \parikhK{w'}{\gaptuple}$.

It is not hard to see that this approach is correct, as soon as we explain further how to implement the first step. Thus, we will now describe the construction of the automaton $A_{w,\gaptuple}$, and then compute the overall complexity of our algorithm. 

Some notations first. Assume $\gaptuple = (C_1,\ldots, C_{k-1})$, where, for every $i \in [k-1]$, the constraint $C_i$ is given as a (black box) procedure $\checkP(i,u)$, which checks $u\in C_i$ in time $O(P(|u|))$ for some polynomial $P$. 

The set of states is defined by
\begin{align*}
	Q_w = \{(0,0)\} \cup \{(i,j) \mid 1 \leq i \leq n, 1 \leq j \leq k \} \cup \{(n+1, k+1)\}
\end{align*}

The state $(0,0)$ is the initial state, and $(n+1, k+1)$ represents an error state. The states $(i, j)$ with $j = k$ and $i \geq j$ are final.
As an intuition, $i$ represents the position in the word $w$, and $j$ represents the position in the subsequence $p$. 

The transitions of the automaton will be defined to reflect the gap constraints:

\begin{align*}
	\delta_w((i,j), a) =
	\begin{cases}
		\{(i', j + 1) \in Q_w \mid i' > i, w[i'] = a, \\
			\myquad[8] w[i+1..i'-1]\in C_j \} & \text{if this set is non-empty,}\\
		\{(n+1, k+1)\} & \text{otherwise.}\\
	\end{cases}
\end{align*}

In both cases the automata are accepting a string $p \in \Sigma^k$
if and only if there is an embedding $e$ satisfying the constraints
with $p \subseq_e{w}$.
The automaton $A_w$ is accepting the words $w[i_1], w[i_2], \ldots, w[i_{k'}]$
with $i_1 < i_2 < \ldots < i_{k}$
while satisfying the constraints
with $w[i_j + 1 .. i_{j+1} - 1] \in C_j$ for $1 \leq j \leq k-1$.
For this, we start in the initial state and follow the path of states
$$(0, 0), (i_1, 1), (i_2, 2), \ldots, (i_{k}, k)$$
with $ 1 \leq i_1 < i_2 < \ldots < i_{k}$ and $i_{k} \geq |w|$,
so the state $(i_{k'}, k')$ is indeed an accepting one.

Additionally,
for a word that is accepted by the automaton $A_w$ through the path
$$(0, 0), (i_1, 1), (i_2, 2), \ldots, (i_{k}, k) \text{,}$$
we get by definition that $i_j < i_{j+1}$ for all $1 \leq j \leq k - 1$,
and that $0 < i_1$.
Each transition ending in state $(i_j, j)$ is labelled with $w[i_j]$ and we have that $w[i_j + 1 .. i_{j+1} - 1] \in C_j$. 

Combined,
we immediately get that the automaton accepts exactly the gapped subsequences of length $k$
satisfying the given constraints,
and the number of distinct paths labelled with with the particular gapped subsequence
equals the number of its occurrences in $w$.

This concludes the description of the automaton $A_{w,\gaptuple}$. To completely prove the statement, we need to evaluate the overall complexity of the algorithm. For that, some more implementation details are needed. 

Let $N$ be the length of the longest of the input strings $w$ and $w'$. 

For each $j\in [k-1]$ and $u\in \{w,w'\}$, we can first identify (and store in a three dimensional array $M_u$) in $O(N^2 P(N))$ time all the factors of $u$ which are part of $C_j$ (that is, $M_u [a][b][j]=1$ if and only if $w[a..b]\in C_j$). Then, the construction of the automaton can be done in $O(N^2k )$ time (as we can also assume $|\Sigma|\leq N$). 

Now, to decide if the automata $A_{w,\gaptuple}$ and $A_{w',\gaptuple}$ are path equivalent, we use the algorithm described in  \cite{Tzeng1992}. The runtime of this part of our algorithm is $O((Nk)^4)$.

Overall, our algorithm runs in $O(N^2 k P(N) + N^4k^4)$. 
\end{proof}

In the case of length, regular or reg-len constraints, the equivalence problem with multiplicities can be solved in $O(\max\{|w|, |w'|\}^4 k^4 + \size(\gaptuple))$ time. The \emph{containment problem with multiplicities} (i.\,e., deciding $\parikhK{w}{\gaptuple}[p] \leq \parikhK{w'}{\gaptuple}[p]$ for all $p \in \Sigma^k$) seems to be more difficult. To our knowledge, whether the case of classical subsequences (i.\,e., length constraints $(0, \infty)$) can be solved in polynomial time is open. On the other hand, for the case of length constraints $(0, 0)$ only (i.\,e., consecutive factors), or of length constraints $(\ell, \ell)$ only (i.\,e., partial words), showing polynomial time solvability is relatively simple.

\end{document}